\RequirePackage[l2tabu, orthodox]{nag}
\documentclass[11pt]{amsart}

\usepackage{tikz} 
\usepackage[a4paper]{geometry}
\usepackage{amsmath, amsfonts, amssymb, amsthm, graphicx, color, booktabs, hyphenat, parskip, bbm, bm, enumerate}
\usepackage[font=small, format=plain, labelfont=bf, up, textfont=it, up]{caption}
\usepackage[normalem]{ulem}
\usepackage{float}   
\usepackage{microtype}
\usepackage[colorlinks=false]{hyperref}
\usepackage{cleveref}

\numberwithin{equation}{section}

\newtheorem{thm}{Theorem}[section]

\newtheorem{prop}{Proposition}[section]

\theoremstyle{definition}

\newtheorem{g-list}{Properties}[section]

\theoremstyle{remark}
\newtheorem{rem}{Remark}
\newtheoremstyle{RHP}  		
  {3pt}					
  {3pt}					
  {\itshape}				
  {}						
  {\bfseries}				
  {}						
  {.5em}					
  {Riemann-Hilbert Problem #2 #3} 

\theoremstyle{RHP}
\newtheorem{rhp}{}[section]

\newcommand{\R}{\mathbb{R}}
\newcommand{\C}{\mathbb{C}}
\newcommand{\N}{\mathbb{N}}
\newcommand{\Z}{\mathbb{Z}}
\newcommand{\Poly}{\mathbb{P}}
\newcommand{\one}{\mathbbm{1}}

\newcommand{\I}{\mathcal{I}}

\newcommand{\ie}{i.e.}

\newcommand{\cf}{cf.}
\newcommand{\etal}{et al.}

\newcommand{\U}{\mathcal{U}}
\newcommand{\Scal}{\mathcal{S}}
\newcommand{\Acal}{\mathcal{A}}
\newcommand{\Bcal}{\mathcal{B}}

\newcommand{\nin}{\not\in}
\newcommand{\lp}{\left(}
\newcommand{\rp}{\right)}
\newcommand{\lb}{\left[}
\newcommand{\rb}{\right]}
\newcommand{\bigo}[1]{\mathcal{O} \left( #1 \right)}

\newcommand{\pd}[3][ ]{\frac{\partial^{#1} #2}{\partial #3^{#1} } }
\newcommand{\od}[3][ ]{\frac{\mathrm{d}^{#1} #2}{\mathrm{d} #3^{#1} } }

\newcommand{\vect}[1]{\bm{#1}}

\newcommand{\sig}{ {\sigma_3} }

\newcommand{\triu}[2][1]{\begin{pmatrix} #1 & #2 \\ 0 & #1 \end{pmatrix}}
\newcommand{\tril}[2][1]{\begin{pmatrix} #1 & 0 \\ #2 & #1 \end{pmatrix}}
\newcommand{\diag}[2]{\begin{pmatrix} #1 & 0 \\ 0 & #2 \end{pmatrix}}
\newcommand{\offdiag}[2]{\begin{pmatrix} 0 & #1 \\ #2 & 0 \end{pmatrix}}

\newcommand{\eps}{\epsilon}

\newcommand{\lam}{\lambda}

\newcommand{\Abel}{\mathcal{A}}

\newcommand{\RR}{\mathcal{R}}
\newcommand{\SSS}{\mathcal{S}}

\newcommand{\Ecal}{\mathcal{E}}

\DeclareMathOperator{\re}{Re}
\DeclareMathOperator{\imag}{Im}

\DeclareMathOperator{\sgn}{sgn}

\DeclareMathOperator*{\res}{Res}

\DeclareMathOperator{\cn}{cn}
\DeclareMathOperator{\dn}{dn}
\DeclareMathOperator{\sd}{sd}
\DeclareMathOperator{\jacobiZ}{Z}

\title[sharp shock regularization]{Regularization of a sharp shock by the defocusing nonlinear Schr\"odinger equation}
\author{Robert Jenkins}
\address{SISSA, Via Bonomea 265, 34136 Trieste, Italy}
\email{\texttt{rjenkins@sissa.it}}

\date{\today}

\begin{document}

\maketitle

\begin{abstract}
The defocusing nonlinear Schr\"odinger (NLS) equation is studied for a family of step-like initial data with piecewise constant amplitude and phase velocity with a single jump discontinuity at the origin. Riemann-Hilbert and steepest descent techniques are used to study the long time/zero-dispersion limit of the solution to NLS associated to this family of initial data. We show that the initial discontinuity is regularized in the long time/zero-dispersion limit by the emergence of five distinct regions in the $(x,t)$ half-plane. These are left, right, and central plane waves separated by a rarefaction wave on the left and a slowly modulated elliptic wave on the right. Rigorous derivations of the leading order asymptotic behavior and error bounds are presented.
\end{abstract}

\section{Introduction}
In this paper we study the defocusing nonlinear Schr\"odinger equation (NLS), given here with the normalization 
\begin{equation}\label{0.1}
	i \eps \psi_t + \frac{\eps^2}{2} \psi_{xx} - |\psi^2| \psi = 0,
\end{equation}
for a fixed class of piecewise constant, steplike, initial data (\cf \eqref{0.2}).
The NLS equation is a canonical model of dispersive wave dynamics, and has been shown to be an excellent model for a wide variety of disparate physical systems, including water waves \cite{Lake77}; plasmas \cite{Taylor70}, \cite{Zak}; nonlinear optics \cite{Ania06}; and Bose-Einstein condensates \cite{Hoeffer06}. Of particular interest is the case in which the dispersion parameter $\eps \ll 1$, which is the natural scaling in both BECs and nonlinear optics\cite{Hoeffer06}, \cite{Kodama99}. The NLS equation is also of intrinsic mathematical interest as one of the principal examples of a completely integrable nonlinear evolution equation.

The zero dispersion limit, \ie\  $\eps \to 0$, of the NLS equation \eqref{0.1} is better understood by introducing the Madelung variables \cite{Madelung},
\begin{equation}\label{fluid vars}
	\rho(x,t) = | \psi(x,t) |^2 \qquad 
	u(x,t) =\eps \imag \lb \partial_x \log \lp  \psi(x,t) \rp \rb
\end{equation}
which transforming the NLS equation into the system of conservation laws
\begin{subequations}\label{conservation}
\begin{align}
\label{mass}	&\pd{\rho}{t} + \pd{}{x}(\rho u) = 0, \\
\label{momentum}	&\pd{}{t}(\rho u) + \pd{}{x} \lp \rho u^2 + \frac{1}{2} \rho^2 \rp 
		= \frac{\eps^2}{4}\pd{}{x} \lp \rho \pd[2]{}{x} \lp \log \rho \rp \rp.
\end{align}
\end{subequations}
When $\eps=0$ these are the Euler equations for an ideal compressible fluid (gas) with local fluid density $\rho$, velocity $u$, and positive pressure $P = \frac{1}{2} \rho^2$. 
It is well known that the Euler system admits solutions which develop gradient catastrophes (infinite derivatives) in finite time. 
However, for $\eps>0$, as the wave steepens the right hand side of the momentum conservation law \eqref{momentum} cannot be treated as a perturbative term and shock formation is avoided by the emergence of expanding regions of $(i)$ rarefaction waves and/or $(ii)$ the onset of slowly modulating $\bigo{\eps}$ wavelength oscillations with $\bigo{1}$ amplitude known as \textit{dispersive} (sometimes \textit{collisionless}) \textit{shock waves} (DSWs). Clearly, when DSWs emerge, a zero dispersion limit cannot exist in the classical sense. Nevertheless, a weak limit does exist for NLS as was shown in \cite{JLM} following the work of \cite{LL83}, \cite{Venakides90} on the Korteweg de Vries (KdV) equation. This weak limit can be understood in terms of the unique minimizer of a certain minimization problem with constraints. The minimizer itself is characterized by its support, which typically is a union of disjoint intervals. The endpoints of these intervals satisfy a system of quasilinear hyperbolic equations
\begin{equation}\label{intro Whitham equations}
	\pd{\lam_j}{t} + v_j(\vect \lam) \pd{ \lam_j}{x} = 0, \quad j = 1,2,\dots, 2G+2  
\end{equation}
where $\vect \lam \in \R^{2G+2}$ and $\lambda_j > \lambda_k$ for $j<k$, called the Whitham equations after their first discoverer \cite{Whit}.

The dynamics of the DSWs themselves can be described as slowly modulating single or multiphase waves, whose modulations are also governed by the Whitham equations \cite{GP87, GKE92}. The modulation theory was worked out for $G=1$ in \cite{Whitham} and for $G \geq 2$ in \cite{FFM} in the context of KdV. The Whitham modulation theory for NLS was worked out in \cite{FL}. In the years following, Whitham theory has been used in the optics and fluid dynamic communities to investigate increasingly complicated structures: including the initial data problem for piecewise constant data (the type considered in this paper) \cite{Kodama99}, \cite{BK06}; the interaction of two DSWs \cite{AH07}; and in \cite{El} a classification of the types of solutions of the Whitham-NLS system for initial data with a discontinuity of the form \eqref{0.2} was given, to name but a few. 

At the same time, the development of the inverse scattering technique for studying completely integrable nonlinear evolution equations has resulted in a huge amount of work on the NLS equation. In particular, the 
nonlinear steepest descent method of Deift and Zhou \cite{DZ93} \cite{DZ94} allows one to make completely rigorous arguments to obtain, in principal, full asymptotic expansions of the solutions of integrable systems in various asymptotic limits. The bulk of the work being done in the integrable systems community has focused on rapidly decreasing initial data $\psi_0(x)$, which decays to zero sufficiently fast as $|x| \to \infty$, \cite{DVZ98}, \cite{DiFM05}, \cite{KMM03}, \cite{TVZ04}, \cite{CG10}, \cite{BT10}, \cite{JM11},. Comparatively, much less time has been devoted to families of non-vanishing initial data. The family of so called finite density initial data $\psi_0(x)$ satisfying $\psi_0(x) \to \varrho e^{i \varphi_\pm}$ as $x \to \pm \infty$ for constants $\varrho > 0$ and $\varphi_\pm \in [0, 2\pi)$ is probably the best understood of these non-vanishing families. As was shown in \cite{BP82, CK85, FT07, DP13}, the scattering theory for non-vanishing data must be constructed on multi-sheet Riemann surfaces, a complication which is not necessary for vanishing data. Nevertheless, results for long time asymptotics for \eqref{0.1} with finite density data were worked out first by Its \etal \ in \cite{Its1}, \cite{Its2} and recently Vartanian \cite{Vartanian1} \cite{Vartanian2} has found very detailed asymptotic formulae for the long time asymptotic behavior of finite density data with and without the presence of (dark) solitons. Another family of nonvanishing data are ``step-like" initial data 
which asymptotically approaches different plane wave states as $x$ approaches either infinity, both in the context of the NLS equation \cite{BKS11, BP82, BV07, CX13} and other important integrable evolution equations \cite{EGK13, KA12, BE00}.  

In this paper it is our goal to make a completely rigorous study of the long-time/zero-dispersion behavior (for the scale invariant data \eqref{0.2} we consider these are the same limit as we will make clear later) of the solution of the NLS equation \eqref{0.1} for the family of sharp step initial data 
\begin{equation}\label{0.2}
	\psi(x,t=0) = \psi_0(x) := \begin{cases}
		1  & x < 0 \\
		A \exp \lp -2 i \mu x/ \eps \rp & x > 0. 
	\end{cases}
\end{equation}
for real constants $A>0$ and $\mu$ using the machinery of inverse scattering and nonlinear steepest descent. 

The hyperbolic nature of the NLS equation suggests that for large $x$ the solution for initial data \eqref{0.2} should resemble a plane wave (zero phase oscillation) with Riemann invariants $\lam_{1,2} = -u/2 \pm \sqrt{\rho}$ (cf. Section~\ref{sec:zero phase}) whose values as $x\to -\infty$ approach $\pm 1$ and approach  
\begin{equation}\label{invariants1}
\lambda_\pm = \mu \pm A
\end{equation}
as $x \to +\infty$. In \cite{El}, using Whitham theory, the authors enumerate six possible long-time behaviors for the data \eqref{0.2} depending upon the relative ordering of these constants $\{-1,1, \lam_-, \lam_+\}$. In each case the discontinuity is regularized by the emergence of two zones in which either DSWs or fan like rarefactions connect three constant states, see Figure~\ref{fig:phasediagram}. 

Our results, which follow below, provide a completely rigorous proof that the leading order asymptotic behavior of the density $\rho$ and velocity $u$ are as predicted by the Whitham theory, and gives bounds on the error. Moreover, our methods provide a superior descriptions of the solution $\psi(x,t)$ as we are able to compute the leading order phase of the solution $\psi(x,t)$. This include terms which are lost in the Whitham averaging process, but nonetheless make $\bigo{1}$ contributions to the solution $\psi$ of \eqref{0.1}. Our paper provides all the tools necessary to easily deal with all six cases identifies in \cite{El}. However, for the sake of brevity, we will provide full details for only one case: $-1 < \lam_- < \lam_+ < 1$ (case $i.$ in \cite{El}), in which both a DSW and rarefaction waves emerge, see Figure~\ref{fig:phasediagram}. 

\begin{figure}[htb]
	\begin{center}
		\includegraphics[width=.7\textwidth]{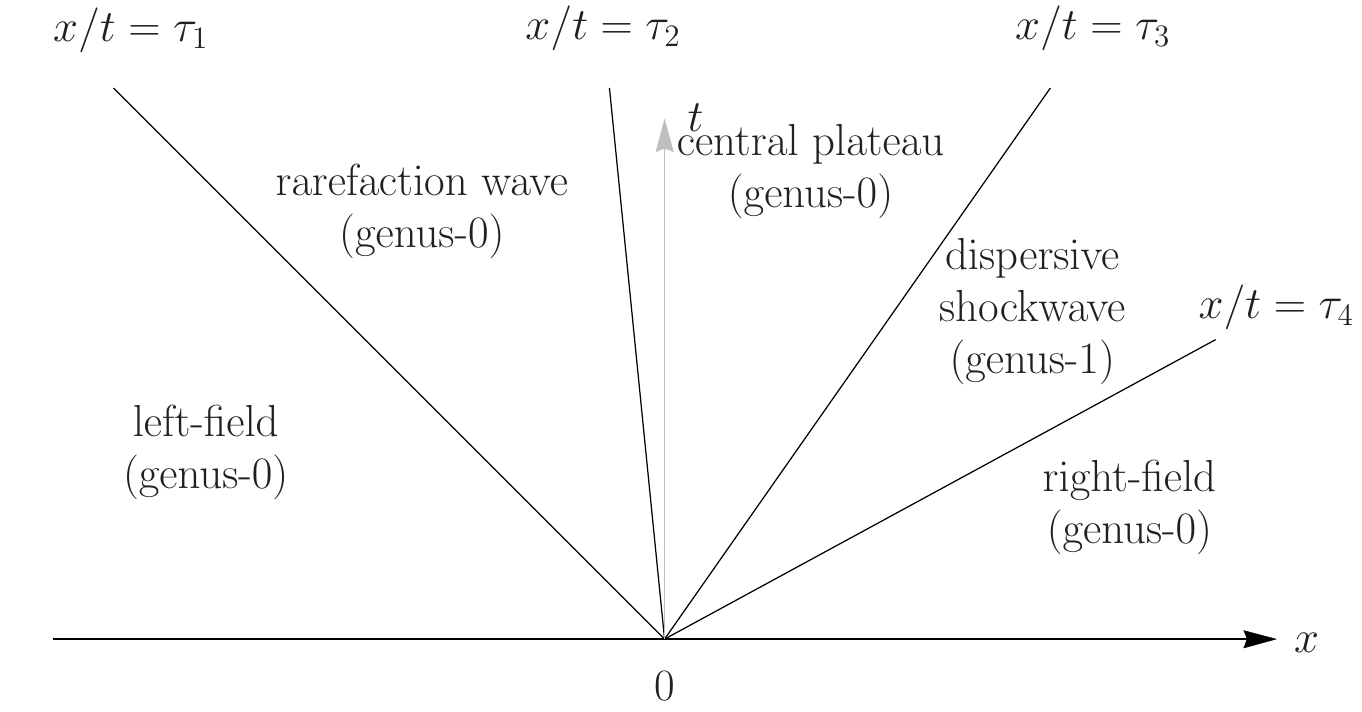}
		\caption{
		The leading order asymptotic behavior of the solution of \eqref{0.1} with initial data \eqref{0.2} in the five different regions of the $(x,t)$-half-plane in the situation when the right Riemann invariants are inside the left invariants: $-1 < \lam_- < \lam_+ < 1$. The transition speeds $\tau_j$, for this ordering of the invariants, are given in Theorem~\ref{thm:main}. 
		\label{fig:phasediagram}
		}
	\end{center}
\end{figure}

In order to compute the phase of the solution we need the reflection coefficient which is part of the scattering data computed in the inverse scattering procedure. For the initial data \eqref{0.2}, and $\lam_\pm$ defined by \eqref{invariants1}, the reflection coefficient generated from \eqref{0.2} is 
$$
	r(z) = -i \frac{\sqrt{z-1}\sqrt{z-\lam_-} - \sqrt{z-\lam_+}\sqrt{z+1}}
	{\sqrt{z-1}\sqrt{z-\lam_-} + \sqrt{z-\lam_+}\sqrt{z+1}},
$$
where each of the roots is principally branched. When $(\lam_-, \lam_+) \subset (-1,1)$, which is the setting or our result, it is easy to check that $r(z)$ is branched on $ (-1,1) \backslash (\lam_-, \lam_+)$, with unit modulus on either side of the branch, \ie,  $|r(z\pm i 0)| = 1$ for $z \in (-1,1) \backslash (\lam_-, \lam_+)$, and $r(z) \sim z^{-1}$ as $z\to \infty$. 
 
\begin{figure}[htb]
\begin{center}
\includegraphics[width=.4\textwidth]{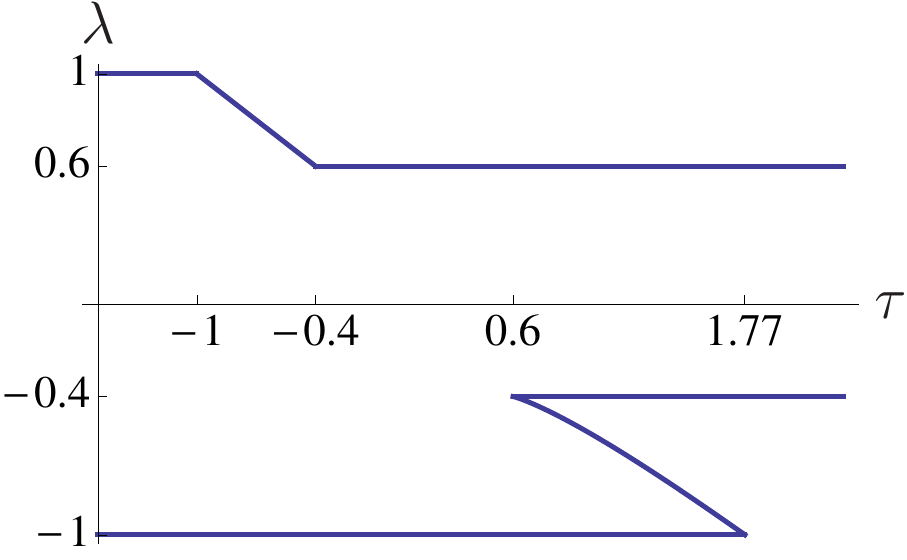}  
\caption{The self-similar evolution of the Riemann invariants $\lam_i(\tau)$ with respect to the similarity variable $\tau = x/t$. In the figures the constants were taken as $\eps=0.001$, $A =0.5$, $\mu = 0.1$ (\ie\ the Riemann invariants of the right side $\lambda_\pm = -u/2 \pm \sqrt{\rho}$ are $-0.4$ and $0.6$ respectively, which lie between the left invariants $\pm 1$).
\label{fig:invariants}
}
\includegraphics[width=.4\textwidth]{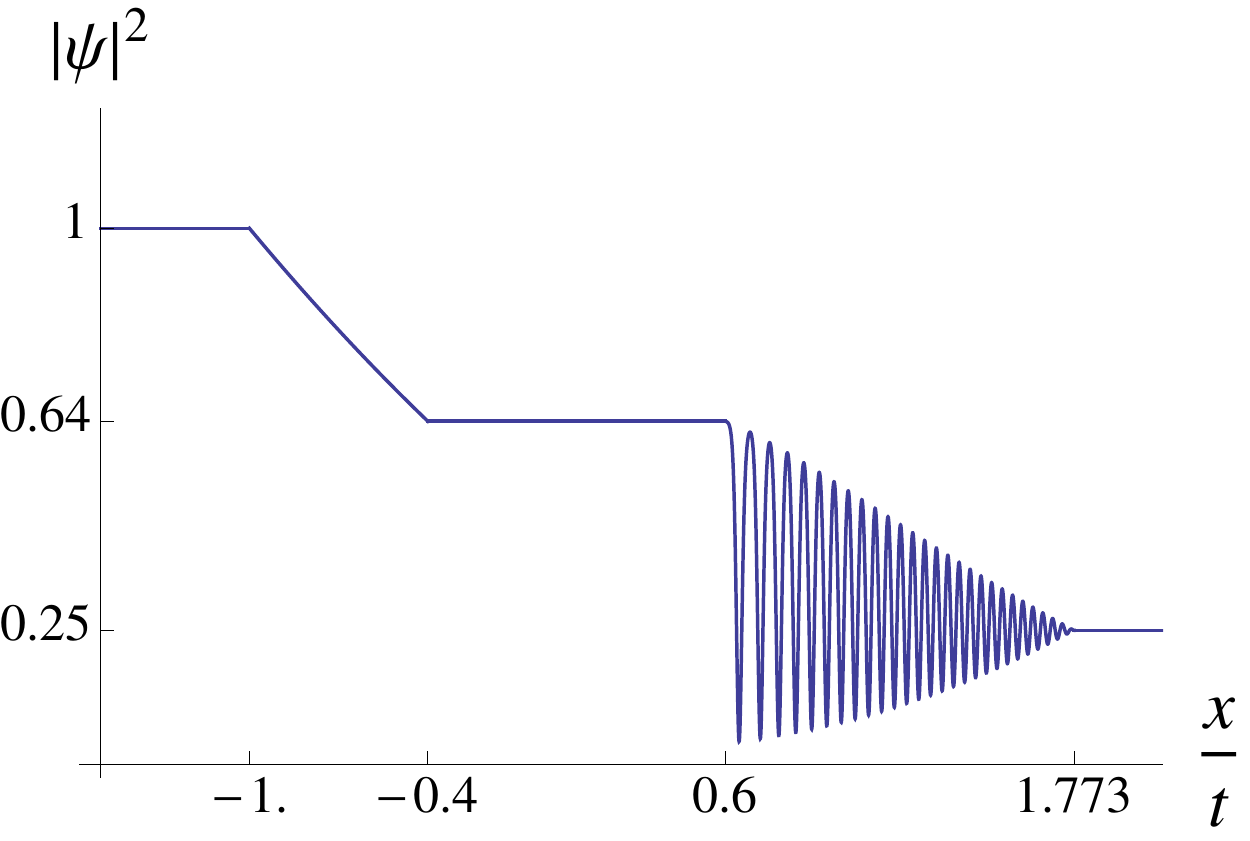}
\hspace{.015\textwidth}
\includegraphics[width=.4\textwidth]{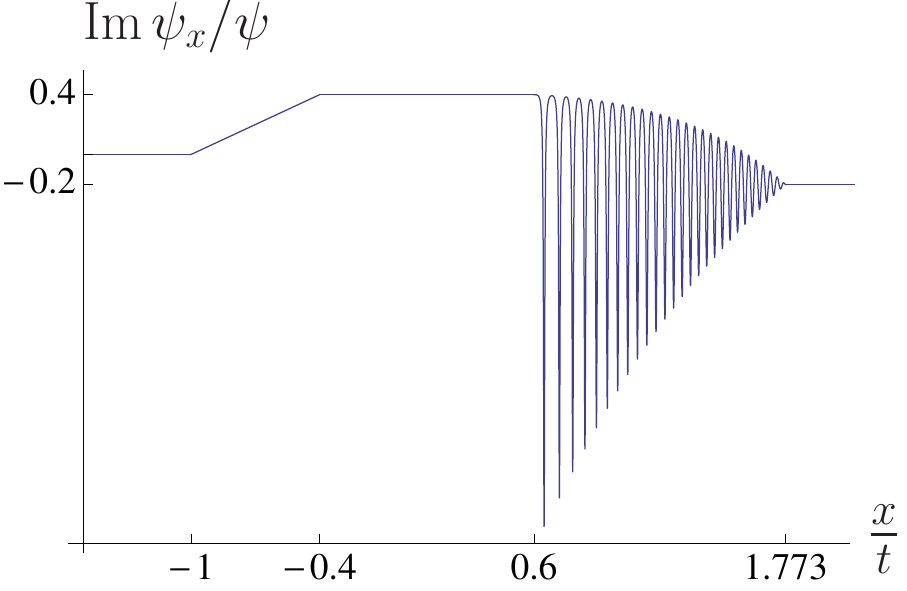}
\caption{\emph{Left:} The leading order asymptotic behavior of the density $\rho = |\psi|^2$ and \emph{Right:} the leading order asymptotic behavior of the velocity $u = \imag \psi_x/\psi$ 
related to the hydrodynamic interpretation \eqref{conservation} of the solution $\psi$ of NLS \eqref{0.1} 
in the small-dispersion/long-time limit for pure-step initial data \eqref{0.2}. 
The parameters used to generate the figures are the same used in Figure~\ref{fig:invariants}.
The initial discontinuity at the origin smooths itself by the emergence of a rarefaction zone on the left, and a modulated elliptic wave front (a DSW) on the right connected by a constant central plateau. 
\label{fig:density}
}
\end{center}
\end{figure}  
 
\begin{thm}\label{thm:main}
Given initial data \eqref{0.2}, if the Riemann invariants $\lambda_\pm = \mu \pm A$ satisfy $-1 < \lambda_- < \lambda_+ < 1$, then the long-time/small-dispersion asymptotic behavior of the solution $\psi(x,t)$ of the NLS equation \eqref{0.1} is given by one of the five following formulae depending on the value of the similarity variable $\tau =x/t$ relative to the transition speeds $\tau_j$ identified as:
\begin{samepage}
\begin{enumerate}[1.]
\item $\tau_1 =-1$ 
\item $\tau_2 = -\frac{1}{2} \lp -1 + 3\lam_+ \rp$ 
\item $\tau_3 =-\frac{1}{2} \lp -1 +  2 \lambda_- +\lambda_+ \rp $ 
\item $\tau_4 =-\frac{1}{2} \lp \lambda_+  + \lambda_- - 2 \rp + 
	\frac{2(1+\lambda_-)(1+\lambda_+)} {\lambda_+ + \lambda_- + 2} $ 
\end{enumerate}
\end{samepage}

\ \\
\begin{enumerate}[1.]
\item For $\tau < \tau_1$, the solution is asymptotically a plane wave with constant amplitude. 
\begin{gather}
	\psi(x,t) = e^{-i t /\eps} e^{-i \phi(x/t) } + \bigo{ \sqrt{ \frac{\eps}{t}}\log \frac{\eps}{t} } \\
	\phi(\tau) = \frac{1}{\pi}  \Bigg(
	 \int\limits_{-\infty}^{-1} + \int\limits_1^{\xi_+(\tau)} \Bigg) \frac{ \log( 1 - |r(z)|^2)}{\sqrt{z^2-1}} dz
	 + \frac{1}{\pi} \Bigg( \int\limits_{-1}^{\lambda_-} + \int\limits_{\lambda_+}^{1} \Bigg) 
	 \frac{  \arg( r_+(z)) }{\sqrt{1-z^2}} d z \nonumber \\
	 \xi_+(\tau) = \frac{1}{4} \lb \sqrt{\tau^2 +8} - \tau \rb \nonumber
\end{gather}
\item For $\tau_1 < \tau < \tau_2$, the solution is described by the rarefaction 
\begin{gather}
	\psi(x,t) = \lp \frac{2t-x}{3t} \rp  e^{(-i/3\eps)(2 t - 2 x  - x^2/t)}  e^{ - i \phi(x/t) } + \bigo{ \lp \frac{\eps}{t} \rp^{2/3} } \\
	\nonumber 
	\phi(\tau) = 
	\frac{1}{\pi} \lp \int_{-\infty}^{-1}  
	\frac{\log(1 -|r(z)|^2)}{\sqrt{(z+1)(z-\lam_s(\tau))}} dz +
	\Bigg( \int_{-1}^{\lambda_-} + \int_{\lambda_+}^{\lambda_s(\tau)} \Bigg) 
	\frac{ \arg( r_+(z)) }{\sqrt{(z+1)(\lam_s(\tau)-z)}} dz \rp \\
	\lam_s(\tau) = \frac{(1-2\tau)}{3} \nonumber
\end{gather}
\item For $\tau_2 < \tau < \tau_3$ the solution is asymptotically described by the (unmodulated) plane wave
\begin{gather}
	\psi(x,t) = \sqrt{\rho}  e^{i (k x - \omega t)/\eps}  e^{ - i \phi_0 } + \bigo{ e^{-ct/\eps} } \\
	\nonumber
	\rho = \lp \frac{\lambda_+ + 1}{2}\rp^2  \qquad k =-\lp  \lambda_+ -1 \rp 
	\qquad \omega = \frac{1}{2} k^2 + \rho \\
	\nonumber
	\phi_0 =  \frac{1}{\pi} \lp \int_{-\infty}^{-1}  
	\frac{\log(1 -|r(z)|^2 )}{\sqrt{(z+1)(z-\lambda_+)}} dz +
	\int_{-1}^{\lambda_-} \frac{ \arg( r_+(z)) }{\sqrt{(\lambda_+ - z)(z+1)}} dz \rp
\end{gather}
\item For $\tau_3 < \tau < \tau_4$ the asymptotic behavior of the solution is described by a slowly modulated one-phase (elliptic) wave, a dispersive shock wave, 
$$
	\psi(x,t) = \sqrt{\rho(x,t)} e^{i S(x,t)} + \bigo{ \lp\frac{\eps}{t}\rp^{2/3}}, 
$$
whose amplitude and phase are given by 
\begin{equation}
	\begin{gathered}
	\rho(x,t) = 
	a_1^2 - (a_1^2 -a_3^2) \dn^2 \lp \sqrt{a_1^2 - a_3^2} \lp \frac{x-Vt}{\eps} +\phi \rp -K(m), m \rp \\
	\begin{aligned}
	S(x,t) =&  \frac{Vx -  ( a_1^2 + a_2^2 + a_3^2-V^2)t  + (x-2Vt)\eta }{\eps} 
		+ 2 (V+\eta)\phi \lp \frac{x}{t} \rp \\
	+& \arg \left\{ \theta_3 \lb \frac{\pi}{2K(m)} \sqrt{a_1^2-a_3^2} \lp \frac{x-Vt}{\eps} + \phi \lp \frac{x}{t} \rp \rp
		-i\pi \frac{F(\varphi, 1-m)}{K(m)}\rb \right\} 
	\end{aligned} \\
	u(x,t) = \eps \pd{S}{x} = \frac{a_1 a_2 a_3}{\rho(x,t)} + V + \bigo{\eps}
	\end{gathered}
\end{equation}	
	The parameters $a_1$, $a_2$, $a_3$, $V$, and the elliptic modulus $m$ are rational functions of 
	the Riemann invariants $\vect \lam = (\lambda_1, \lambda_2, \lambda_3, \lambda_4) = (\lam_+, \lam_-, \lam_s(x/t), -1)$ given by \eqref{one phase parameters} where $\lam_s(\tau)$ satisfies the self-similar system of genus one Whitham equations \eqref{3.26}. The other parameters are
	\begin{equation}
	\begin{gathered}
	\phi(\tau) = \frac{1}{2\pi} \int_{-\infty}^{-1} \frac{ \lp z + V \rp \log( 1 - |r(z)|^2) }
	{\sqrt{\prod_{j=1}^4 (z- \lam_j)}} dz 
	 - \frac{1}{2\pi} \int_{-1}^{\lam_s(\tau)} \frac{ \lp z + V \rp \arg r_+(z)}
	 {\sqrt{-\prod_{j=1}^4 (z- \lam_j)}} dz, \\
	 \eta = \lam_1 - (\lam_1 - \lam_4) Z \lp -\frac{\lam_3-\lam_4}{\lam_1-\lam_4}, n \rp, \\
	 n = - \lp \frac{\lam_3-\lam_4}{\lam_1-\lam_4}\rp,  \qquad 
	 \varphi = \arcsin \sqrt{\frac{\lam_2-\lam_4}{\lam_1-\lam_4}}.
	\end{gathered}
	\end{equation}
	 Here $K$ and $E$ are the complete elliptic integrals of the first and second kind respectively, $F$ is the incomplete elliptic integral of the first kind, and $Z$ is the Jacobi zeta function.

\item For $\tau > \tau_4$ the leading order behavior of the solution is given by a plane wave which, up to the phase $e^{-i\chi(x/t)}$, is the time-evolution of the right half of the initial data:
\begin{gather}
	\psi(x,t) = A e^{ -i( 2\mu x + (A^2+2\mu^2) t)/\eps}  e^{ - i \phi(x/t) } 
	+ \bigo{ \sqrt{\frac{\eps}{t}} \log \frac{\eps}{t}  } 
	\\ \nonumber
	\phi(\tau) = 
	\frac{1}{\pi}  \int_{-\infty}^{\xi_-(\tau)}  \frac{\log(1 -|r(z)|^2 )}{\sqrt{(z-\lambda_+)(z-\lambda_-)}} dz, 
	\\ \nonumber
	\xi_-(\tau) = \frac{2\mu - \tau}{4} - \frac{\sqrt{(2\mu+\tau)^2+8A^2}}{4}.   
\end{gather}
\end{enumerate}	
\end{thm}	

\begin{rem}
 The convergence of the solution $\psi(x,t)$  as $\eps \to 0$ to the given leading order formulae is uniform in any sector $\{ (x,t) \in \R \times (T, \infty) \, :\,  \frac{x}{t} \in [a,b] \}$ which avoids the transition speeds, \ie\ $\tau_j \nin [a,b]$, $j=1,2,3,4$. Moreover, though perhaps not immediately obvious from the formulae, the leading order behavior is continuous across each of the four transitions as can be checked by hand, or as seen in Figure~\ref{fig:density}. 
\end{rem}  

\begin{rem}
The leading order hydrodynamic density $\rho = |\psi(x,t)|^2$ and velocity $u(x,t) = \eps \imag \lb \partial_x \log \psi(x,t) \rb$ computed from the formulae in Theorem~\ref{thm:main} agree with the results predicted by Whitham theory techniques in \cite{El}. The new contribution of this paper is the computation of the complex phase of $\psi(x,t)$ and the explicit bounds on the error. Specifically, the slowly evolving phase term $\phi$ in each of the five formulae is new and does not appear in the Whitham theory as it constitutes a perturbative term in the computation of the velocity $u$ but nonetheless contributes an $\bigo{1}$ correction to the complex phase of the solution $\psi(x,t)$. 
\end{rem}

\begin{rem}
Though we consider only the case $-1 < \lam_- < \lam_+ < 1$ the other five possible cases (\ie\ orderings of $-1,1,\lam_-$, and $\lam_+$) regularize the initial discontinuity in a similar way, and we provide all the necessary tools to complete these computations. In each case, five sectors emerge in the $(x,t)$ half-plane as in Figure~\ref{fig:phasediagram}: the far left and right fields exhibit plane wave (genus zero) oscillations which match the initial data for $t=0$, while the three middle zones consist of either rarefaction and/or dispersive shock waves separated by a central plateau that is either a plane wave or, when $1 < \lam_- <\lam_+$, a standing (unmodulated) elliptic wave. 
\end{rem}

\begin{rem}
The choice to normalize the left half of the initial data \eqref{0.2} to have $\rho=1$ and $\mu =0$ is not a restriction, any sharp step of the form 
\begin{equation*}
	\widetilde{\psi}(x,0) = \widetilde{\psi}_0(x) := \begin{cases}
		A_L \exp \lp -2 i \mu_L x/ \eps \rp & x < 0 \\
		A_R \exp \lp -2 i \mu_R x/ \eps \rp & x > 0. 
	\end{cases}
\end{equation*}
with $A_L$ and $A_R$ not both zero can be reduce to our normalized data; in the case that $A_L \neq 0$, the change of variables 
$$
	\widetilde \psi(x,t) = A_L \psi(A_L(x-2\mu_L t), A_L^2 t) e^{-2i\mu_L(x+\mu t)/\eps}
$$
results in a new unknown $\psi(x',t')$ solves \eqref{0.1} with initial data \eqref{0.2} in the new coordinate frame.	
\end{rem}

\subsection{Organization of the rest of the paper}
In Section~\ref{sec:fluid} we briefly review the NLS-Whitham equations for zero and one phase waves and discuss their self-similar solutions. In Section~\ref{sec:forward scattering} we discuss the integrable structure of the NLS equation, compute the scattering data for the step initial data \eqref{0.2}, and state the Riemann-Hilbert problem satisfied by the solution of \eqref{0.1}-\eqref{0.2} in full detail. In Section~\ref{sec:gfunctions} we construct the so called $g$-functions that are needed in the inverse scattering analysis and show that their evolution is governed by the NLS-Whitham equations. Finally in Section~\ref{sec:inverse} we use the Deift-Zhou steepest descent procedure to derive the asymptotic behavior to the solution of Riemann-Hilbert problem \ref{rhp:1} for every real value of $\tau =x/t$, which proves the results of Theorem~\ref{thm:main}.

Before proceeding we comment on notation. Throughout the paper we make use of the Pauli matrices
$$
	\sigma_1 = \offdiag{1}{1}, \quad \sigma_2 = \offdiag{i}{-i}, \quad \sigma_3 = \diag{1}{-1}.
$$
In particular we use the matrix power notation $f^\sig = \lp \begin{smallmatrix} f & 0 \\ 0 & f^{-1} \end{smallmatrix} \rp$ for any scalar $f$. 

Regarding complex variable notation, $z^*$ denotes the complex conjugate of a complex number $z$; for a scalar function $f$, $f^*(z)$, or compactly just $f^*$, denotes the Schwarz reflection through the real axis $f^*(z) = f(z^*)^*$. Given a piecewise smooth oriented contour $\gamma \in \C$ and a function $f$ analytic in $\C \backslash \gamma$, for $z \in \gamma$, $f_\pm(z)$ is defined as the non-tangential limit of $f(w)$ as $w$ approaches 
$z$ from the left/right with respect to the orientation of $\gamma$. 
Finally, given a pair of real numbers $a,b$ or a vector $\vect \lam \in \R^4$ we define
$$
	\RR(z;a,b) = \sqrt{(z-a)(z-b)} \qquad \RR(z; \vect \lam) = \sqrt{\prod_{j=1}^4 (z-\lam_j)} 
$$
to be finitely branched along the real axis such that $\RR(z;a,b) \sim z$ and $\RR(z;\vect \lam) \sim z^2$ as $z\to \infty$.

%

\section{Hydrodynamic form and modulation theory}\label{sec:fluid}
The Madelung change of variables \eqref{fluid vars} transforms the NLS equation into the system of conservation laws
\begin{subequations}\label{conservation2}
\begin{align}
\label{mass2}	&\pd{\rho}{t} + \pd{}{x}(\rho u) = 0 \\
\label{momentum2}	&\pd{}{t}(\rho u) + \pd{}{x} \lp \rho u^2 + \frac{1}{2} \rho^2 \rp 
		= \frac{\eps^2}{4}\pd{}{x} \lp \rho \pd[2]{}{x} \lp \log \rho \rp \rp,
\end{align}
\end{subequations}
If $\eps$ is formally set to zero, then it is well known that the resulting Euler system exhibits shock formation (infinite gradients) in finite time. For $\eps >0$ the right hand side of \eqref{momentum2} ameliorates the formation of shocks by introducing growing regions of rapid oscillations into the solution. These rapid oscillations are well approximated in terms of slowly modulating one-phase waves, whose modulations satisfy Whitham's averaging equations \cite{Whitham}, \cite{FFM}, \cite{FL}. For general initial data, over the course of the evolution, the number of phases need to describe the wave may change as Riemann invariants are born or merge, though for long times the system will exhibit only single phase oscillations \cite{GT02}.  
For the single shock initial data we consider here \eqref{0.2}, we will see that only elliptic (one phase) oscillations develop. We summarize below the Whitham equations for zero and one phase oscillations only. 

\subsection{Zero-phase oscillations\label{sec:zero phase}}
Before wave breaking occurs, the solution of \eqref{0.1} has bounded derivatives, and the limiting Euler equations for $\rho$ and $u$ should well-approximate the solution. That is, our solution is well described by the slowly modulating periodic wave
\begin{equation}
	\psi_0(x,t) = \sqrt{\omega_0 - k_0^2} e^{i \theta_0}, \quad
	\partial_x \theta_0 = k_0/\eps \quad \partial_t \theta_0 = - \omega_0/\eps
\end{equation}
whose density and velocity
\begin{equation}
	\rho(x,t) = | \omega_0 - k_0^2| \qquad
	u(x,t) = k_0
\end{equation}
satisfy the Euler equations (\eqref{conservation} with $\eps = 0$). The Euler equations can be written in the Riemann invariant form
\begin{equation}\label{0-phase riemann invariants}
	\begin{gathered}
	\pd{\lambda_j}{t} + v_j(\vect \lambda) \pd{\lambda_j}{x} = 0, \qquad j=1,2, \\
	\lambda_1 = -\frac{u}{2} + \sqrt{\rho}, \qquad \lambda_2 = -\frac{u}{2} - \sqrt{\rho} \\
	v_+(\vect \lambda) = -\frac{1}{2} \lp 3\lambda_1 + \lambda_2 \rp \qquad
	v_-(\vect \lambda) = -\frac{1}{2} \lp \lambda_1 + 3 \lambda_2 \rp.
	\end{gathered}
\end{equation}

\subsection{One-phase oscillations}
If instead we suppose that the solution exhibits a single fast phase, then the density $\rho$ and velocity $u$ are instead described asymptotically in terms of a modulating one-phase (elliptic) waves described in terms of four slowly varying Riemann invariants $\vect{\lambda} = \{\lambda_i \}_{i=1}^4$, $\lambda_1 > \lambda_2 > \lambda_3 > \lambda_4$: 
\begin{equation}\label{one phase solution}
	\begin{gathered}
		\rho(x,t; \vect \lambda) = a_1^2 - (a_1^2-a_3^2) 
		\dn^2 \lp \sqrt{a_1^2-a_3^2}\frac{x-V t}{\eps}, m \rp \\
		u(x,t; \vect \lambda) = V - \frac{a_1 a_2 a_3}{\rho(x,t,\vect \lambda)}
	\end{gathered}
\end{equation}	
\begin{equation}\label{one phase parameters}
	\begin{aligned}
	a_1 &= -\frac{1}{2}\lp \lambda_1 + \lambda_2 - \lambda_3 - \lambda_4 \rp \\
	a_2 &= -\frac{1}{2}\lp \lambda_1 - \lambda_2 + \lambda_3 - \lambda_4 \rp \\
	a_3 &= -\frac{1}{2}\lp \lambda_1 - \lambda_2 - \lambda_3 + \lambda_4 \rp \\
	V & = -\frac{1}{2} \lp \lambda_1 + \lambda_2 + \lambda_3 + \lambda_4 \rp \\
	m &= \frac{(\lambda_1 - \lambda_2)(\lambda_3-\lambda_4)}
	{(\lambda_1 - \lambda_3)(\lambda_2- \lambda_4)} 
	\end{aligned}
\end{equation}
The evolution of the Riemann invariants $\lam_i$ is governed by the diagonal first order system: 
\begin{equation}\label{NLS-Whitham}
	\begin{gathered}
		\pd{\lambda_i}{t} + v_i(\vect \lambda) \pd{\lambda_i}{x} = 0, \\
		v_{j}(\vect \lambda) = V(\vect \lambda) + \lp 2 \pd{}{\lambda_j} \log L(\vect \lambda) \rp^{-1} ,\\
	L(\vect \lambda) = \sqrt{2} \int_{\lambda_2}^{\lambda_1} \frac{ d \tau} { \sqrt{-\prod_{j=1}^4 (\tau - \lambda_j)} } 
	= \frac{ 2\sqrt{2} K(m) }{\sqrt{(\lambda_1 - \lambda_3)(\lambda_2 - \lambda_4)}},
	\end{gathered}
\end{equation}	
which can be obtained by averaging the first four conservation laws for NLS over a period of \eqref{one phase solution}.

\subsection{Self-similar evolution}
If we suppose that the Riemann invariants $\lam_j$ depend on $(x,t)$ only through a similarity variable $\tau = x/t$, then the Whitham equations \eqref{intro Whitham equations} are equivalent to 
\begin{equation}\label{self-similar Whitham system}
	 \lp v_j(\vect \lam) -\tau \rp  \pd{\lam_j}{\tau} = 0, \quad j = 1,2,\dots, 2G+2.  
\end{equation}
So that each $\lam_j$ is either constant or its speed satisfies $v_j(\vect \lam) = \tau$. Moreover, since the NLS-Whitham system is strictly hyperbolic \cite{BK06}, \ie, $v_j(\vect \lam) < v_k(\vect \lam)$ for $j<k$ provided $\lam_1 > \lam_2 >\dots> \lam_{2G+2}$, it follows that at most one of the speeds can satisfy $v_j(\vect \lam) = \tau$, and therefore in a self-similar evolution at most one of the Riemann invariants is not constant.

\section{Scattering of the shock initial data}\label{sec:forward scattering}

It is well known that NLS is completely integrable \cite{ZS71} in the sense that it is equivalent to the existence of simultaneous solution $\phi(x,t)$ of the Lax pair
\begin{subequations}\label{0.3}
	\begin{align}
\label{0.3a}		\eps \Phi_x &=  -iz \sig \Phi + \Psi(x,t) \Phi, \\
\label{0.3b}		i \eps \Phi_t &= z^2 \sig \Phi + iz \Psi(x,t) \Phi + \frac{1}{2} \lp \Psi(x,t)^2 + \eps \Psi_x(x,t) \rp \sig \Phi.
	\end{align}
\end{subequations}
given the matrix potential 
\begin{equation*}
	\Psi(x,t) = \offdiag{ \psi(x,t) }{ -\psi(x,t)^*}.
\end{equation*}

If we consider a plane wave solution of \eqref{0.1}, $\psi^p(x,t) =  A e^{i(k x - \omega t)/\eps},\ \omega = \omega(k) = A^2+k^2/2$,  then the exact simultaneous solution of the Lax pair \eqref{0.3} is given by 
\begin{equation}\label{sim sol for plane waves}
	\Phi^p(x,t) = e^{-\frac{i}{2\eps}(kx-\omega t) \sig} \Ecal(z; -k/2-A,-k/2+A) e^{ \frac{i}{\eps} 
	\Lambda(z,-k/2- A, -k/2+A) ( x + (z-k/2) t) \sig},
\end{equation}	   
where
\begin{equation}\label{0.4b}
	\begin{gathered}
	\Lambda(z; \Acal,\Bcal) := \sqrt{z - \Acal}\sqrt{z - \Bcal}, \\
	\beta(z; \Acal,\Bcal) := \lp \frac{z-\Bcal}{z-\Acal} \rp^{1/4}, \\
	\Ecal(z, \Acal, \Bcal) := \begin{pmatrix}
		\frac{ \beta(z; \Acal, \Bcal) + \beta(z; \Acal, \Bcal)^{-1} }{2} &
		- \frac{ \beta(z; \Acal, \Bcal) - \beta(z; \Acal, \Bcal)^{-1} }{2i} \\
		\frac{ \beta(z;\Acal, \Bcal) - \beta(z; \Acal, \Bcal)^{-1} }{2i} &
		\frac{ \beta(z; \Acal, \Bcal) + \beta(z; \Acal, \Bcal)^{-1} }{2}. 
	\end{pmatrix}
	\end{gathered}
\end{equation} 
and $\Lambda(z; \Acal, \Bcal)$ and $\beta(z; \Acal, \Bcal)$ are defined to be branched on $[\Acal, \Bcal]$ and normalized such that 
$$
	\begin{aligned} \Lambda(z; \Acal, \Bcal) = z + \bigo{z^{-1}} \\ 
	\beta(z; \Acal, \Bcal) = 1 + \bigo{z^{-1}} \end{aligned}
	\qquad z \to \infty.
$$	

For initial data $\psi_0(x)$ which is asymptotic to a plane wave for large $x$, i.e., $\psi_0(x) \sim \psi^p(x)$ as $x\to \pm \infty$, it is reasonable to define the Jost function solutions of \eqref{0.3a} to be those whose asymptotic behavior is given by $\Phi^{p}$. For our particular family of initial data \eqref{0.2} this implies that our left and right normalized Jost functions satisfy
\begin{equation}\label{0.4}
	\begin{aligned}
	&\lim_{x \to - \infty}  \Phi_L(x,t) e^{i \Lambda(z; -1,1) x \sig} = \Ecal(z; 1) \\
	&\lim_{x \to  \infty} e^{i \mu x  \sig} \Phi_R(x,t) e^{i \Lambda(z; \lam_-, \lam_+) x \sig} = \Ecal(z; \lam_-, \lam_+)
	\end{aligned}
\end{equation}
 
For brevity we will use the shorthands 
$$
\begin{aligned}
&\beta_L(z) := \beta(z, -1, 1), 
\qquad
&\Lambda_L(z) := \Lambda(z, -1, 1), 
\qquad
&\Ecal_L(z) := \Ecal(z,-1,1), \\ 
&\beta_R(z) := \beta(z, \lam_-, \lam_+), 
\qquad
&\Lambda_R(z) := \Lambda(z, \lam_-, \lam_+), 
\qquad
&\Ecal_R(z) := \Ecal(z,-\lam_-, \lam_+),
\end{aligned}
$$
and we denote the branch cut intervals of these functions:
$$
	I_L=(-1,1) \qquad I_R = (\lam_-, \lam_+)
$$
Note, that these branch points are exactly the Riemann invariants for the (constant) plane wave solutions \eqref{0-phase riemann invariants} corresponding to each half of the initial data \eqref{0.2}.

\subsection{Forward scattering of our pure shock initial data} 
For general step-like initial data one can prove existence and analytic extension (in $z$) theorems for the Jost functions \cite{BP82}, \cite{FT07}. However, for initial data given by \eqref{0.2} the Jost functions are explicit:
\begin{equation}\label{0.5}
\begin{split}
	\Phi_L(x;z) &=
	\begin{cases}
		\Ecal_L(z) e^{-i \Lambda_L(z) x \sig/\eps} & x < 0 \\ 
		e^{-i \mu_R x \sig/\eps} \Ecal_R(z) e^{-i \Lambda_R(z) x \sig/\eps} \Ecal_R^{-1}(z) \Ecal_L(z) &  x > 0
	\end{cases} \\
	\Phi_R(x;z) &=
	\begin{cases}
		\Ecal_L(z) e^{-i \Lambda_L(z) x \sig / \eps} \Ecal_L^{-1}(z) \Ecal_R(z) & x < 0 \\ 
		e^{-i \mu_R x \sig / \eps} \Ecal_R(z) e^{-i \Lambda_R(z) x \sig / \eps}  &   x > 0
	\end{cases} 
\end{split}
\end{equation}

\begin{prop}\label{prop:1}
For $k \in \{ L,R\}$, let $\Phi_k(z;z)$ be defined by \eqref{0.5}. The following properties are easily verified:
\begin{enumerate}[1.]
	\item $\det \Phi_k = 1$.
	\item $\Phi_k(x;z)$ is analytic for $z \in \C \backslash I_k$.
	\item $e^{i \mu_k x \sig / \eps} \Phi_k(x;z) e^{i \Lambda_k(z) x \sig / \eps} = I + \bigo{z^{-1}}$ as $z \to \infty$.
	\item For $z \in I_k$, $\Phi_k(z)$ takes continuous boundary values satisfying 
$$
	\Phi_{k+}(z) = \Phi_{k-} (z) \offdiag{-1}{1} \qquad z \in I_k.
$$
	\item $(z - p)^{1/4} \Phi_k(x;z)$ is bounded as $z \to p$ where $p$ is either endpoint of $I_k$.
\end{enumerate}	
\end{prop}
	
From the Jost functions, we define the scattering matrix 
\begin{equation}\label{0.6}
	S(z) := \Phi_R^{-1} \Phi_L(z)  = \Ecal_R^{-1}(z) \Ecal_L(z) = 
	\begin{pmatrix}	 a(z) & b^*(z) \\ b(z) & a^*(z) \end{pmatrix}
\end{equation}
where the scattering functions and reflection coefficient are given by
\begin{equation}\label{0.7}
\begin{gathered}
	a(z) = \frac{ \beta_L(z) \beta_R(z)^{-1} + \beta_L(z)^{-1} \beta_R(z) }{2} \\
	b(z) = \frac{ \beta_L(z) \beta_R(z)^{-1} - \beta_L(z)^{-1} \beta_R(z) }{2i} \\
	r(z)  = \frac{ b(z) }{ a(z) } = -i \frac{\beta_L(z)^2 - \beta_R(z)^2}{\beta_L(z)^2 + \beta_R(z)^2}.
\end{gathered}
\end{equation}
By direct calculation, or as a consequence of \eqref{0.5}, \eqref{0.6} and Proposition~\ref{prop:1}, we see that the scattering functions are analytic in $\C \backslash \lp I_L \triangle I_R \rp$\footnote{$A \triangle B$ denotes the disjoint union of $A$ and $B$.} and satisfy the jump relations
\begin{equation}\label{0.8}
	\begin{gathered}
	a_+(z) = \phantom{-}b_-^*(z) \qquad b_+(z)  = \phantom{-}a^*_-(z) \qquad  z \in I_L \backslash ( I_L \cap I_R), \\ 
	a_+(z) = -b_-^*(z) \qquad b_+(z) = -a^*_-(z)  \qquad  z \in I_R \backslash ( I_L \cap I_R).
	\end{gathered}
\end{equation}
It follows that $r(z)$ is also analytic for $z \in \C \backslash (I_L \triangle I_R)$ and 
\begin{equation}\label{0.9a}
	r_+(z) = \frac{1}{r_-^*(z)} \qquad   z \in I_L \triangle I_R
\end{equation}	
Furthermore, from \eqref{0.7} it is easy to verify that 
\begin{equation}\label{0.9b}
	\begin{aligned}
	z \in \C \backslash (I_L \triangle I_R)& \quad \Longrightarrow \quad	 |r(z)|^2 < 1  \\
	z \in I_L \triangle I_R&  \quad \Longrightarrow \quad |r_\pm(z)| = 1
	\end{aligned}
\end{equation}

\begin{prop} 
The function $a(z)$ defined by \eqref{0.7} has no zeros in the complex plane.
\end{prop}
\begin{proof} 
The mapping  $w= \beta(z; \Acal)$ defined for any $\Acal >0$ by \eqref{0.4b}  is a conformal map of $\C \backslash [-1,1] \to \U$, where $\U = \{ w \in \C \backslash \{0\}\, : \, |\arg w| < \frac{\pi}{4} \}$, such that 
$\beta(\C^\pm) = \U \cap \C^\pm$. Since $\beta_L(z) = \beta(z,1)$ and $\beta_R(z) = \beta(z-\mu, A)$ are just real translation and scalings of $\beta(z)$, each is such a conformal mapping into $\U$ and it follows that $\re  \frac{\beta_L(z)}{\beta_R(z)} > 0$ and, thus $\re a(z) >0$, for all $z$. 
\end{proof}    

One consequence of \eqref{0.9a}-\eqref{0.9b} is that the the transmission coefficient $1/a(z)$ does have zeros on the real axis. Indeed the squared transmission coefficient 
\begin{equation}\label{0.9c}
	\frac{1}{|a(z)|^2} = 1 - r(z) r^*(z)  
	=  \frac{ 4 \beta_L(z)^2 \beta_R(z)^2 }{( \beta_L(z)^2 + \beta_R(z)^2)^2}  
\end{equation}
is analytic for $z \in \C \backslash (\I_L \triangle \I_R)$ and vanishes as a square root at each of the four branch points. It has no other zeros or poles.

Using the time dependent Jost functions we construct the piecewise analytic function
\begin{equation}\label{0.10}
	m(z;x,t) :=  \begin{cases}
	\begin{bmatrix} \frac{ \phi_L^{(1)}(x,t;z)}{a(z)} & \phi_R^{(2)}(x,t;z) \end{bmatrix} e^{i (tz^2 + xz) \sig / \eps} & z \in \C^+ \medskip \\
	\begin{bmatrix} \phi_R^{(1)}(x,t;z) & \frac{ \phi_L^{(2)}(x,t;z)}{a^*(z)}   \end{bmatrix} e^{i(tz^2 +x z) \sig / \eps} & z \in \C^-
	\end{cases}
\end{equation}
The function $m(z;x)$ satisfies the following Riemann Hilbert problem:
\begin{rhp}[for $m(z;x, t)$]\label{rhp:1}
Find a $2\times2$ function $m(z;x,t)$ with each of the following properties:
\begin{enumerate}[1.]
	\item $m(z;x,t)$ is analytic in $\C \backslash \R$.
	\item $m(z;x,t) = I + \bigo{z^{-1}}$ as $z \to \infty$.
	\item For $z \in \R$, $m$ satisfies the jump relation $m_+(z;x,t) = m_-(z;x,t) v(z;x,t)$ where
	\begin{equation}\label{1.11}
		v(z,x,t) = \begin{cases}
			\begin{pmatrix} 1 - r r^* & -r^* e^{-2i\theta/\eps} \\ r e^{2i\theta/\eps} & 1 \end{pmatrix} & z \in \R \backslash (I_L \cup I_R) \\
			\begin{pmatrix} 0 & -r_-^* e^{-2i\theta/\eps} \\ r_+ e^{2i\theta/\eps} & 1 \end{pmatrix} & z \in I_L \backslash (I_L \cap I_R) \\
			\begin{pmatrix} (a_+ a_-^*)^{-1} & - e^{-2i\theta/\eps} \\  e^{2i\theta/\eps} & 0 \end{pmatrix} & z \in  I_R \backslash (I_L \cap I_R) \\
			\offdiag{-e^{-i\theta/\eps}}{e^{i\theta/\eps}} & z \in I_L \cap I_R
		\end{cases}
	\end{equation}
	where 
	$$
		r = r(z) \quad \text{and} \quad \theta = \theta(x,t,z) := xz + t z^2.
	$$
	\item $m(z;x,t)$ is bounded at each finite $z$ except the points $p$, $p\in\{\lambda_-, \lambda_+\}$ where it admits the singular behavior 
	\begin{equation}
		\begin{aligned}
		m(z;x,t) &= \bigo{ \begin{matrix} 
		(z-p)^{1/4} & (z-p)^{-1/4} \\ (z-p)^{1/4} & (z-p)^{-1/4} 
		\end{matrix} }, \quad z \in \C^+, \qquad p \in \{ \lambda_-, \lambda_+ \}, \\
		m(z;x,t) &= \bigo{ \begin{matrix} 
		(z-p)^{-1/4} & (z-p)^{1/4} \\ (z-p)^{-1/4} & (z-p)^{1/4} 
		\end{matrix} }, \quad z\in \C^-, \qquad p \in \{ \lambda_-, \lambda_+ \}.
		\end{aligned}
	\end{equation}  
\end{enumerate}

Let $m_{12}$(z;x,t) denote the $(1,2)$-entry of the matrix $m(z;x,t)$. If a solution of the above Riemann Hilbert problem exists, the function 
\begin{equation}\label{recovery}
	\psi(x,t) := - 2i \lim_{z \to \infty} m_{12}(z;x,t)
\end{equation}
is a solution of \eqref{0.1}.
\end{rhp}

\section{Constructing the $g$-functions of self-similar wave motion.}\label{sec:gfunctions}
One of the essential tools in the steepest descent analysis of Riemann-Hilbert problems is the construction of what is known as a $g$-function, whose role is to renormalize oscillatory or exponentially large factors in the jump matrices. As in the KdV setting \cite{Venakides90}, this function can be characterized as the log transform of the minimizing measure of a certain minimization problem. For a large class of initial data this minimizer is supported on a finite union of disjoint intervals, and the deformation of the endpoints of these intervals as $(x,t)$ vary are governed by the Whitham equations for NLS. 
Here we construct the possible genus-0 and genus-1 $g$-functions possible for self-similar motion following the method of \cite{GT02}. The method clearly generalizes to higher genus.

Suppose that we are given a set of $2G+2$ real points $\lam_1, \lam_2, \dots, \lam_{2G+2}$ ordered such that $\lam_1 > \lam_2 > \dots > \lam_{2G+2}$. Label the intervals $J_k = (\lam_{2k+1}, \lam_{2k+2})$, $k=0,\dots, G$ and $J = \bigcup_{k=0}^G J_k $. We call the intervals $J_k$ the `bands' and their complement $\R \backslash \overline{J}$ the `gaps'. Finding the $g$-function, the log transform of the minimizer of the minimization problem, is equivalent to constructing a scalar function $g(z)$ with the following properties:
\begin{table}[H]
\caption{\label{glist} Analytic properties of the $g$-function.}
\begin{enumerate}[1.]
	\item $g(z)$ is analytic in $\C \backslash ( \lam_{2G+2}, \lam_1)$.
	\item $2\theta(z;x,t) - g_+(z) - g_-(z) = 2\alpha_k$, for $z \in J_k$, $k=0,\dots , G$.
	\item $g(z) = g_\infty + \bigo{z^{-1}}$ as $|z| \to \infty$.
	\item $g(z) - \theta(z;x,t) = \bigo{ (z-\lam_k)^{q_k}}$ as $z \to \lam_k$.
	\item $\imag g(z) = 0, z \in  \R \backslash \overline{J}$.
	\end{enumerate}
\end{table}

\begin{rem} The growth condition at endpoints is often omitted in the literature, as it is generically understood to be ``3/2 vanishing"  at each endpoint. However, it is a necessary condition for uniqueness. In every cases we consider $q_k \in \{3/2, 1/2 \}$ so that the problem for $d \varphi$ following \eqref{2.6.0} always has a unique solution.
\end{rem} 

Consider the Riemann surface of genus $G \geq 0$:
\begin{align*}\label{surface}
	\SSS_G &:= \left\{ P = (z, \RR), \ \RR^2 = \prod_{j=1}^{2G+2} (z - \lam_j) \right\}, \\
	\lam_1 &> \lam_2 > \dots > \lam_{2g+2}, \lam_j \in \R
\end{align*}
The projection $\pi(P) = z$, defines $\SSS_G$ as a two-sheet cover of $C\Poly^1$.
We take our basis $\{ a_j, b_j \}_{j=1}^G$ of the homology group $H_1(\SSS_G)$ so that $a_j$ lies entirely on the upper sheet and encircles with positive (counterclockwise) orientation $\overline{J_j} = [\lam_{2j+1} , \lam_{2j+2}]$, $j=1,\dots, G$, while $b_j$ emerges from $J_0= (\lam_2, \lam_1)$ on the upper sheet passes counterclockwise to the lower sheet through $J_j = (\lam_{2j+1}, \lam_{2j+2})$ and returns to the initial point entirely on the lower sheet, see Figure~\ref{fig:homology}.

\begin{figure}[htb]
	\begin{center}
		\includegraphics[width=.8\textwidth]{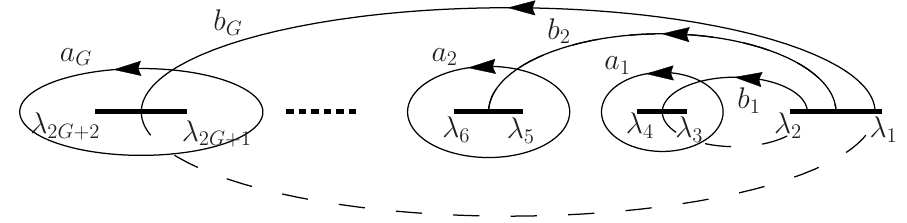}
		\caption{Our choice of homology basis $a_1,\dots, a_G$ and $b_1, \dots b_G$ associated with the genus $G$ hyperelliptic Riemann surface $\SSS_G =\left\{ P = (z, \RR)\, :\,  \RR^2= \prod_{j=1}^{2G+2} (z - \lam_j) \right\}$. 
		\label{fig:homology} 
		}
	\end{center}
\end{figure}

Let $\nu_j$, $j=1,\dots, G$ denote the canonical basis of holomorphic one-forms (Abelian differentials of the first kind) on $\Sigma_G$:
\begin{equation}
	\nu_j(z) = \frac{ c_{j,1} z^{G-1} + c_{j,2} z^{G-2} + \dots + c_{j,G}} {\RR} dz
\end{equation}
where the constants $c_{j,i}$ are uniquely determined by the normalization conditions
$$
	\oint_{a_k} \nu_j = \delta_{kj}, \quad j,k =1 ,\dots, G.
$$	

Additionally let $\omega^{(k)}$, $k=0,1$, denote the Abelian differentials of the second kind on $\SSS_G$ 
given by
\begin{equation}\label{2.8.0}
	\begin{aligned}
	\omega^{(k)} &= \frac{ P_{k}(z, \vect{\lam})}{\RR}  dz \\
	P_{k}(z; \vect{\lam}) &= z^{k+G+1} + \Gamma_1 z^{k+G} 
	+ \dots + \Gamma_{k+1}z^G+a_{k,1} z^{G-1} + \dots + a_{k,G}
	\end{aligned}
\end{equation}
where $\Gamma_j = \Gamma_j(\vect \lam)$ are the coefficients of the expansion
\begin{equation}
	\RR(z; \lam) = \lp \prod_{k=1}^{2G+2} (z - \lam_k) \rp^{1/2} = z^{G+1} \lp 1 + \frac{ \Gamma_1}{z} + \dots + \frac{\Gamma_m}{z^m} + \dots \rp,
\end{equation}
and the $a_{k,j} = a_{k,j}(\vect \lam)$ are determined by the normalization condition
\begin{equation}\label{2.7.0}
	\oint_{a_j} \omega^{(k)} = 0, \qquad j=1,\dots, G.
\end{equation}
For large arguments $\omega^{(k)}$ admits the expansion
\begin{equation}\label{2.6.0}
	\omega^{(k)} = \pm \lb z^k + \bigo{z^{-2}} \rb dz, \qquad P \to (+\infty, \pm \infty),
\end{equation}
so it has poles of order $2k+2$ at $(\infty, \pm \infty)$.


Now for given vanishing conditions $\rho_k \in (2\N_0+1)/2$ we want to construct a differential $d\varphi$ which has the following properties:
\begin{enumerate}[1.]
	\itemsep -.75em
	\item $d\varphi$ is meromorphic on $\SSS_G$ whose only poles are at $(+\infty, \pm \infty)$.
	\item $d\varphi \mp d\theta$ is locally holomorphic as $P \to (+\infty, \pm \infty)$.
	\item $\oint_{a_k} d\varphi = 0$ for $k=1,\dots, G$.
	\item $d\varphi = \bigo{ (z-\lam_k)^{\rho_k-1} dz }$ as $z \to \lam_k$.
\end{enumerate}
For any choice of moduli $\vect{\lam}$, the first three conditions define a meromorphic differential of the second kind, which given by 
\begin{equation}\label{2.10.0} 
	d\varphi = 2t\,  \omega^{(1)} + x\, \omega^{(0)}.
\end{equation}
If the fourth condition is also satisfied, then the function 
\begin{equation}
	g(z)  = \theta(z) - \theta(\lam_1) -  \int_{\lam_1}^z d \varphi,
\end{equation}
where the path of integration lies in $\C \backslash (\lam_{2G+2}, \lam_{2G})$, satisfies the conditions in Table~\ref{glist}. 

Moreover, the function 
\begin{equation}
	\varphi(z) := \int_{\lam_1}^z d \varphi
\end{equation}
is analytic in $\C \backslash \bigcup_{k=0}^G (\lam_{2k+1}, \lam_{2k+2} )$ and satisfies the jump relations
\begin{equation}
	\varphi_+(z) + \varphi_-(z) = \begin{cases}
		0 & z \in (\lam_1, \lam_2) \\
		\oint_{b_k} d \varphi, & z \in (\lam_{2k+1}, \lam_{2k+2}), \quad k=1,\dots, G
	\end{cases}
\end{equation}

For our purposes we consider the following situation. The half-plane $(x,t \geq 0)$ is divided into distinct domains $D_m$ such that in each $D_m$ we have a fixed genus $G \geq 0$ and the moduli $\lam_j$ are split into two types:

\begin{enumerate}[]
	\item \textbf{1. Hard edges:} These $\lam_j$ are known and constant for $(x,t) \in D_m$. We require that $d\varphi = \bigo{ (z-\lam_j)^{-1/2} dz }$ as $z\to \lam_j$. \smallskip
	\item \textbf{2. Soft edges:} These $\lam_j$ are allowed to move for $(x,t) \in D_m$; 
	their motion is described by the condition that $d\varphi = \bigo{ (z-\lam_j)^{1/2} dz}$ as $z\to \lam_j$. 
Using \eqref{2.7.0} and \eqref{2.8.0} the soft edge condition is equivalent to 
\begin{equation}\label{2.12.0}	
		\textbf{$\lam_j$ is a soft edge if:\ }   
		x - V_j(\vect \lam) t = 0, 
		\qquad 
		V_j(\vect \lam) = - 2 \frac{P_{1}(\lam_j, \vect \lam)}{P_{0}(\lam_j, \vect \lam)} 
\end{equation}
\end{enumerate}
\Cref{2.12.0} is simply states that the motion of the branch points $\lam_i$ are described by the self-similar solutions of the genus-G Whitham equations \eqref{self-similar Whitham system}. Furthermore, as the Whitham equations for dNLS are strictly hyperbolic \cite{Kodama99}, it follows that any self-similar solutions of the Whitham equations admits at most one soft edge. 

\begin{rem}\label{soft edge variation}
Note that if $d\varphi$ has a soft edge $\lam_s$, then the differential $d\varphi_{\lam_s}$ obtained by differentiating $d\varphi$ with respect to the parameter $\lam_s$ is identically zero. Using \eqref{2.6.0}-\eqref{2.10.0} it follows that $d\varphi_{\lam_s}$ has no poles at either infinity, and from \eqref{2.12.0} is regular at the soft edge $\lam_s$ as well. Therefore $d\varphi_{\lam_s}$ is a holomorphic differential all of whose $a$-periods vanish, \ie, $d\varphi_{\lam_s} \equiv 0$.
\end{rem}

\subsection{Self-similar genus zero g-functions}\label{sec:onecutg}
In the genus zero case $\vect \lam = (\lam_1, \lam_2)$ and the first homology group is trivial as any closed loop is homotopic to a point. The polynomials associated with our second kind differentials \eqref{2.8.0} are given by
\begin{equation}\label{2.13.0}
	\begin{aligned}
	P_{0}(z, \vect \lam) &= z - \frac{1}{2} e_1(\vect \lam), \\
	P_{1}(z, \vect \lam) &= z^2 - \frac{1}{2} e_1(\vect \lam)  z  
	+ \lp \frac{1}{2} e_2(\vect \lam) - \frac{1}{8} e_1(\vect \lam )^2 \rp,  	
	\end{aligned}
\end{equation}
where 
\begin{align*}
	e_1( \vect \lam) = \sum_{j=1}^{2G+2} \lam_j, \qquad
	e_2( \vect \lam) = \sum_{1 \leq j <k}^{2G+2} \lam_j \lam_k
\end{align*}
are the first two elementary symmetric polynomials. The genus-0 speeds $V_{j}$ in \eqref{2.12.0} are given by 
\begin{equation}\label{2.14.0}	
	V_{j}(\vect \lam) = -\frac{1}{2} e_1(\vect \lam) - \lam_j, \qquad j = 1,2.
\end{equation}
and 
\begin{equation}\label{2.10}
	d\varphi =  \frac{ 2t P_{1}(z,\vect \lam) + x P_{0}(z, \vect \lam)} { \RR(z; \vect \lam) } dz 
		=  \frac{ 2t (z-\xi_+)(z-\xi_-)} { \RR(z; \vect \lam) } dz 
\end{equation}	
where $\RR(z; \vect \lam)  = \sqrt{(z-\lam_1)(z-\lam_2)}$ is cut on $(\lam_2, \lam_1)$ and $\RR \sim z$ as $z \to \infty$.
\subsubsection{The one-cut, hard edged case (plane waves)}\label{sec:onecutg_hard} 
If we suppose that $\{\lam_1, \lam_2\}$ are known (constant) hard edges, then the stationary points, 
the zeros of $d \varphi$, are given by 
\begin{equation}\label{2.11}
	\xi_\pm = \frac{\lam_1+ \lam_2 - \tau}{4} \pm \frac{1}{4} 
	\sqrt{ (\lam_1+ \lam_2 + \tau)^2 + 2(\lam_1-\lam_2)^2}, \quad \tau = \frac{x}{t}.
\end{equation}  
Each is a monotone decreasing function of $\tau$ with the following special values: 
\begin{equation}
\renewcommand*{\arraystretch}{1.25}
\begin{array}{r|cccc}  
 \tau & -\infty & -\frac{1}{2} \lp 3\lam_1+\lam_2 \rp & -\frac{1}{2} \lp \lam_1+ 3\lam_2 \rp & \infty \\ \hline 
 \xi_-(\tau) & \frac{1}{2} \lp \lam_1+\lam_2 \rp & \frac{1}{4} \lp \lam_1+3\lam_2 \rp & \lam_2 & - \infty \\
\xi_+(\tau) & \infty & \lam_1 & \frac{1}{4} \lp 3\lam_1+\lam_2 \rp & \frac{1}{2} \lp \lam_1 + \lam_2 \rp \\
\end{array}
\end{equation}

With $d \varphi$ defined by \eqref{2.10}, the $g$-function, analytic for $z \in \C \backslash (\lam_2, \lam_1)$, is given by:
\begin{equation}\label{2.12}
	g(z) := \theta(z) - \theta(\lam_1) -  \int_{\lam_1}^z d\varphi 
\end{equation}
where the path of integration does not pass through the branch cut $(\lam_2, \lam_1)$. 
The integral term can be computed explicitly, 
\begin{equation}\label{2.22}
	\varphi(z) := 
	\int_{\lam_1}^z d\varphi =
	 t \RR(z, \vect \lam) \lp z + \frac{1}{2}(\lam_1 + \lam_2) + \tau \rp.
\end{equation}
Clearly, $g$ is bounded at infinity by virtue of the growth condition on $dg$ and 
\begin{equation}\label{2.24}
	g(\infty) = - \theta(\lam_1) +  x \lp \frac{ \lam_1 + \lam_2}{2} \rp  + 
	t \lb  \lp \frac{ \lam_1 + \lam_2}{2} \rp^2 + \frac{1}{2}  \lp \frac{ \lam_1 - \lam_2}{2} \rp^2 \rb.
\end{equation}

For $\varphi$ defined by \eqref{2.22} the structure of the imaginary signature table depends on the position of the two real stationary points $\xi_\pm(\tau)$ relative to the branch points $\lam_1$ and $\lam_2$.
For $\tau < \frac{1}{4}( \lam_1 + 3\lam_2)$, the level set $\imag \varphi =0$ consist of the real axis minus the cut and an asymptotically vertical contour through $\xi_-(\tau) < \lam_2$. For $\tau > \frac{1}{4}( 3 \lam_1 + \lam_2)$ the situation is reversed, and the vertical contour passes through $\xi_+(\tau) > \lam_1$. For $\frac{1}{4}(\lam_1 + 3\lam_2) < \tau  <  \frac{1}{4}(3\lam_1 + \lam_2)$, both $\xi_-(\tau)$ and $\xi_+(\tau)$ lie on the cut. In this case the vertical component of $\imag \varphi = 0$ passes through the point 
\begin{equation}
	\xi_0(\tau) = -\frac{1}{2}(\lam_1 + \lam_2 + 2\tau)
\end{equation}
which lies between $\xi_-(\tau)$ and $\xi_+(\tau)$. See Figure~\ref{fig:1}.

\begin{figure}[thbp]
\begin{center}
	\includegraphics[width = .32\textwidth]{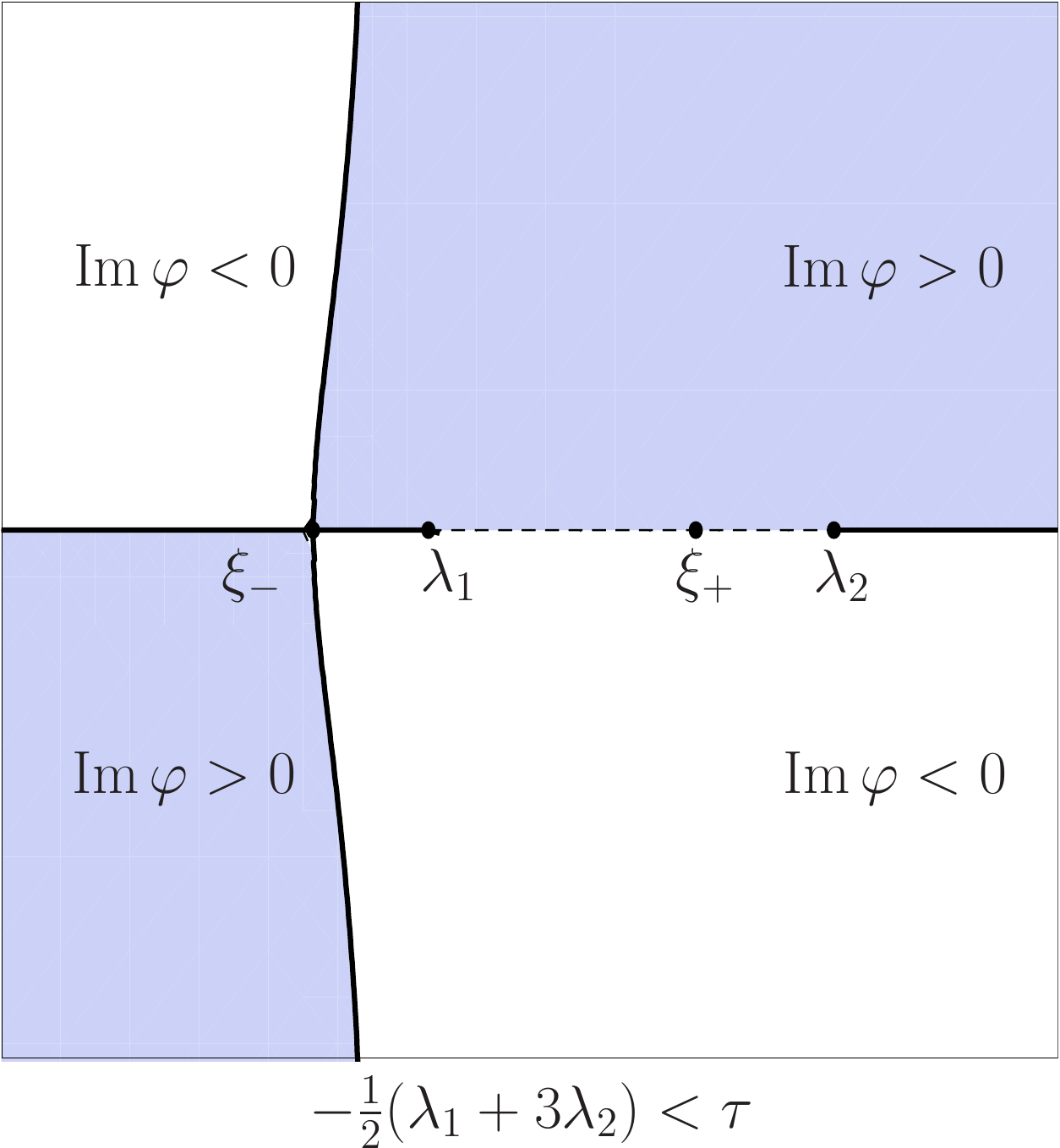}
	\includegraphics[width = .32\textwidth]{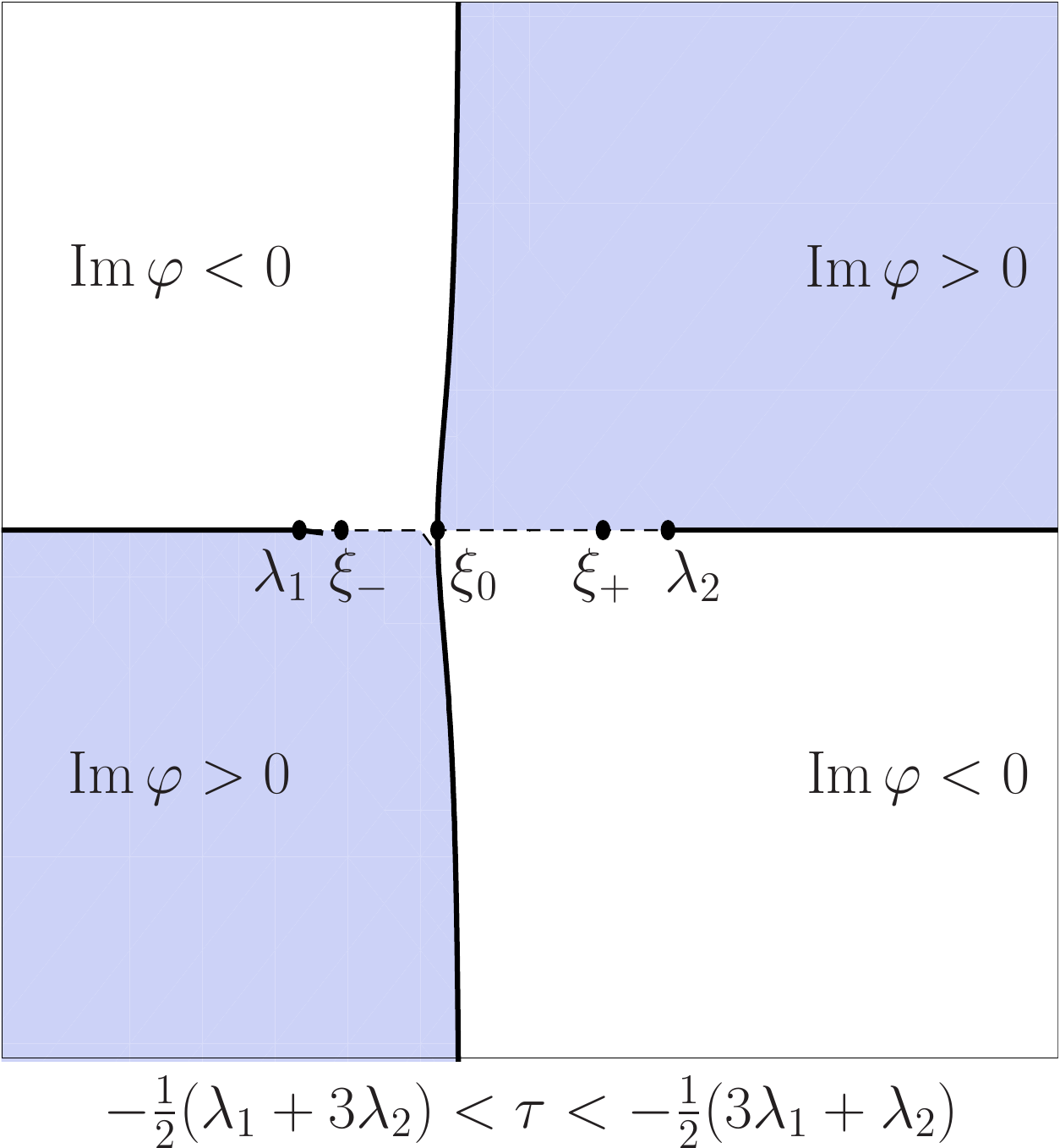}
	\includegraphics[width = .32\textwidth]{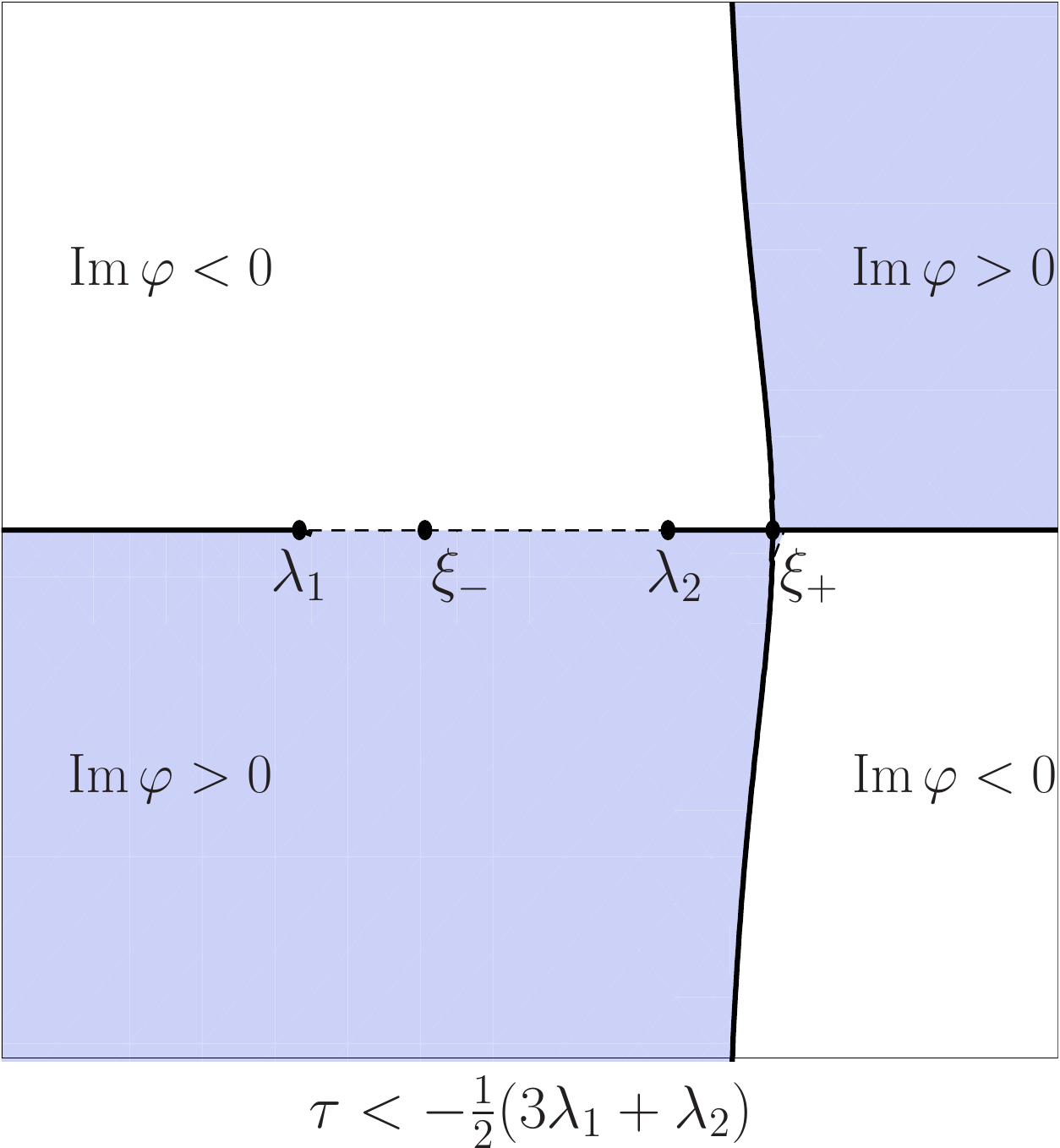}		
\caption{The topological structure of the sign table for $\imag \varphi$ bifurcates as shown as the stationary phase points $\xi_\pm(\tau)$ pass through $\alpha$ and $\beta$, the branch points of $\varphi$. 
\label{fig:1}}
\end{center}
\end{figure}

\subsubsection{The one-cut, hard/soft edge case (rarefaction waves)}\label{sec:onecutg_soft} 

If $dg$ is cut on a single interval $(\lam_1, \lam_2)$, and we suppose that one branch point is a soft edge $\lambda_s$ and the other is a known hard edge $\lambda_h$, then the conditions \eqref{2.12.0}, \eqref{2.14.0} 
effectively `pin' one zero of the numerator in \eqref{2.10} to $\lam_s$, leaving one stationary point $\xi$.
Solving these conditions gives the motion of the soft edge $\lam_s$ and stationary point $\xi$ in terms of $x,t,$ and $\lam_h$:
\begin{equation}\label{hard/soft motion}
\begin{aligned}
	\lam_s &= -\frac{1}{3} \lp 2 \tau + \lam_h \rp. \\
	\xi &=  \frac{1}{4}\lp \lam_s + 3\lam_h \rp = \frac{1}{6} \lp 4\lam_h - \tau \rp.
\end{aligned}
\end{equation}
Note that $\xi$ always lies on the branch $(\lam_2, \lam_1)$. 

In this notation $d\varphi$ has the explicit representation
\begin{equation}
	d\varphi = 2t \lp \frac{z - \lam_s}{z - \lam_h} \rp^{1/2} (z - \xi) dz  
\end{equation}
As before we define
\begin{equation}\label{hard/soft g}
	\begin{gathered}
		g(z) = \theta(z) - \theta(\lam_1) - \varphi(z), \\
		\varphi(z) = \int_{\lam_1}^z d\varphi
		= 2t \int_{\lam_1}^z \lp \frac{\lambda - \lam_s}{\lambda - \lam_h} \rp^{1/2} (\lambda - \xi) d\lambda
		= t (z-\lam_s)^{3/2}(z-\lam_h)^{1/2}. 
	\end{gathered}
\end{equation}
The zero level set of $\imag \varphi$ always consists of the real axis minus the cut $(\lam_1, \lam_2)$ and two trajectories emerging from $\lam_s$ into the upper and lower half-planes respectively. The resulting signature table for $\imag \varphi$ is given in Figure~\ref{fig:2}. Finally, we compute the limit
\begin{equation}\label{hard/soft g infty}
	g(\infty) = \frac{t}{8}  \lp \lam_h^2 -6 \lam_h \lam_s -3 \lam_s^2 \rp - \theta(\lam_1). 
\end{equation}

\begin{figure}[htbp]
\begin{center}
	\includegraphics[width = .4\textwidth]{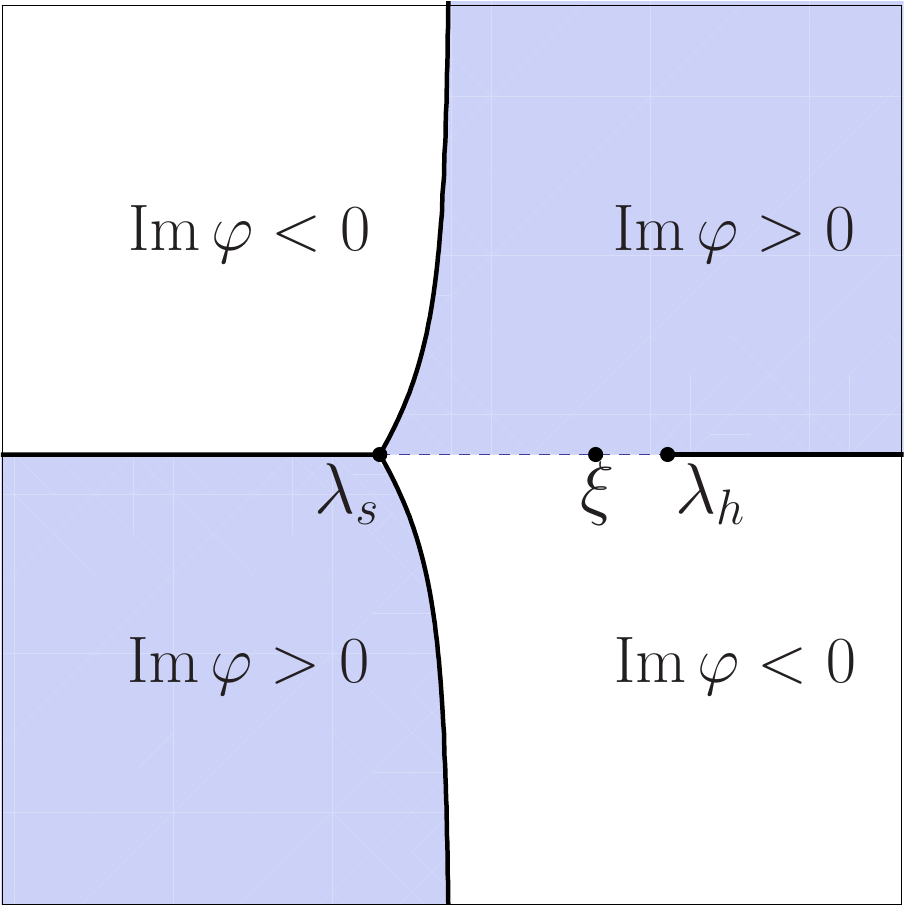}
	\hspace{.1\textwidth}
	\includegraphics[width = .4\textwidth]{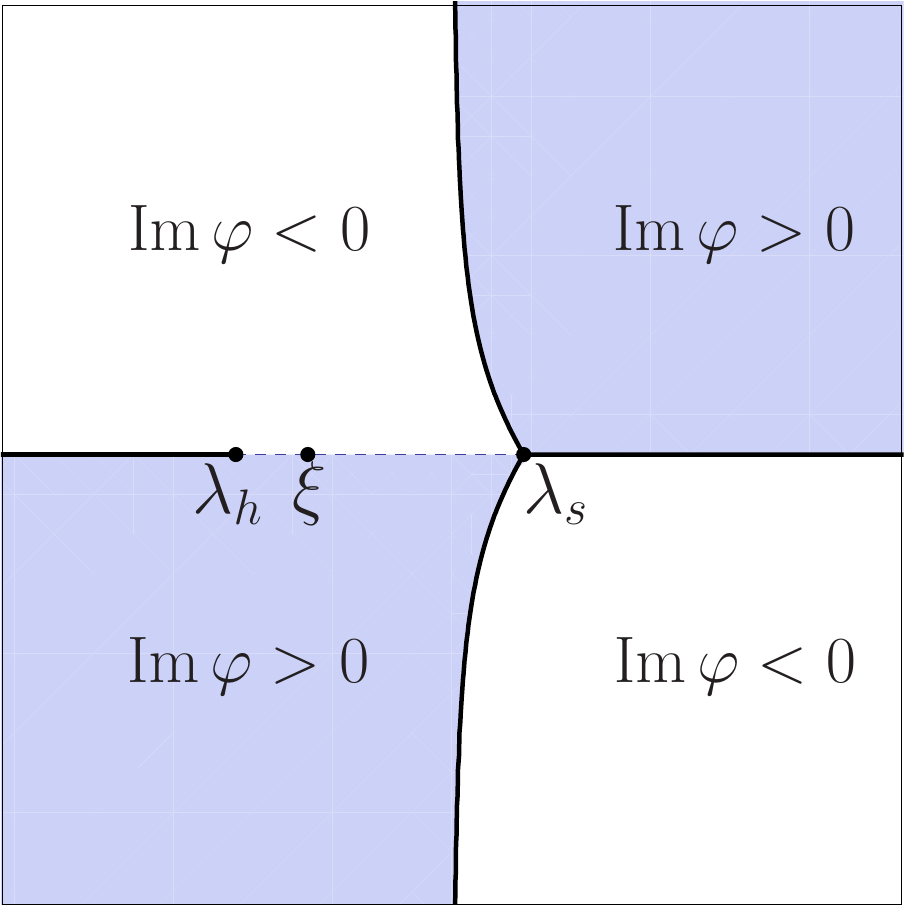}		
\caption{The topological structure of the sign table for $\imag \varphi$ corresponding to the one-cut $g$-function with hard edge, $\lam_h$, and soft edge, $\lam_s$,  see \eqref{hard/soft motion}-\eqref{hard/soft g}. The stationary phase point $\xi$ always lies along the branch cut between 
\label{fig:2}}
\end{center}
\end{figure}

\subsection{Self-similar genus one g-functions}\label{sec:twocutg}
In the genus one case, there are four ordered branch points $\vect \lam = (\lam_1, \lam_2, \lam_3, \lam_4)$, $\lam_1 > \lam_2 > \lam_3 > \lam_4$. The polynomials associated with \eqref{2.8.0} are given by
\begin{equation}\label{2.15.0}
	\begin{aligned}
	P_{0}(z, \vect \lam) &= z^2 - \frac{1}{2} e_1(\vect \lam) z 
	+ a_{0,1}, \\
	P_{1}(z, \vect \lam) &= z^3 - \frac{1}{2} e_1(\vect \lam)  z^2  
	+ \lp \frac{1}{2} e_2(\vect \lam) - \frac{1}{8} e_1(\vect \lam )^2 \rp z + a_{1,1},
	\end{aligned}
\end{equation}
and the differential \eqref{2.10.0} is given by 
\begin{equation}\label{2.32}
	d\varphi =  \frac{ 2t P_{1}(z,\vect \lam) + x P_{0}(z, \vect \lam)} { \RR(z; \vect \lam) } dz, 
\end{equation}	
where $\RR(z; \vect \lam)  = \prod_{k=1}^4\sqrt{(z-\lam_k)}$ is cut on $(\lam_4, \lam_3) \cup (\lam_2, \lam_1)$ and $\RR \sim z^2$ as $z \to \infty$.

The coefficients $a_{0,1}$ and $a_{1,1}$ in \eqref{2.15.0} can be computed explicitly from \eqref{2.7.0} \cite{BF71}:
\begin{equation}
	\begin{aligned}
		a_{0,1} &= \frac{1}{2} (\lam_1 \lam_2 + \lam_3 \lam_4) 
		- \frac{1}{2}(\lam_1- \lam_3)(\lam_2- \lam_4) \frac{ E(m)}{K(m)} \\
		a_{1,1} &= \frac{1}{8} ( \lam_1 \lam_2 - \lam_3 \lam_4)(\lam_1 + \lam_2 - \lam_3 - \lam_4) 
		- \frac{1}{8}e_1(\vect \lam)  (\lam_1- \lam_3)(\lam_2- \lam_4) \frac{ E(m)}{K(m)}. 
	\end{aligned}
\end{equation}
Here $K(m)$ and $E(m)$ are the complete elliptic integrals of the first and second kind respectively with modulus 
$$
	m =m(\vect \lam) = \frac { (\lam_1 - \lam_2)(\lam_3 - \lam_4)} {(\lam_1 - \lam_3)(\lam_2 - \lam_4)}.
$$
Clearly, $m \in (0,1)$ as $\lam_1 > \lam_2 > \lam_3 > \lam_4$. The speeds $V_{j}$ defined by \eqref{2.12.0} can be expressed as 
\begin{equation}\label{2.17.0}
	\begin{aligned}
	V_{j}(\vect \lam) &= -\frac{1}{2} e_1(\vect \lam) + \lp 2 \pd{}{\lam_j} \log L(\vect \lam) \rp^{-1} ,\\
	L(\vect \lam) &= \sqrt{2} \int_{\lam_2}^{\lam_1} \frac{ d \tau} { |\RR(\tau,\vect \lam)| } 
	= \frac{ 2\sqrt{2} K(m) }{\sqrt{(\lam_1 - \lam_3)(\lam_2 - \lam_4)}},
	\end{aligned}
\end{equation}
which are precisely the speeds of the one-phase Riemann invariants for the NLS-Whitham system \eqref{NLS-Whitham}.

\subsubsection{The two-cut, one soft edge case (modulated elliptic waves)} 
If we suppose that one of the branch points, denoted  $\lam_s$, is allowed to evolve as a soft edge while the other branch points are constant hard edges, then the cubic polynomial $2t P_{1}(z, \vect \lam) + x P_{0}(z, \vect \lam)$ has one zero in each band interval; this is a necessary consequence of the fact that $d\varphi$ has been normalized so that all of its $a$-cycles vanish. We label these zeros $\xi_-(\tau) \in (\lam_4, \lam_3)$ and $\xi_+(\tau) \in (\lam_2, \lam_1)$. The remaining zero of the cubic 
polynomial lies at the soft edge, $\lam_s$:
$$
	2 P_{1}(\lam_s, \vect \lam) + \tau P_{0}(\lam_s, \vect \lam)  = 0, \qquad \tau  = \frac{x}{t}, 
	\quad \vect \lam \backslash \lam_s \text{ constant}.
$$	 
This equation determines the motion of the soft edge and, as described by \eqref{self-similar Whitham system} and\eqref{2.17.0}, the motion is exactly that of a self-similar solution of the Whitham equations for the genus-one Riemann invariants of defocusing NLS.

Writing $2t P_{1}(z, \vect \lam) + x P_{0}(z, \vect \lam) = 2t (z-\lam_s(\tau))(z-\xi_-(\tau))(z-\xi_+(\tau))$ we find by comparing coefficients that
\begin{equation}\label{2.31}
	\begin{gathered}
 	 \xi_+(\tau) + \xi_-(\tau) = \frac{1}{2} e_1(\vect \lam) - \lam_s - \frac{\tau}{2} \\
	\xi_+(\tau)\xi_-(\tau) =\frac{1}{2} e_2(\vect \lam)  - \frac{1}{8} e_1(\vect \lam)^2 + \lp \lam_s - \frac{1}{2} e_1(\vect \lam) \rp \lp \lam_s + \frac{\tau}{2} \rp
\end{gathered}
\end{equation}
from which the motion of these station phase points are easily determined. We may write the differential $d \varphi = d\theta - dg$ as
\begin{equation}
	d\varphi = 2t \frac{ (z- \lam_s)(z-\xi_-)(z-\xi_+) }{ \prod_{k=1}^4 \sqrt{z-\lam_k}}
\end{equation}
As before we define 
\begin{equation}\label{twoband g}
	\begin{gathered} 
		g(z) = \theta(z) - \theta(\lam_1) - \varphi(z) \\
	\varphi(z) = 2t \int_{\lam_1}^z \frac{ (\lambda- \lam_s)(\lambda-\xi_-)(\lambda-\xi_+) }
	{ \prod_{k=1}^4 \sqrt{\lambda-\lam_k}} d\lambda 
	\end{gathered}
\end{equation}
so that $\varphi(z)$ is analytic in $\C \backslash ((\lam_4, \lam_3) \cup (\lam_2, \lam_1) )$ and  satisfies the relations
\begin{equation}
\begin{aligned}
	\varphi_+(z) + \varphi_-(z) &= 0 \qquad z \in(\lam_2, \lam_1), \\
	\varphi_+(z) + \varphi_-(z) &= \oint_b \varphi \qquad z \in(\lam_2, \lam_1). 
\end{aligned}
\end{equation}		

Finally, we determine the structure of the signature table for $\imag \varphi$.	
The differential $d \varphi$ is real valued on the real axis minus the bands, with vanishing $a$-cylces, and locally $d\varphi = \bigo{ (z-\lam_h)^{-1/2}}$ at each hard edge and $d\varphi = \bigo{ (z-\lam_s)^{1/2}}$ at the soft edge. 
It follows that the zero level set of $\imag \varphi$ consists of the real axis minus the bands $(\lam_4,\lam_3) \cup(\lam_2, \lam_1)$ and two trajectories emerging from the soft edge $\lam_s$ to infinity through the upper and lower half-planes respectively. 
The resulting signature table for $\imag \varphi$ is given in Figure~\ref{fig:twoband_signs}.	

\begin{figure}[t]
	\begin{center}
		\includegraphics[width=.6\textwidth]{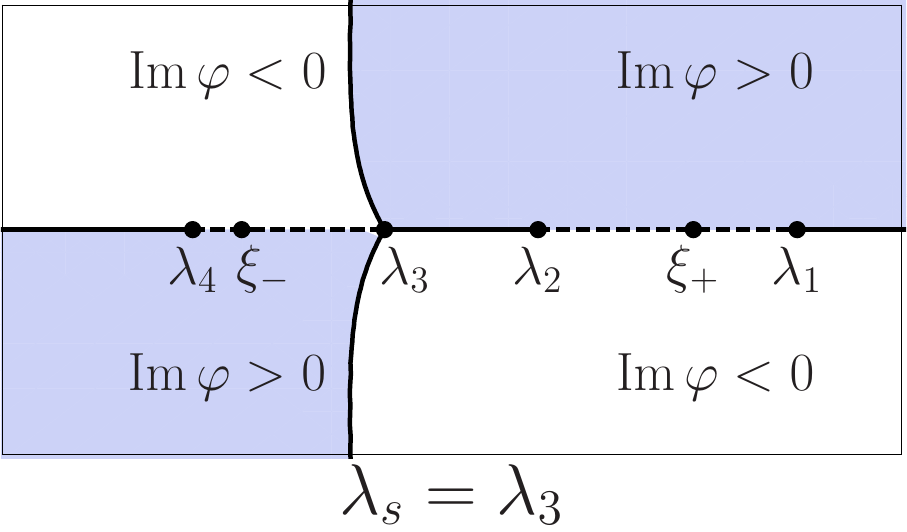}
		\caption{The topological structure of the sign table for $\imag \varphi$ corresponding
		 to the genus one self-similar $g$-function \eqref{twoband g} in the case
		  where the soft edge $\lam_s$  is $\lam_3$ and the other edges are fixed.
		 \label{fig:twoband_signs}
		 }
	\end{center}
\end{figure}

\subsubsection{The two-cut, all hard edge case (unmodulated elliptic waves)}
Though we will not need it in our analysis, the other possible genus-1 $g$-function for self-similar motion is one in which all of the branch points $\vect \lam = (\lam_1, \lam_2, \lam_3, \lam_4)$ are fixed. In this case the phase $\varphi = \int_{\lam_1}^z d\varphi$ has three stationary points at the real roots of $2tP_1(z, \vect \lam) + xP_0(z, \vect \lam)$. Necessarily one root must lie in each band $(\lam_4,\lam_3)$ and $(\lam_2, \lam_1)$, but the third root can vary across the real axis. The signature table for $\imag \varphi$ in this case consist of four components, as in the genus zero case Figure~\ref{fig:1}, but with two cut intervals along the real axis. The point at which the level set $\imag \varphi = 0$ crosses the real axis is the third root when it lies in a gap, or when the third root also lies in a band, the branches of $\imag \varphi = 0$ intersect at the zero of $\varphi$ between the two roots in that band.

\section{Steepest descent analysis}\label{sec:inverse}
We are ready to begin to study solutions of RHP~\ref{rhp:1}. Throughout the section we will refer to the constants $\lam_\pm = \mu \pm A$ which represent the constant Riemann invariants corresponding to the right half of the initial data \eqref{0.2}, and are the endpoints of the interval $\I_R$ related to the branching structure of the reflection coefficient \eqref{0.7}. The course of the inverse analysis depends on the ordering of $\lam_+$ and $\lam_-$ relative to $\pm 1$, the Riemann invariants of the left half of \eqref{0.2}. In Theorem~\ref{thm:main} we only consider the case $-1 < \lam_- < \lam_+ < 1$ and so we will only perform the inverse analysis in this case. It should be clear to the familiar reader how to adapt our calculations to the other five cases without much effort.

We begin the inverse analysis by cataloging a family of jump matrix transformations needed for the nonlinear steepest descent factorizations. We then introduce the initial jump factorizations common to each of the five asymptotic zones identified in Theorem~\ref{thm:main}. Finally, moving left-to-right, we go through the details of establishing the asymptotic behavior of the solution in each of the five zones. As we will see, in this case, when $-1 < \lam_- < \lam_+ < 1$, the initial shock is regularized by a region of rarefaction on the left and a shock wave on the right separated by a central planar plateau.

\subsection{An almanac of matrix factorizations}
Here we record several matrix factorizations that we will refer to when we deform contours onto steepest descent paths. The factorizations are grouped according to the intervals on which they will be used. The off-diagonal exponential factors are omitted but can be included by multiplying on the left and right by the appropriate diagonal factors.

For $z \in \R \backslash (\I_L \cup \I_R)$:
\begin{subequations}\label{1.1}
	\begin{align}
	\label{1.1a}
	\begin{pmatrix} 1 -r r^* & -r^*  \\ r  & 1 \end{pmatrix}
	&= \triu{-r^* } \tril {r } \\
	\label{1.1b}
	&= \tril{ \frac{r }{1-r r^*} } (1- r r^* )^{\sig}
	\triu{ \frac{-r^*}{1-r r^*} }
	\end{align}
\end{subequations}

For $z \in \I_L \backslash ( \I_L \cap \I_R)$, where $r_+ = 1/r_-^*$:
\begin{subequations}\label{1.2}
	\begin{align}
	\label{1.2a}
	\begin{pmatrix} 0 & -r_-^*  \\ r_+  & 1 \end{pmatrix}
	&= \triu{-r_-^* } \tril {r_+  } \\
	\label{1.2b}
	&= \tril{ \frac{r_-}{1-r_- r_-^*} }
	\offdiag{ - r_-^* }{ r_+ }
	\triu{ \frac{-r_+^* }{1-r_+ r_+^*} }
	\end{align}
\end{subequations}

For $z \in \I_R \cap \I_L$, where $r$ is analytic and $r^* = -r$:
\begin{subequations}\label{1.4}
	\begin{align}	
	\label{1.4a}
	\offdiag{-1}{1} 
	&= 	\triu{ -r^*}
		\begin{pmatrix} 0 & -1 \\ 1 & 0 \end{pmatrix}
		\tril{r} \\
	\label{1.4b}
	&=	\tril{ \frac{ r}{1-r r^*} }
		\begin{pmatrix} 0 & -1 \\ 1 & 0 \end{pmatrix}
		\triu{ \frac{ -r^*}{1-rr^*} }		
	\end{align}
\end{subequations}

In our main theorem, Theorem~\ref{thm:main}, we suppose that $\I_R \subset \I_L$, so the above factorizations are sufficient to perform the inverse analysis. In some of the other five cases the following factorization is also needed. 

For $z \in \I_R \backslash ( \I_L \cap \I_R)$, where $r_+ = 1/r_-^*$:
\begin{subequations}\label{1.3}
	\begin{align}
	\label{1.3a}
	\begin{pmatrix} (a_+ a_-^*)^{-1}  & -1  \\ 1  & 0 \end{pmatrix}
	&= \triu{-r_-^* }
	 \offdiag{-1}{1}
	 \tril {r_+  } \\
	\label{1.3b}
	&= \tril{ \frac{r_-}{1-r_- r_-^*} }
	( a_+ a_-^* )^{-\sig}
	\triu{ \frac{-r_+^* }{1-r_+ r_+^*} }
	\end{align}
\end{subequations}
When $a$ and $b$ are analytic, $(a a^*)^{-1} = 1 - r r^*$, the quantity $(a_+ a_-^*)^{-1}$ in the above factorization is a non-vanishing extension of $1- r r^*$ into $\I_R \backslash ( \I_L \cap \I_R)$.

\subsection{The standard sequence of matrix transformations}
In the subsequent sections we describe the steepest descent analysis for RHP~\ref{rhp:1} in each of the six possible parameter regimes. In order to streamline this procedure, we record the sequence of transformations which lead from the initial RHP to one which is amenable to asymptotic expansion. In each case the transformation is the same up to redefinition of the $g$-functions, deformations of the various domains of definition, and the transition ``times". In what follows we will define the $g$-functions and domains for each instance and point out the critical behavior at each transition time appropriate to each case. It will then remain in each case to compute the leading order behavior of the solution of RHP~\ref{rhp:1}.

The transformation to an asymptotically stable limit can be done in two steps. First, we introduce a $g$-function of genus $G$ with branch points $\lam_1 > \lam_2 > ... > \lam_{2G+2}$ by making the global change of variable $m \mapsto M$ 
\begin{equation}\label{gtrans}
	m(z) = e^{ -i g(\infty) \sig/ \eps} M(z) e^{i g(z) \sig/ \eps},
\end{equation}
which seeks to remove rapid oscillations from the problem. 
Second, we introduce steepest descent contours $\Gamma_i$, $i=1,2$ in $\C^+$ and their complex conjugate images $\Gamma_i^*$ in $\C^-$ in order to deform the jumps onto contours on which they are near identity. The exact shape of these contours is determined by the given $g$-function, but in each case $\Gamma_1$ lies to the right of $\Gamma_2$ and each returns to the real axis at exactly one point, which may or may not be distinct. This divides $\C^+$ (and $\C^-$) into three regions which we label from right-to-left as $\Omega_i, \ i=1,2,3$ (and $\Omega_i^*,\ i=1,2,3$). Using these regions we make the piecewise-analytic transformation $M \mapsto N$ defined by 
\begin{equation}\label{Lenses}
	M(z) = 
	\begin{cases}
		N(z) \tril{r(z) e^{2i(\varphi(z) + \theta(\lam_1))/\eps}} & z \in \Omega_1 \smallskip \\
		N(z) \triu{r^*(z) e^{-2i( \varphi(z) + \theta(\lam_1))/\eps}} & z \in \Omega_1^* \\
		N(z) & z \in \Omega_2 \cup \Omega_2^* \\
		N(z) \triu{ \frac{-r^*(z)}{1- r(z) r^*(z)} e^{-2i (\varphi(z)+ \theta(\lam_1))/\eps} } & z \in \Omega_3 \smallskip \\
		N(z) \tril{ \frac{-r(z) }{1- r(z) r^*(z)}e^{2i( \varphi(z)+ \theta(\lam_1))/\eps} } & z \in \Omega_3^*  
	\end{cases}
\end{equation}
the new unknown $N$ has jumps on the real axis and on each of the $\Gamma_i$'s. 

\subsection{The far left field: $\tau < -1$}

We expect that for large negative $\tau$, that is $x \ll -t$, the solution should resemble the plane wave specified by the left half of the initial data \eqref{0.2}. At the level of the RHP this means that we expect that the $g$-function should be cut on $\I _L = (-1,1)$ with two hard edges. Using the results of Section~\ref{sec:onecutg} we define the $g$-function 
\begin{gather}
	g(z) = \int_{1}^z d\theta - 2t \frac{(\lambda-\xi_-)(\lambda - \xi_+)}{\RR(\lambda; -1,1)} d\lambda
\intertext{where the stationary phase points are given by}
	\xi_\pm = \xi_\pm(\tau) = -\frac{\tau}{4} \pm \frac{1}{4} \sqrt{\tau^2 + 8}.
\end{gather}
and analytic for $z \in \C \backslash \I_L$.
	
For $\tau \leq -1$, the stationary points satisfy $\xi_+ \geq 1$ with equality only when $\tau = -1$; 
for each $\tau \leq -1$ the other stationary point $\xi_- \in (-1,0)$. As such the imaginary sign table for the function 
\begin{equation}\label{4.1phase}
	\varphi(z) =  2t \int_1^z   \frac{(\lambda-\xi_-)(\lambda - \xi_+)}{\RR(\lambda; -1,1)} d\lambda 
\end{equation}
looks like Figure~\ref{fig:1}(a). We open lens along the steepest descent paths through $\xi_+(\tau)$ as depicted in Figure~\ref{fig:3} and define the mapping from $m \mapsto N$ using \eqref{gtrans}-\eqref{Lenses}. The result is the following problem for the new unknown $N(z)$:
\begin{rhp}\label{rhp:N left}
Find a $2\times2$ matrix $N$ with the following properties
\begin{enumerate}[1.]
	\item $N(z)$ is analytic in $\C \backslash \Gamma_N$, 
	$\Gamma_N = (-\infty, \xi_+(\tau)) \bigcup_{i=1}^2 (\Gamma_i \cup \Gamma_i^*)$.
	
	\item $N(z) = I + \bigo{z^{-1} }$ as $z\to \infty$.
	
	\item $N(z)$ takes continuous boundary values on $\Gamma_N$ away from points of self intersection and branch points which satisfy the jump relation $N_+(z) = N_-(z) V_N(z)$ where
	\begin{equation}\label{Ncase1_1}
		V_N(z) = 
		\begin{cases}
			( 1- r(z) r^*(z) )^{\sig} & z \in (-\infty, \xi_+(\tau)) \backslash \I_L \\
			\offdiag{ -r_-^*(z) e^{-2i \theta(1)/\eps} }{ r_+(z) e^{2i\theta(1)/\eps} }
				& z \in \I_L \backslash \I_R \smallskip \\
			\offdiag{ e^{-2i \theta(1)/\eps} }{ e^{2i \theta(1)/\eps} } & z \in \I_R \cap \I_L	 \smallskip  \\
			\tril{r(z) e^{2i(\varphi(z) + \theta(1))/\eps}} & z \in \Gamma_1 \smallskip \\
			\triu{ \frac{ -r^*(z)}{1 - r(z) r^*(z) } e^{-2i(\varphi(z)+\theta(1))/\eps} } & z \in \Gamma_2
		\end{cases}
	\end{equation}

	\item $N(z)$ is bounded except at the points $\{ 1, -1, \lambda_+, \lambda_- \}$ where
	\begin{equation}
	\begin{aligned}	
		N(z) &= \left\{ 
		\begin{aligned}
			&\bigo{ \begin{matrix} 1 & (z-p)^{-1/2} \\ 1 & (z-p)^{-1/2} \end{matrix} } , \quad z \in \Omega_3 \\ 
			&\bigo{ \begin{matrix} (z-p)^{-1/2} & 1 \\  (z-p)^{-1/2} & 1\end{matrix} } , \quad z \in \Omega_3^* 
		\end{aligned}  
		\right. \qquad p \in \{-1,1\} \\
		N(z) &= \left\{ 
		\begin{aligned}
			&\bigo{ \begin{matrix} 
			(z-p)^{1/4} & (z-p)^{-1/4} \\ (z-p)^{1/4} & (z-p)^{-1/4} \end{matrix} }, \quad z \in \Omega_3 \\
			&\bigo{ \begin{matrix}
			(z-p)^{-1/4} & (z-p)^{1/4} \\ (z-p)^{-1/4} & (z-p)^{1/4} \end{matrix} }, \quad z\in \Omega_3^*
		\end{aligned}
		\right. \qquad p \in \{\lambda_-, \lambda_+ \}
	\end{aligned}
	\end{equation}
	
\end{enumerate}
\end{rhp}

\begin{rem}
	Throughout this section we give the jumps of the various Riemann-Hilbert problems only on the real axis and in the upper half-plane. The contours deformations we use all respect the original symmetry $m(z;x,t) = \sigma_2 m(z^*;x,t)^* \sigma_2$ of RHP~\ref{rhp:1}. It follows that the jump along a contours $\Gamma_k^* \in \C^-$ is given by $\sigma_2 v^*(z^*;x,t) \sigma_2$ where $v(z;x,t)$ is the jump defined along $\Gamma_k \in \C^+$.
\end{rem}

\begin{figure}[th]
	\begin{center}
		\includegraphics[width=.6\textwidth]{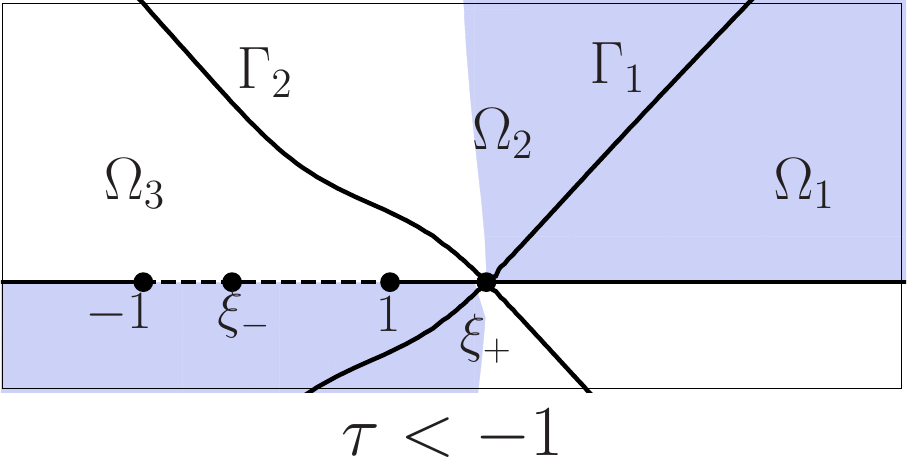}		
		\caption{The contours $\Gamma_i$ and regions $\Omega_i$ used to define the 
		map $M \mapsto N$ (c.f. \eqref{Lenses}) for $x/t = \tau$ in the left planar zone (defined above). As $\tau$ increases the stationary phase points $\xi_\pm(\tau)$ decrease, at $\tau = -1$, the boundary of the zone, $\xi_+$ collides with $1$;  $\xi_-$ lies within $(-1,1)$ for all $\tau$ in the zone. Blue regions correspond to $\imag \varphi>0$ and white regions to $\imag \varphi<0$. 
		\label{fig:3}
		}
	\end{center}
\end{figure}	
		
\subsubsection{Constructing a parametrix for $\tau < -1$}
The jumps of $N(z)$ along $\Gamma_i,\ i=1,2$ and their c.c's are all near identity at any positive distance from the real axis because the contours lie in regions in which the off diagonal entries are exponentially decaying. As a result, to leading order the solution $N(z)$ should be given by the model problem produced by neglecting the jumps off the real axis in \eqref{Ncase1_1}. 

Define
 \begin{multline}\label{Ddef}
	D(z) = \exp \lb  \frac{i \theta(1)}{\eps} + \frac{\RR(z; -1,1)}{2\pi i} \left\{ 
	 \lp \int_{-\infty}^{-1}+\int_{1}^{\xi_+}  \rp \frac{ \log( 1 - r(\lambda) r^*(\lambda))}{\RR(\lambda; -1,1)}
	 \frac{d \lambda}{\lambda -z } \right. \right. \\ + \left. \left. \lp \int_{ -1}^{\lambda_-} + \int_{\lambda_+}^1 \rp
	 \frac{  \log( r_+(\lambda)) }{\RR_+(\lambda; -1,1)}
	 \frac{d \lambda}{\lambda -z } 
	 \right\}
	 \rb
\end{multline}
As the following proposition describes, this function is constructed to remove the jumps along the real axis, or reduce to constants where they cannot be removed. Simultaneously, the growth behavior at the branch points is simplified.

\begin{prop}\label{prop:D}
The function $D: \C \backslash (-\infty, \xi_+) \to \C$ defined by \eqref{Ddef} has the following properties:
\begin{enumerate}[1.]
	\item $D$ is analytic in $ \C \backslash (-\infty, \xi_+)$, and takes continuous boundary values 
	on $(-\infty, \xi_+)$ except at the endpoints of integration in \eqref{Ddef}.
	
	\item As $z \to \infty$, $D(z) \to D(\infty)  + \bigo{z^{-1}}$ where
	\begin{equation}
		D(\infty) = e^{ i \theta(1)/\eps} 
		e^{\left[ -\frac{1}{2\pi i} \lp 
		 \lp \int_{-\infty}^{-1} + \int_1^{\xi_+} \rp
		 \frac{ \log( 1 - r(\lambda) r^*(\lambda))}{\RR(\lambda; -1,1)} d\lambda
		+ \lp \int_{-1}^{\lambda_-}+\int_{\lambda_+}^1 \rp  
		\frac{  \log( r_+(\lambda)) }{\RR_+(\lambda; -1,1)} d\lambda \rp \right]}
	\end{equation}
	
	\item For $z \in (-\infty, \xi_+(\tau))$, $D(z)$ satisfies the jump relations
	$$
	\begin{cases}
	D_+(z) / D_-(z) = 1 -r(z) r^*(z)
		& z \in (-\infty, \xi_+(\tau)) \backslash \I_L \\
	D_+(z) D_-(z) = r_+(z) e^{2i \theta(1) /\eps} & z \in \I_L \backslash \I_R \\
	D_+(z) D_-(z) = e^{2i \theta(1) /\eps} & z \in \I_L \cap \I_R
	\end{cases}
	$$
	
	\item $D(z)$ exhibits the following singular behavior at each endpoint of integration:
\begin{equation}
	\begin{cases}
		D(z) = (z - p)^{ \frac{1}{4} \sgn \imag z} D_0(z)  &  z \to p \\
		D(z) = (z - \xi_+)^{i \kappa(\xi_+)} D_0(z) & z \to \xi_+
	\end{cases}
\end{equation}
where $p \in \{ -1, 1, \lambda_-, \lambda_+\}$ is any of the four branch points, $\kappa(z) = -\frac{1}{2\pi} \log (1 - r(z) r^*(z))$, and $D_0(z)$ is a bounded function taking a definite limit as z approaches each singular point non-tangentially.
\end{enumerate}
\end{prop}

\begin{proof} 
Each of these properties follows immediately from the general properties of Cauchy-type integrals and the local behavior of $r$ and $1-r r^*$ at the endpoints of integration which can be read off from \eqref{0.7} and \eqref{0.9c}. For the behavior at the endpoints of integration the standard reference is \cite{Musk}.
\end{proof} 

Using the function $D(z)$ the change of variables
\begin{equation}\label{Qdef}
	Q(z) = D(\infty)^{\sig} N(z) D(z)^{-\sig}
\end{equation}
results in the following RHP for $Q$. 

\begin{rhp}{for $Q$:}\label{rhp:Q1} 
Find a $2\times2$ matrix $Q$ with the following poroperties
\begin{enumerate}[1.]
	\item $Q(z)$ is analytic in $\C \backslash \Gamma_Q$, 
	$\Gamma_Q = (-1, 1) \bigcup_{i=1}^2 (\Gamma_i \cup \Gamma_i^*)$.
	
	\item $Q(z) = I + \bigo{z^{-1} }$ as $z\to \infty$.
	
	\item $Q(z)$ takes continuous boundary values on $\Gamma_N$ away from endpoints and points of self intersection satisfying the jump relation $N_+(z) = N_-(z) V_Q(z)$ where
	\begin{equation}\label{Qcase1_1}
		V_Q(z) = 
		\begin{cases}
			\offdiag{ -1 }{ 1 } & z \in (-1,1) \smallskip \\
			\tril{r(z) D^{-2}(z) e^{2i(\varphi(z) + \theta(1))/\eps}} & z \in \Gamma_1 \smallskip \\
			\triu{ \frac{ -r^*(z) D^2(z) }{1 - r(z) r^*(z) } e^{-2i(\varphi(z)+\theta(1))/\eps} } & z \in \Gamma_2
		\end{cases}
	\end{equation}

	\item $Q(z)$ is bounded except at the points $\{ 1, -1 \}$ where it admits 1/4-root singularities in each entry.
	
\end{enumerate}
\end{rhp}

The jumps of $Q(z)$ off the real axis converge pointwise to the identity, and the limiting problem on the real axis has a simple solution. Using the small-norm theory for RHPs we can prove that the solution $Q(z)$ of RHP \ref{rhp:Q1} exists and takes the form
\begin{equation}\label{ErrDef}
	Q(z) = \begin{cases}
		E(z) P_{\xi_+}(z) & z \in \U_{\xi_+} \\
		E(z) P_\infty(z) & \text{elsewhere}
	\end{cases}
\end{equation}

The outer model $P_\infty(z)$ is the solution of the limiting problem on the real axis given by
$$
	P_\infty(z) = \Ecal(z;-1,1),
$$ 
where $\Ecal$, defined by \eqref{0.4b}, is related to the exact plane wave solution of the ZS scattering problem for the initial data produced by extending the left side of \eqref{0.2} to the entire real line. 

The outer model is a uniform approximation of $Q(z)$ except for inside a small neighborhood of $\xi_+$ where the contours $\Gamma_i$ return to the real axis. As the local behavior of the jumps is Gaussian, a local model $P_{\xi_+}$ can be constructed from parabolic cylinder functions. The construction is standard and the details are omitted, see for example the appendix to \cite{JM11}. 
The crucial fact is that the the resulting RHP for $E(z)$ has jumps which are uniformly small everywhere in the complex plane with the largest contribution coming from the boundary $\partial \U_{\xi_+}$. Small norm theory guarantees the existence of $E(z)$ and its asymptotic expansion can be computed. Once this is done, the series of transformations from $m(z)$ to $Q(z)$ can be inverted to produce the asymptotic expansion of the original problem $m(z)$. From this the leading order behavior of the solution of \eqref{0.1}-\eqref{0.2} for $\tau< -1$ is given by
\begin{equation}\label{far left field}
	\begin{gathered}
	\psi(x,t) = e^{-i t /\eps} e^{-i \phi(x/t) } + \bigo{ \sqrt{ \frac{\eps}{t}}\log \frac{\eps}{t} } \\
	\phi(\tau) = \exp \Bigg[ \frac{1}{\pi}  \Bigg(
	 \int\limits_{-\infty}^{-1} + \int\limits_1^{\xi_+(\tau)} \Bigg) \frac{ \log( 1 - r(z) r^*(z))}{\sqrt{\lambda^2-1}}
	d \lambda 
	 + \frac{1}{\pi} \int_{ \I_L  \backslash \I_R}  \frac{  \arg( r_+(\lambda)) }{\sqrt{1-\lambda^2}}
	d \lambda 
	 \Bigg]
\end{gathered}
\end{equation}

\subsection{Rarefaction zone: $-1 < \tau < -\frac{1}{2} \lp -1 +3 \lambda_+ \rp$}

\begin{figure}[th]
	\begin{center}		
		\includegraphics[width=.6\textwidth]{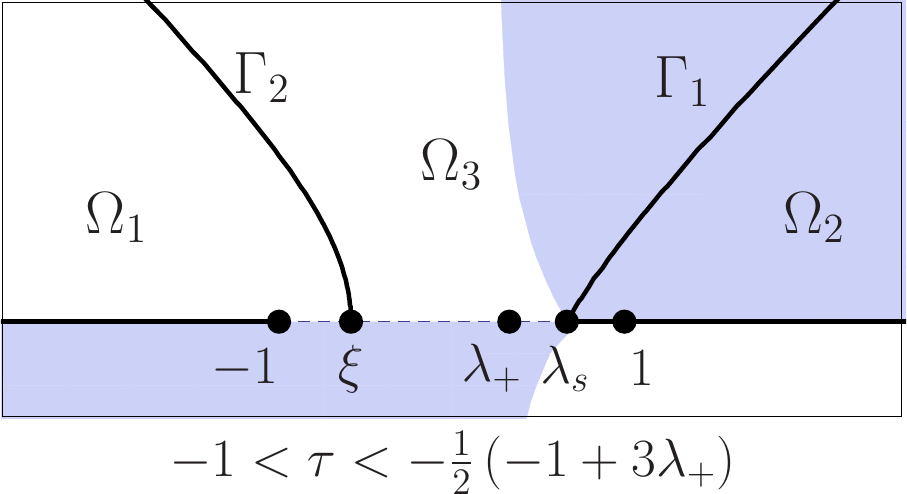}
		\caption{The regions $\Omega_k$ and contours $\Gamma_k$ used to define the transformation $M \mapsto N$ for $x/t = \tau$ in the rarefaction zone (defined above). As $x/t = \tau$ increases across the zone, $\xi(\tau)$ and $\lam_s(\tau)$ move to the right. The limits of the rarefaction zone are characterized by the soft edge $\lam_s$ colliding with $1$ and $\lambda_+$.
	Blue regions correspond to $\imag \varphi >0$ and white regions $\imag \varphi < 0$.
		\label{fig:rarefaction}
		}
	\end{center}
\end{figure}

As $\tau$ increases beyond $-1$ the stationary phase point $\xi_+(\tau)$ of the far left field phase function \eqref{4.1phase} moves inside $\I_L$ at $z=1$. When this happens, the previous factorization \eqref{Ncase1_1} creates an exponentially large jumps on the interval $(\xi_+, 1)$. So, for $\tau > -1$ we introduce a new $g$-function with a single cut $(-1, \lam_s)$ whose soft edge $\lam_s$ satisfies $\lam_s(\tau = -1) =1$. Using the results of Section~\ref{sec:onecutg_soft}, define
\begin{equation}\label{3.5}
	g(z) = \int_{\lam_s}^z d \theta - 2t \sqrt{ \frac{\lambda - \lam_s}{\lambda+1} } (\lambda - \xi) d \lambda
\end{equation}
analytic for $z \in \C \backslash (-1, \lam_s)$ where 
\begin{equation}\label{3.6}
	\begin{gathered}
		\lam_s(\tau) = -\frac{1}{3} \lp 2\tau - 1 \rp, \qquad \qquad \xi(\tau) = -\frac{1}{6} \lp 4 + \tau \rp, 
		\\
		g(\infty) = - \theta(\lam_s) + \frac{t}{6} \lp 2 - 2\tau - \tau^2 \rp 
	\end{gathered}
\end{equation}
Over the interval $-1 \leq \tau \leq -\frac{1}{2}(3\lambda_+ -1 )$, the soft edge $\lam_s(\tau)$ decreases linearly from $1$ to $\lambda_+$ and the stationary phase point $\xi(\tau)$ decreases linearly from $-1/2$ to $(\lambda_+-3)/4$. For each $\tau$ in this interval $-1 < \xi(\tau) < \lam_s(\tau) < 1$.

The modified phase function
\begin{equation}\label{3.7}
	\varphi(z) =  2t  \int_{\lam_s}^z  \sqrt{ \frac{\lambda -\lam_s}{\lambda+1} } (\lambda - \xi) d \lambda = t(z-\lam_s)^{3/2} (z+1)^{1/2}
\end{equation}
has an imaginary sign table of the form given in Figure~\ref{fig:2}b. We open lenses along the steepest descent paths through $\lam_s$ and $\xi$ which define the contours $\Gamma_i$ and regions $\Omega_i$,  see Figure~\ref{fig:rarefaction}. The resulting problem for $N(z)$ defined by \eqref{gtrans}-\eqref{Lenses} is as follows. 

\begin{rhp}{for $N$:} Find a $2\times2$ matrix $N$ with the following properties
\begin{enumerate}[1.]
	\item $N(z)$ is analytic in $\C \backslash \Gamma_N$, 
	$\Gamma_N = (-\infty, \lam_s) \bigcup_{i=1}^2 (\Gamma_i \cup \Gamma_i^*)$.
	\item $N(z) = I + \bigo{z^{-1} }$ as $z\to \infty$.
	\item $N(z)$ takes continuous boundary values on $\Gamma_N$ away from points of self intersection and branch points which satisfy the jump relation $N_+(z) = N_-(z) V_N(z)$ where
	\begin{equation}\label{Ncase1_2}
		V_N(z) = 
		\begin{cases}
			(1 - r(z)r^*(z) )^{\sig} & z \in (-\infty, -1)  \\
			T(z) & z \in (-1, \lambda_+) \\
			\tril{r(z) e^{2i(\varphi(z)+\theta(\lam_s))/\eps}} & z \in \Gamma_1 \smallskip \\ 
			\triu{ \frac{ -r^*(z)}{1 - r(z) r^*(z) } e^{-2i(\varphi(z)+\theta(\lam_s))/\eps} } & z \in \Gamma_2
		\end{cases}
	\end{equation}
	\item $N(z)$ is bounded except at the points $z = p,\, p \in \{\lambda_+, \lambda_- , -1\}$ where
	\begin{equation}\label{Nbounds2}
	\begin{aligned}	
		N(z) &= \left\{ 
		\begin{aligned}
			&\bigo{ \begin{matrix} 1 & (z+1)^{-1/2} \\ 1 & (z+1)^{-1/2} \end{matrix} } , \quad z \in \C^+ \\ 
			&\bigo{ \begin{matrix} (z+1)^{-1/2} & 1 \\  (z+1)^{-1/2} & 1\end{matrix} } , \quad z \in \C^- 
		\end{aligned}  
		\right. \\
		N(z) &= \left\{ 
		\begin{aligned}
			&\bigo{ \begin{matrix} 
			(z-p)^{1/4} & (z-p)^{-1/4} \\ (z-p)^{1/4} & (z-p)^{-1/4} \end{matrix} }, \quad z \in \C^+ \\
			&\bigo{ \begin{matrix}
			(z-p)^{-1/4} & (z-p)^{1/4} \\ (z-p)^{-1/4} & (z-p)^{1/4} \end{matrix} }, \quad z\in \C^-
		\end{aligned}
		\right. \qquad p \in \{\lambda_-, \lambda_+ \}
	\end{aligned}
	\end{equation}
	
	The precise form of the jump $T(z)$ in \eqref{Ncase1_2} depends on the position of $\xi = \xi(\tau)$ relative to $\lambda_\pm$:
	\begin{equation}
		T(z) = \begin{cases}
			\offdiag{ -r_-^*(z) e^{-2i \theta(\lam_s)/\eps} }
			{ r_+(z) e^{2i\theta(\lam_s)/\eps} } 
			& z \in \lp (-1, \lambda_-) \cup (\lambda_+, \lam_s) \rp \cap \{ z < \xi \} \smallskip \\
			\begin{pmatrix}	 
				0 &  -r^*_-(z)e^{-2i \theta(\lam_s)/\eps} \\ 
				r_+(z) e^{2i \theta(\lam_s)/\eps} & e^{-2i\varphi_+(z)} 
			\end{pmatrix} 
			& z \in \lp (-1, \lambda_-) \cup (\lambda_+, \lam_s) \rp \cap \{ z > \xi \} \smallskip \\
			\offdiag{ e^{-2i \theta(\lam_s)/\eps} }{ e^{2i \theta(\lam_s)/\eps} } & z \in (\lambda_-, \lambda_+) 
			\end{cases}
		\end{equation}
\end{enumerate}
\end{rhp}
\subsubsection{Rarefaction parametrix}
The jump matrices of the RHP for $N(z)$ take well defined asymptotic limits whose values are independent of the ordering of $\xi(\tau)$ and $\lambda_-$. 
The jumps off the real axis approach identity pointwise, and along the real axis the jumps take well defined limits, up to phase constants depending on $\eps$. As before, we first introduce a scalar function $D(z)$ which simplifies the limiting problem by reducing the limiting problem to one with constant jumps. Define
\begin{multline}\label{3.10}
	D(z) = \exp \lb \frac{i \theta(\lam_s)}{\eps} + \frac{\RR(z; -1,\lam_s)}{2\pi i} \lp 
	\int_{(-\infty, -1) } \frac{  \log(1 -r(\lambda) r^*(\lambda) )}{\RR(\lambda; -1,\lam_s)}
	 \frac{d \lambda}{\lambda -z } \right. \right.
	\\ + \left. \left. \int_{(-1,\lambda_-) \cup (\lambda_+, \lam_s)}  \frac{  \log( r_+(\lambda)) }{\RR_+(\lambda; -1,\lam_s)}
	 \frac{d \lambda}{\lambda -z } 
	 \rp
	 \rb.
\end{multline}  

\begin{prop}\label{prop:D2}
The function $D: \C \backslash (-\infty, \lam_s) \to \C$ defined by \eqref{3.10} has the following properties:
\begin{enumerate}[1.]
	\item $D$ is analytic in $ \C \backslash (-\infty, \lam_s)$, and takes continuous boundary values 
	on $(-\infty, \lam_s)$ except at the endpoints of integration in \eqref{Ddef}.
	
	\item As $z \to \infty$, $D(z) \to D(\infty)  + \bigo{z^{-1}}$ where
	\begin{equation}
		D(\infty) = e^{ i \theta(\lam_s)/\eps} 
		e^{\left[ -\frac{1}{2\pi i} \lp 
		 \int_{-\infty}^{-1} \frac{ \log( 1 - r(\lambda) r^*(\lambda))}{\RR(\lambda; -1,1)} d\lambda
		+ \lp \int_{-1}^{\lambda_-}+\int_{\lambda_+}^{\lam_s} \rp  
		\frac{  \log( r_+(\lambda)) }{\RR_+(\lambda; -1,1)} d\lambda \rp \right]}
	\end{equation}
	
	\item For $z \in (-\infty, \lam_s)$, $D(z)$ satisfies the jump relations
	$$
	\begin{cases}
	D_+(z) / D_-(z) = 1 -r(z) r^*(z)
		& z \in (-\infty, -1)\\
	D_+(z) D_-(z) = r_+(z) e^{2i \theta(\lam_s) /\eps} & z \in (-1, \lambda_-) \cup (\lambda_+, \lam_s) \\
	D_+(z) D_-(z) = e^{2i \theta(\lam_s) /\eps} & z \in (\lambda_-, \lambda+)
	\end{cases}
	$$
	
	\item $D(z)$ exhibits the following singular behavior at each endpoint of integration:
\begin{equation}
	\begin{cases}
		D(z) = (z - p)^{ \frac{1}{4} \sgn \imag z} D_0(z)  &  z \to p \\
	\end{cases}
\end{equation}
where $p \in \{ -1, \lambda_-, \lambda_+\}$ and $D_0(z)$ is a bounded function taking a definite limit as z approaches each point non-tangentially.
\end{enumerate}
\end{prop}

Using $D(z)$, the change of variables
\begin{equation}\label{Qdef2}
	D(\infty)^\sig Q(z) D(z)^{-\sig}
\end{equation}
results in the following RHP for Q:

\begin{rhp}{for $Q$:}\label{rhp:Q2}
Find a $2\times2$ matrix $Q$ with the following poroperties
\begin{enumerate}[1.]
	\item $Q(z)$ is analytic in $\C \backslash \Gamma_Q$, 
	$\Gamma_Q = (-1, \lam_s) \bigcup_{i=1}^2 (\Gamma_i \cup \Gamma_i^*)$.
	
	\item $Q(z) = I + \bigo{z^{-1} }$ as $z\to \infty$.

	\item $Q(z)$ takes continuous boundary values on $\Gamma_N$ away from endpoints and points of self intersection satisfying the jump relation $N_+(z) = N_-(z) V_Q(z)$ where
	\begin{equation}\label{Ncase1_2}
		V_Q(z) = 
		\begin{cases}
			\offdiag{-1}{1} & z \in (\lambda_-, \lambda_+) \smallskip \\
			\offdiag{-1}{1} & z \in \lp (-1, \lambda_-) \cup (\lambda_+, \lam_s) \rp \cap \{ z < \xi \} \smallskip\\
			\begin{pmatrix} 0 & -1 \\ 1 & \frac{D_+(z)}{D_-(z)} e^{-2 i \varphi_+(z)/\eps} \end{pmatrix}
			& z \in \lp (-1, \lambda_-) \cup (\lambda_+, \lam_s) \rp \cap \{ z > \xi \}  \smallskip \\
			\tril{r(z) D^{-2}(z) e^{2i(\varphi(z) + \theta(1))/\eps}} & z \in \Gamma_1 \smallskip \\
			\triu{ \frac{ -r^*(z) D^2(z) }{1 - r(z) r^*(z) } e^{-2i(\varphi(z)+\theta(1))/\eps} } & z \in \Gamma_2.
		\end{cases}
	\end{equation}
		
	\item $Q(z)$ is bounded except at the points $z=-1$ where it admits 1/4-root singularities in each entry.
	
\end{enumerate}
\end{rhp}

The jumps of $Q(z)$ off the real axis converge pointwise to identity, and on the real axis the jump of $Q(z)$ is either constant, or uniformly exponentially close to the same constant. Using the small norm theory for RHPs we can prove that the solution $Q(z)$ of RHP \ref{rhp:Q2} exists and takes the form 
\begin{equation}\label{ErrDef2}
	Q(z) = \begin{cases}
		E(z) P_{\lam_s}(z) & z \in \U_{\lam_s} \\
		E(z) P_\infty(z) & \text{elsewhere}
	\end{cases}
\end{equation}

The outer model $P_\infty(z)$ is the solution of the limiting problem on the real axis given by
$$
	P_\infty(z) = \Ecal(z;-1,\lam_s),
$$ 
where $\Ecal$, defined by \eqref{0.7}, is related to the Jost functions for the plane wave initial data whose (scaled) Riemann invariants are -1 and the linearly evolving $\lam_s = \lam_s(\tau)$ given by \eqref{3.6}. The outer model is a uniform approximation of $Q(z)$ except for a small neighborhood of $\lam_s$ where the contours $\Gamma_i$ return to the real axis. The local $3/2$-vanishing indicates that the local model $P_{\lam_s}$ should be constructed from Airy functions. The construction is standard \cite{DKMVZ} and the details are omitted. 
The crucial fact is that the the resulting RHP for $E(z)$ has jumps which are uniformly small everywhere in the complex plane with the largest contribution coming from the boundary $\partial \U_{\lam_s}$. Small norm theory guarantees the existence of $E(z)$ and its asymptotic expansion can be computed. 

Once this is done, the series of transformations from $m(z)$ to $Q(z)$ can be inverted to produce the asymptotic expansion of the original problem $m(z)$. From this the leading order behavior of the solution of \eqref{0.1}-\eqref{0.2} for $-1< \tau< \frac{1}{2} (1-3\lambda_+)$ is given by

\begin{equation}\label{3.11}
	\begin{gathered}
	\psi(x,t) = \lp \frac{2-\tau}{3} \rp  e^{-it(2- 2\tau - \tau^2)/3\eps}  e^{ - i \phi(x/t) } + \bigo{ \eps } \\
	\phi(\tau) = 
	\frac{1}{\pi} \lp \int_{-\infty}^{-1}  \frac{\log(1 -r(\lambda) r^*(\lambda) )}{\sqrt{(\lambda+1)(\lambda-\lam_s)}} d\lambda +
	\int_{(-1,\lambda_-) \cup (\lambda_+, \lam_s)}  \frac{ \arg( r_+(\lambda)) }{\sqrt{(\lambda+1)(\lam_s-\lambda)}} d\lambda \rp
	\end{gathered}
\end{equation}

\subsection{The central plateau: $-\frac{1}{2} \lp -1 + 3\lambda_+ \rp < \tau < -\frac{1}{2} \lp -1 +  \lambda_+ + 2 \lambda_- \rp$ }\ \\
For $\tau = -\frac{1}{2} \lp -1 + 3\lambda_+ \rp$ the soft edge $\lam_s$ defined by \eqref{3.6} of the rarefaction $g$-function \eqref{3.5} collides with $\lambda_+$, the upper boundary of $\I_R$.  
If $\lam_s< \lambda_+$ then the factorization \eqref{Lenses} leaves a non-vanishing component in the (1,1)-entry of $V_N$ on $(\lam_s, \lambda_+)$ which is exponentially large. The $g$-function must be modified to account for this. For $\tau > -\frac{1}{2} \lp -1 + 3\lambda_+ \rp$ we use the results of Section~\ref{sec:onecutg_hard} to define a $g$-function, with a single fixed cut $(-1, \lam_+)$: 
\begin{equation}\label{3.12}
	g(z) = \int_{\lambda_+}^z d\theta - 2t \frac{(\lambda- \xi_-)(\lambda - \xi_+)}
	{\RR(\lambda;-1, \lambda_+)} d \lambda
	= \theta \big|_{\lambda_+}^z  - t \RR(z,-1,\lambda_+)(z-\xi_0),
\end{equation}
where 
\begin{equation}\label{3.13}
	\begin{gathered}
	\xi_0 = -\frac{1}{2} \lp -1 + \lambda_+ \rp - \tau \\
	\xi_\pm = \xi_\pm(\tau) = \frac{1}{4} \lp -1 + \lambda_+ - \tau \pm \sqrt{ (-1+\lambda_+ + \tau)^2 + 2	(\lambda_+ + 1)^2} \rp.	
	\end{gathered}
\end{equation}	
are ordered such that -$1 < \xi_- < \xi_0  < \xi_+ < \lambda_+$ for $-\frac{1}{2} \lp -1 + 3\lambda_+ \rp < \tau < -\frac{1}{2} \lp -1 +  \lambda_+ + 2 \lambda_- \rp$. As such, both of the stationary points of the modified phase function
\begin{equation}\label{3.14}
	\varphi(z) = 2t \int_{\lambda_+}^z  \frac{(\lambda- \xi_-)(\lambda - \xi_+)}
	{\RR(\lambda;-1, \lambda_+)} d \lambda = - t \RR(z;-1,\lambda_+)(z-\xi_0)
\end{equation}
lie on its branch cut and the transition point for the signature of $\imag \varphi$ occurs at $\xi_0$ which lies between them, see Figure~\ref{fig:4}. The lens contours $\Gamma_i$ used to define \eqref{Lenses} are taken as the steepest descent contours through $\xi_\pm$. The contours $\Gamma_i$ and corresponding regions $\Omega_i$ are as depicted in Figure~\ref{fig:4}. 

\begin{figure}[t]
	\begin{center}
		\includegraphics[width = .6\textwidth]{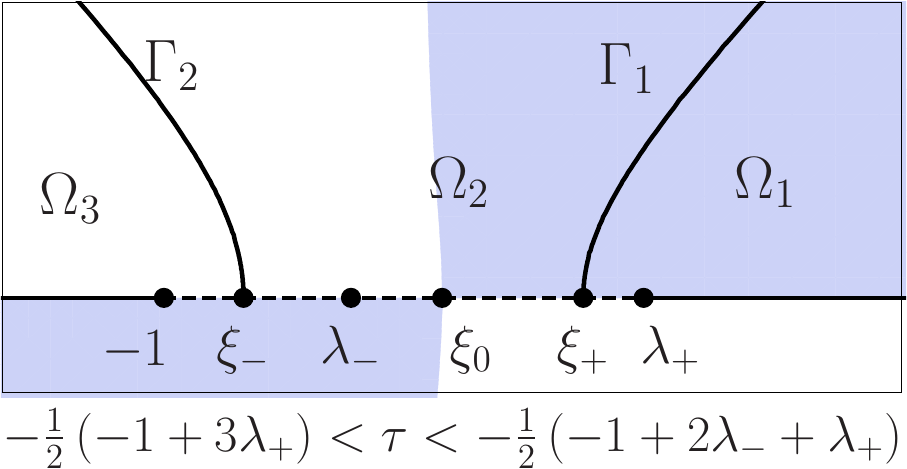}
		\caption{The contours $\Gamma_i$ and regions $\Omega_i$ used to define the 
		map $M \mapsto N$ (c.f. \eqref{Lenses}) for $x/t = \tau$ in the central plateau (defined above).
		As $\tau$ varies across the region, $-1< \xi_-(\tau)< \xi_0(\tau) < \xi+(\tau)< \lam_+$ 
		are each decreasing. The lower bound on $\tau$ in this region is characterized 
		by the collision $\xi_+(\tau) = \lam_+$ and the upper bound by $\xi_0(\tau) = \lam_-$. 
		The lesser stationary phase point $\xi_-(\tau)$ may lie on either side of $\lambda_-$ 
		for allowed values of $\tau$. Blue regions correspond to $\imag \varphi >0$ 
		and white regions $\imag \varphi < 0$.
		\label{fig:4}
		}
	\end{center}
\end{figure}

The result of \eqref{gtrans}-\eqref{Lenses} using \eqref{3.12} is the following RHP for $N(z)$:

\begin{rhp} Find a $2\times2$ matrix-valued function $N$ with the following properties
\begin{enumerate}[1.]
	\item $N(z)$ is analytic in $\C \backslash \Gamma_N$, 
	$\Gamma_N = (-\infty, \lambda_+) \bigcup_{i=1}^2 (\Gamma_i \cup \Gamma_i^*)$.
	\item $N(z) = I + \bigo{z^{-1} }$ as $z\to \infty$.
	\item $N(z)$ takes continuous boundary values on $\Gamma_N$ away from points of self intersection and branch points which satisfy the jump relation $N_+(z) = N_-(z) V_N(z)$ where
	\begin{align}\label{3.15}
		V_N(z) = 
		\begin{cases}
			(1 - r(z) r^*(z) )^{\sig} & z \in (-\infty, -1)  \\
			T(z)	& z \in (-1, \xi_+) \\
			\begin{pmatrix} 0 & -e^{-2i \theta(\lambda_+)/\eps} \\
		 	e^{2i \theta(\lambda_+)/\eps} & 0 
			\end{pmatrix} & z \in (\xi_+,  \lambda_+) \smallskip \\
			\tril{r(z) e^{2i (\varphi(z)+ \theta(\lambda_+))/\eps}} & z \in \Gamma_1 \smallskip \\
			\triu{ \frac{ -r^*(z)}{1 - r(z) r^*(z) } e^{-2i (\varphi(z) + \theta(\lambda_+))/\eps} } & z \in \Gamma_2
		\end{cases}
	\end{align}
	
	\item $N(z)$ is bounded except at the points $z = \{-1,\lambda_-, \lambda_+\}$ where the local growth bound at each point are given by
\begin{equation}\label{Nbounds3}
	\begin{aligned}	
		N(z) &= \left\{ 
		\begin{aligned}
			&\bigo{ \begin{matrix} 1 & (z+1)^{-1/2} \\ 1 & (z+1)^{-1/2} \end{matrix} } , \quad z \in \C^+ \\ 
			&\bigo{ \begin{matrix} (z+1)^{-1/2} & 1 \\  (z+1)^{-1/2} & 1\end{matrix} } , \quad z \in \C^- 
		\end{aligned}  
		\right. \\
		N(z) &= \left\{ 
		\begin{aligned}
			&\bigo{ \begin{matrix} 
			(z-\lambda_-)^{1/4} & (z-\lambda_-)^{-1/4} \\ 
			(z-\lambda_-)^{1/4} & (z-\lambda_-)^{-1/4} \end{matrix} }, \quad z \in \C^+ \\
			&\bigo{ \begin{matrix}
			(z-\lambda_-)^{-1/4} & (z-\lambda_-)^{1/4} \\ 
			(z-\lambda_-)^{-1/4} & (z-\lambda_-)^{1/4} \end{matrix} }, \quad z\in \C^-
		\end{aligned}
		\right. \\
		N(z) & = \bigo{\begin{matrix} 
			(z-\lambda_+)^{-1/4} & (z-\lambda_+)^{-1/4} \\ 
			(z-\lambda_+)^{-1/4} & (z-\lambda_+)^{-1/4}
			\end{matrix}}
	\end{aligned}
	\end{equation}
\end{enumerate}
\end{rhp}

$T(z)$ is one of the following sets of twist matrices, which depends on the ordering of $\xi_-$ and $\lambda_-$:

If $\xi_- > \lambda_-$ then
\begin{subequations}\label{Tdef}
\begin{equation}\label{Tdefa}
	T(z) =
	\begin{cases}
		\offdiag{ -r_-^*(z) e^{-2i \theta(\lambda_+)/\eps} }{ r_+(z) e^{2i\theta(\lambda_+)/\eps} }
			& z \in (-1, \lambda_-) \smallskip \\
		\offdiag{ -e^{-2i \theta(\lambda_+)/\eps} }{ e^{2i \theta(\lambda_+)/\eps} } 
		& z \in (\lambda_-, \xi_+),	
	\end{cases}	
\end{equation}
or if $\xi_- < \lambda_-$, then
\begin{equation}\label{Tdefb}
	T(z) = 
	\begin{cases}
		\offdiag{ -r_-^*(z) e^{-2i \theta(\lambda_+)/\eps} }{ r_+(z) e^{2i\theta(\lambda_+)/\eps} }
			& z \in (-1, \xi_-) \smallskip \\
		\begin{pmatrix}  0 & -r_-^*(z) e^{-2i \theta(\lambda_+) / \eps} \\ 
		r_+(z) e^{2i \theta(\lambda_+) / \eps} & e^{-2i \varphi_+(z)}
		\end{pmatrix} & z \in (\xi_-, \lambda_-) \smallskip \\
		\offdiag{ -e^{-2i \theta(\lambda_+)/\eps} }{ e^{2i \theta(\lambda_+)/\eps} } 
		& z \in (\lambda_-, \xi_+).	
	\end{cases}
\end{equation}
\end{subequations}

Examining $T(z)$, $N(z)$ has near identity jump matrices only if $\xi_0 > \lambda_-$; if $\xi_0 < \lambda_-$, then on the segment $(\xi_0, \lambda_-) \subset (\xi_-, \lambda_-)$ the jump \eqref{Tdefb} is exponentially large in the $(2,2)$-entry. This defines the upper boundary of the central plateau region: the upper boundary is the unique $\tau$ such that $ \xi_0 = \lambda_-$:
$$
	\tau = -\frac{1}{2} \lp -1 +  \lambda_+ + 2 \lambda_- \rp 
	\quad \Longleftrightarrow \quad 
	\xi_0(\tau) = \lambda_-,
$$	
Provided that $\xi_0 > \lambda_-$, \ie, $-\frac{1}{2} \lp -1 + 3\lambda_+ \rp < \tau < -\frac{1}{2} \lp -1 +  \lambda_+ + 2 \lambda_- \rp$, the limiting value of the jump matrices of $N$ are the same in all cases
$$
	V_N(z) \sim 
	\begin{cases}
		(1 - r(z) r^*(z))^{\sig} & z \in (-\infty, -1) \\ 
		\offdiag{ -r_-^*(z) e^{-2i \theta(\lambda_+)/\eps} }{ r_+(z) e^{2i\theta(\lambda_+)/\eps} }
			& z \in (-1, \lambda_-) \smallskip \\
		\offdiag{ -e^{-2i \theta(\lambda_+)/\eps} }{ e^{2i \theta(\lambda_+)/\eps} } 
		& z \in (\lambda_-, \lambda_+).
	\end{cases}	
$$	

\subsubsection{Constructing the parametrix in the central plateau}
As before the RHP for $N(z)$ has a well defined asymptotic limit--independent of the ordering of $\xi_-(\tau)$ and $\lambda_-$. 
The jumps off the real axis approach identity pointwise, and the jumps along the real axis take well defined limits, up to phase constants depending on $\eps$. Again we introduce a scalar function $D(z)$ which reduces the limiting problem to one with constant jumps. Define
\begin{multline}\label{3.10}
	D(z) = \exp \lb \frac{i \theta(\lambda_+)}{\eps} + \frac{\RR(z; -1, \lambda_+)}{2\pi i} \lp 
	\int_{(-\infty, -1) } \frac{  \log(1 -r(\lambda) r^*(\lambda) )}{\RR(\lambda; -1, \lambda_+)}
	 \frac{d \lambda}{\lambda -z } \right. \right.
	\\ + \left. \left. \int_{(-1,\lambda_-)}  \frac{  \log( r_+(\lambda)) }{\RR_+(\lambda; -1,\lambda_+)}
	 \frac{d \lambda}{\lambda -z } 
	 \rp
	 \rb.
\end{multline}  

\begin{prop}\label{prop:D3}
The function $D: \C \backslash (-\infty, \lambda_+) \to \C$ defined by \eqref{3.10} has the following properties:
\begin{enumerate}[1.]
	\item $D$ is analytic in $ \C \backslash (-\infty, \lambda_+)$, and takes continuous boundary values 
	on $(-\infty, \lambda_+)$ except at the endpoints of integration in \eqref{Ddef}.
	
	\item As $z \to \infty$, $D(z) \to D(\infty)  + \bigo{z^{-1}}$ where
	\begin{equation}
		D(\infty) = e^{ i \theta(\lambda_+)/\eps} 
		e^{\left[ -\frac{1}{2\pi i} \lp 
		 \int_{-\infty}^{-1} \frac{ \log( 1 - r(\lambda) r^*(\lambda))}{\RR(\lambda; -1,\lambda_+)} d\lambda
		+ \int_{-1}^{\lambda_-}
		\frac{  \log( r_+(\lambda)) }{\RR_+(\lambda; -1,\lambda_+)} d\lambda \rp \right]}
	\end{equation}
	
	\item For $z \in (-\infty, \lambda_+)$, $D(z)$ satisfies the jump relations
	$$
	\begin{cases}
	D_+(z) / D_-(z) = 1 -r(z) r^*(z)
		& z \in (-\infty, -1)\\
	D_+(z) D_-(z) = r_+(z) e^{2i \theta(\lambda_+) /\eps} & z \in (-1, \lambda_-) \\
	D_+(z) D_-(z) = e^{2i \theta(\lambda_+) /\eps} & z \in (\lambda_-, \lambda+)
	\end{cases}
	$$
	
	\item $D(z)$ exhibits the following singular behavior at each endpoint of integration:
\begin{equation}
	\begin{cases}
		D(z) = (z - p)^{ \frac{1}{4} \sgn \imag z} D_0(z)  &  z \to p \\
	\end{cases}
\end{equation}
where $p \in \{ -1, \lambda_- \}$ and $D_0(z)$ is a bounded function taking a definite limit as z approaches each point non-tangentially.
\end{enumerate}
\end{prop}

Using $D(z)$, the change of variables
\begin{equation}\label{Qdef3}
	D(\infty)^\sig Q(z) D(z)^{-\sig}
\end{equation}
results in the following RHP for Q:

\begin{rhp}{for $Q$:}\label{rhp:Q3}
Find a $2\times2$ matrix $Q$ with the following poroperties
\begin{enumerate}[1.]
	\item $Q(z)$ is analytic in $\C \backslash \Gamma_Q$, 
	$\Gamma_Q = (-1, \lambda_+) \bigcup_{i=1}^2 (\Gamma_i \cup \Gamma_i^*)$.
	
	\item $Q(z) = I + \bigo{z^{-1} }$ as $z\to \infty$.
	
	\item $Q(z)$ takes continuous boundary values on $\Gamma_N$ away from endpoints and points of self intersection satisfying the jump relation $N_+(z) = N_-(z) V_Q(z)$ where
	\begin{equation}\label{Qplateau}
		V_Q(z) = 
		\begin{cases}
			\begin{pmatrix} 0 & -1 \\ 1 & \one_{\xi_-(\tau) < z < \lambda_-} \frac{D_+(z)}{D_-(z)} e^{-2 i \varphi_+(z)/\eps} \end{pmatrix}
			& z \in (-1, \lambda_+) \smallskip \\
			\tril{r(z) D^{-2}(z) e^{2i(\varphi(z) + \theta(1))/\eps}} & z \in \Gamma_1 \smallskip \\
			\triu{ \frac{ -r^*(z) D^2(z) }{1 - r(z) r^*(z) } e^{-2i(\varphi(z)+\theta(1))/\eps} } & z \in \Gamma_2.
		\end{cases}
	\end{equation}
	Here, $\one_{a<z<b}$ is the indicator function of the set $(a,b)$. If $b<a$ than this is the indicator of the empty set and the function is identically zero.
	
	\item $Q(z)$ is bounded except at the points $z=-1, \lambda_+$ where it admits 1/4-root singularities in each entry.
	
\end{enumerate}
\end{rhp}

The jumps of $Q(z)$ off the real axis converge pointwise to identity, and the limiting problem on the real axis, regardless of the position of $\xi_-(\tau)$ relative to $\lambda_-$, is uniformly exponentially near a constant twist. Using the small norm theory for RHPs we can prove that the solution $Q(z)$ of RHP \ref{rhp:Q3} exists and takes the form 
\begin{equation}\label{ErrDef3}
	Q(z) = E(z) P_\infty(z) 
\end{equation}
The outer model $P_\infty(z)$ is the solution of the limiting problem on the real axis given by
$$
	P_\infty(z) = \Ecal(z; -1, \lambda_+),
$$ 
where $\Ecal$, defined by \eqref{0.4b}, is related to the solution of the ZS system \eqref{0.3} for a plane wave potential $\psi(x)$ whose Riemann invariants \eqref{0-phase riemann invariants} are -1 and $\lambda_+$.


The outer model in this case is
\begin{equation}\label{3.16}
	P^\infty(z) = D(\infty)^{-\sig}  \Ecal(z; -1, \lambda_+)  D(z)^\sig
\end{equation}	
where now 
\begin{multline}\label{3.17}
	D(z) = \exp \lb \frac{i \theta(\lambda_+)}{\eps} + \frac{\RR(z; -1,\lambda_+)}{2\pi i} \lp 
	\int_{(-\infty, -1) } \frac{  \log( 1 - r(\lambda)r^*(\lambda))}{\RR(\lambda; -1,\lambda_+)}
	 \frac{d \lambda}{\lambda -z } \right. \right.
	\\ + \left. \left. \int_{(-1,\lambda_-) }  \frac{  \log( r_+(\lambda)) }{\RR_+(\lambda; -1,\lambda_+)}
	 \frac{d \lambda}{\lambda -z } 
	 \rp
	 \rb.
\end{multline}

The outer model is uniformly accurate for each fixed $\tau \in  \lp -\frac{1}{2} \lp -1 + 3\lambda_+ \rp,  -\frac{1}{2} \lp -1 +  \lambda_+ + 2 \lambda_- \rp \rp$, in the entire complex plane, provided that $\xi_+(\tau) < \lambda_+ -\delta$ and $\xi_0 > \lambda_- + \delta$ for any fixed constant $\delta>0$.

The resulting behavior of the solution of \eqref{0.1}-\eqref{0.2} is 
\begin{equation}\label{3.24}
	\begin{gathered}
	\psi(x,t) = \lp \frac{\lambda_+ + 1}{2} \rp  e^{-i (k x - \omega t)/\eps}  e^{ - i \phi_0 } + \bigo{ e^{-ct/\eps} } \\
	 k = \lambda_+ -1 \qquad \omega = -\frac{1}{2} (\lambda_+ -1)^2 - \frac{1}{4} \lp \lambda_+ + 1 \rp^2 \\
	\phi_0 = 
	\frac{1}{\pi} \lp \int_{-\infty}^{-1}  \frac{\log(1 -r(\lambda) r^*(\lambda) )}{\sqrt{(\lambda+1)(\lambda-\lambda_+)}} d\lambda +
	\int_{-1}^{\lambda_-} \frac{ \arg( r_+(\lambda)) }{\sqrt{(\lambda_+ - \lambda)(\lambda+1)}} d\lambda \rp
	\end{gathered}
\end{equation}

\subsection{The modulation zone: \\
{ $  -\frac{1}{2} \lp -1 +  \lambda_+ + 2 \lambda_- \rp < \tau <  
-\frac{1}{2} \lp \lambda_+  + \lambda_- - 2 \rp 
+ \frac{2(1+\lambda_-)(1+\lambda_+)} {\lambda_+ + \lambda_- + 2} $} }\  \\

When $\tau$ increases beyond $-\frac{1}{2}(-1 + 2\lambda_- + \lambda_+)$ the point $\xi_0(\tau)$ (defined by \eqref{3.13}) lies to the right of $\lambda_-$ this makes the (2,2) entry of the jump $V_N$ defined by \eqref{3.15} exponentially large on the interval $(\xi_0, \lambda_-)$. To arrive at a stable limit problem we modify the $g$-function to include a gap interval below $\lambda_-$, with a soft upper edge. Following Section~\ref{sec:twocutg} define the $g$-function, analytic for $z \in \C \backslash \lp (-1, \lam_s) \cup (\lambda_-, \lambda_+) \rp $:
\begin{equation}\label{3.25}
	g(z) = \int_{\lambda_+}^z d\theta - 2t \frac{(\lambda- \lam_s)(\lambda- \xi_-)(\lambda-\xi_+) }{ \RR(\lambda; -1, \lam_s, \lambda_-, \lambda_+)} d \lambda.
\end{equation}
The motion of the soft edge $\lam_s = \lam_s(x/t)$ is given by the self-similar solution of the Whitham equations:
\begin{equation}\label{3.26}
	\begin{gathered}
	\frac{x}{t} = V_3(\lam_+, \lam_-, \lam_s,-1) = 
	-\frac{1}{2} (-1 + \lam_s + \lambda_- + \lambda_+) 
	- \frac{ \lam_s + 1}{1 - \frac{ \lambda_- +1}{\lambda_- - \lam_s} \frac{ E(m)}{K(m)} } \\
	m = \frac{ (\lambda_+ - \lambda_-)(\lam_s+1)}{ (\lambda_+ - \lam_s)(\lambda_- + 1)} 
	\end{gathered}
\end{equation}
where $K(m)$ and $E(m)$ are the complete elliptic integrals of the first and second kind respectively.
The above equation is solvable for each $\lam_s \in (-1, \lambda_-)$. Using \eqref{3.26} it's easy to verify the two-band solution degenerates when:
\begin{equation}
	\begin{aligned}
		\lam_s \to  \lambda_- &, \qquad \tau \to -\frac{1}{2} \lp \lambda_+ + 2\lambda_- - 1 \rp, \\
		\lam_s \to  -1 &, \qquad \tau \to -\frac{1}{2} \lp \lambda_+ + \lambda_- - 2 \rp 
		+ \frac{2(\lambda_+ + 1)(\lambda_- + 1)}{\lambda_+ + \lambda_- +2}, 
	\end{aligned}
\end{equation}
which define the transition from the modulation zone to the plane wave zones which it separates. The two stationary phase points $\xi_-(\tau)$, and $\xi_+(\tau)$ which lie one in each band can be computed from \eqref{2.31}. 

The phase function
\begin{equation}
	\varphi(z) = 2t \int_{\lambda_+}^z 
	\frac{ ( \lambda - \lam_s)(\lambda - \xi_-)(\lambda-\xi_+)}
	{\RR(\lambda; -1, \lam_s, \lambda_-, \lambda_+)} d \lambda 
\end{equation}
is analytic in $\C \backslash (-1, \lam_s) \cup (\lambda_-, \lambda_+)$ and satisifies the jump relation
\begin{equation}
	\varphi_+(z) + \varphi_-(z) = \begin{cases}
		0 & z \in (\lambda_-, \lambda_+) \\
		\gamma = 2(\varphi_+(\lam_s) - \varphi(\lambda_-)) & z \in (-1, \lam_s)
	\end{cases}
\end{equation}

\begin{figure}[t]
	\begin{center}
		\includegraphics[width=.8\textwidth]{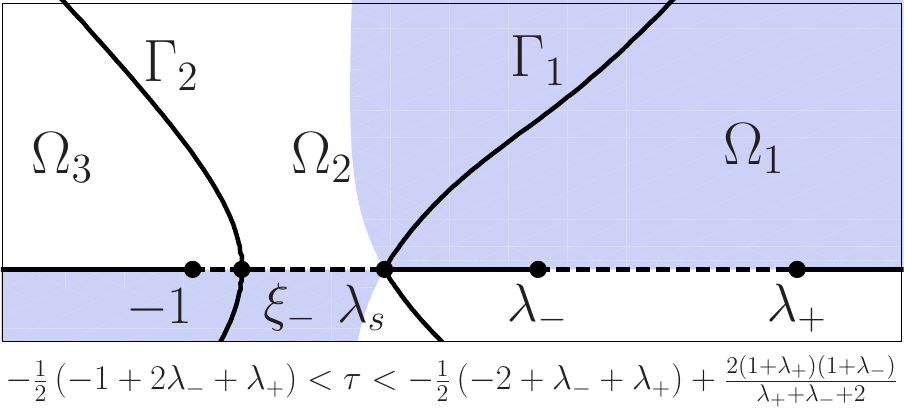}
		\caption{The contours $\Gamma_i$ and regions $\Omega_i$ used to define the map 
		$M \mapsto N$ (c.f. \eqref{Lenses}) for $\tau=x/t$ in the modulation zone (given above). 
		As $\tau = x/t$ varies across the modulation zone the soft edge $\lam_s(\tau)$ 
		decreases according to the Whitham evolution \eqref{3.26}. At the lower and upper bounds 
		of the modulation zone $\lambda_s(\tau)$ collides with $\lam_-$ and $-1$ respectively. 
		The lower collision ($\lam_s = \lam_-$) is the soliton limit of the modulated 
		 wavefront and the upper collision ($\lam_s = -1$) is the zero amplitude limit.
		Blue regions correspond to $\imag \varphi >0$ and white regions $\imag \varphi < 0$.
		\label{fig:twoband_contours}
		}
	\end{center}	
\end{figure}

The structure of the zero level set of $\imag \varphi$ resembles that in Figure~\ref{fig:twoband_signs}. In order to define the mapping from $m \mapsto N$ given by \eqref{gtrans}-\eqref{Lenses} in the two band case we open lenses from $\xi_-(\tau)$ (opening to the left) and from $\lam_s(\tau)$ (opening to the right) as shown in Figure~\ref{fig:twoband_contours}. The result of \eqref{gtrans}-\eqref{Lenses} with $g$ given by \eqref{3.25} is the following RHP for N(z):

\begin{rhp} Find a $2\times2$ matrix-valued function $N$ with the following properties
\begin{enumerate}[1.]
	\item $N(z)$ is analytic in $\C \backslash \Gamma_N$, 
	$\Gamma_N = (-\infty, \lam_s(\tau)) \cup (\lambda_-, \lambda_+) \bigcup_{i=1}^2 (\Gamma_i \cup \Gamma_i^*)$.
	\item $N(z) = I + \bigo{z^{-1} }$ as $z\to \infty$.
	\item $N(z)$ takes continuous boundary values on $\Gamma_N$ away from points of self intersection and branch points which satisfy the jump relation $N_+(z) = N_-(z) V_N(z)$ where
	\begin{align}\label{3.30}
		V_N(z) = 
		\begin{cases}
			(1 - r(z) r^*(z) )^{\sig} & z \in (-\infty, -1 )  \\
			\begin{pmatrix}
				0 & -r^*(z) e^{-i \gamma/\eps} e^{-2i\theta(\lambda_+)/\eps} \\
				r(z) e^{i \gamma/\eps} e^{2i\theta(\lambda_+)/\eps} &
				0
			\end{pmatrix} & z \in (-1,\xi_-(\tau)) \smallskip \\
			\begin{pmatrix}
				0 & -r^*(z) e^{-i \gamma/\eps} e^{-2i\theta(\lambda_+)/\eps} \\
				r(z) e^{i \gamma/\eps} e^{2i\theta(\lambda_+)/\eps} &
				e^{i\gamma/\eps} e^{-2i \varphi_+/\eps} 
			\end{pmatrix} & z \in (\xi_-(\tau), \lam_s(\tau)) \smallskip \\
			\offdiag{ -e^{-2i \theta(\lambda_+)/\eps} }{ e^{2i \theta(\lambda_+)/\eps} } 
				& z \in (\lambda_-, \lambda_+) \\
			\tril{r(z) e^{2i\theta(\lambda_+)/\eps} e^{2i\varphi(z)/\eps} } 
			& z \in \Gamma_1 \smallskip \\
			\triu{ \frac{ -r^*(z)}{1 - r(z) r^*(z) } e^{-2i \theta(\lambda_+)/\eps} e^{-2i\varphi(z)/\eps} } 
			& z \in \Gamma_2
		\end{cases}
	\end{align}

	\item $N(z)$ is bounded except at the point $z = -1, \lam_-, \lam_+$  where
	\begin{equation}
		\begin{aligned}
		N(z) &= 
		\bigo{ \begin{matrix} 1 & (z+1)^{-1/2} \\ 1 & (z+1)^{-1/2} \end{matrix} } , \quad z \in \Omega_3 \\
		N(z) &= 
		\bigo{ \begin{matrix} (z+1)^{-1/2} & 1 \\  (z+1)^{-1/2} & 1\end{matrix} } , \quad z \in \Omega_3^* \\
		N(z) & = \bigo{\begin{matrix} 
			(z-\lambda_\pm)^{-1/4} & (z-\lambda_\pm)^{-1/4} \\ 
			(z-\lambda_\pm)^{-1/4} & (z-\lambda_\pm)^{-1/4}
		\end{matrix}
		}
		\end{aligned}
	\end{equation}	
\end{enumerate}
\end{rhp}

\subsubsection{Constructing the parametrix in the modulation zone}
In the long-time/small dispersions limit, the jumps of $N(z)$ along the real axis have well defined limits up to $\eps$-dependent constants, while the jumps on the non-real contours approach identity uniformly at any distance from $\lam_s$ (the convergence at $\xi_-$ is uniform provided $\lam_s$ and -1 are well separated): 
\begin{align}\label{3.32.0}
		V_{N}(z) \underset{\eps \to 0}{\sim}
		\begin{cases}
			(1 - r(z) r^*(z) )^{\sig} & z \in (-\infty, -1 )  \\
			\offdiag{-r^*(z) e^{-i \gamma/\eps} e^{-2i\theta(\lambda_+)/\eps} }
				{r(z) e^{i \gamma/\eps} e^{2i\theta(\lambda_+)/\eps}} 
				& z \in (-1,\lam_s(\tau)) \smallskip \\
			\offdiag{ -e^{-2i \theta(\lambda_+)/\eps} }{ e^{2i \theta(\lambda_+)/\eps} } 
				& z \in (\lambda_-, \lambda_+). \\
		\end{cases}
\end{align}
In order to build a uniformly accurate parametrix, we introduce a scalar function $D(z)$ which reduces this limiting problem to one with constant jumps. Define

\begin{equation}\label{3.34.0}
\begin{gathered}		
	D(z) = D_0(z) D_1(z) \bigskip \\
	D_0(z) =  \exp \lb -\frac{\pi i}{4} +
	\frac{ \RR(z, \vect \lam)}{2i\pi} \lp \int_{-\infty}^{-1} \frac{ \log(1-r(s)r^*(s)) }
	{\RR(s, \vect \lam)} \frac{ds}{s-z} +
	\int_{-1}^{\lam_s} \frac{ \log r_+(s) }{\RR_+(s, \vect \lam)}\frac{ds}{s-z} \rp \rb,
	\smallskip \\
	D_1(z) = \exp \lb \frac{i \theta(\lambda_+)}{\eps} +
	\frac{ \RR(z, \vect \lam)}{2i\pi} 
	\int_{-1}^{\lam_s} \frac{i\gamma/\eps}{\RR_+(s, \vect \lam)}\frac{ds}{s-z}  \rb.
	\end{gathered}
\end{equation}
Here $\vect \lam = \{ -1, \lam_s, \lambda_-, \lambda_+ \}$ are the branch points of $\RR(z, \vect \lam)$ and $g(z)$. 

\begin{prop}\label{prop:D4}
The function $D: \C \backslash (-\infty, \lam_s) \to \C$ defined by \eqref{3.10} has the following properties:
\begin{enumerate}[1.]
	\item $D$ is analytic in $ \C \backslash (-\infty, \lam_s)$, and takes continuous boundary values 
	on $(-\infty, \lam_s)$ except at the endpoints of integration in \eqref{Ddef}.
	
	\item As $z \to \infty$, $D(z) \to D(\infty)\lb 1  + \bigo{z^{-1}} \rb $ where
	\begin{equation}\label{3.34.0}
		D^\infty(z) = e^{-i\pi/4} e^{i\theta(\lambda_+)/\eps} e^{i( \phi_0(z) + \eps^{-1} \phi_1(z) )} \\
	\end{equation}	
	and $\phi_k(z)$, $k=0,1$ are the linear functions 
	\begin{equation}\label{3.37.0}
		\begin{gathered}
		\phi_{0}(z) = \sum_{j=0}^1 
		\frac{z^j}{2\pi} \int_{-\infty}^{\lam_s} \frac{ \lp w +V \rp^j 
		\lp \log( 1 - |r(w)|^2) \one_{(-\infty, -1)} + \log r_+(w) \one_{(-1, \lam_s)}\rp}
		{\RR_+(w, \vect \lam)} dw, \\
		\phi_1(z) = \sum_{j=0}^1 
		\frac{z^j}{2\pi} \int_{-1}^{\lam_s} \frac{ i\gamma \lp w +V \rp^j }
		{\RR_+(w, \vect \lam)} dw. 
		\end{gathered}
	\end{equation}
	where $V = - \frac{e_1(\vect \lam)}{2} = -\frac{1}{2} \sum_{j=1}^4 \lam_j$.
	
	Note, that as $|r_+(s)| = 1$ and $\re (\RR_+(s, \vect \lam)) = 0$ for $s \in (-1, \lam_s)$, each 
	$\phi_k(z)$ is a real (linear) polynomial. 
	
	\item For $z \in (-\infty, \lam_s)$, $D(z)$ satisfies the jump relations
	\begin{equation}\label{3.33.0}
		\begin{array}{rl@{\quad}l}
			D_+(z) / D_-(z) &=  1 - r(z) r^*(z) & z \in (-\infty, -1) \smallskip \\
			D_+(z) D_-(z) &= -i r_+(z) e^{i\gamma/\eps} e^{2i\theta(\lambda_+)/\eps} 
				& z \in (-1, \lam_s) \smallskip \\
			D_+(z) D_-(z) &= -i e^{2i\theta(\lambda_+)/\eps} & z \in (\lambda_-, \lambda_+)
		\end{array}
	\end{equation}	
	
	\item $D(z)$ exhibits the following singular behavior at each endpoint of integration:
	\begin{equation}
	\begin{cases}
		D(z) = (z +1)^{ \frac{1}{4} \sgn \imag z} D_0(z)  &  z \to -1 \\
		D(z) = D_0(z) & z \to \lam_s
	\end{cases}
\end{equation}
where in each case $D_0(z)$ is a (different) bounded function taking a definite limit as z approaches each point non-tangentially.
\end{enumerate}
\end{prop}

We also introduce 
\begin{equation}
	\alpha(z) = \lp \frac{z-\lambda_+}{z-\lambda_-} \rp^{1/4} \lp \frac {z - \lam_s}{z+1} \rp^{1/4}
\end{equation}
branched along $(-1, \lam_s)$ and $(\lambda_-, \lambda_+)$ and normalized such that $\alpha(z) \sim 1$ as $z \to \infty$
to define the transformation 
\begin{equation}\label{3.38.0}
	N(z) = \alpha(z) Q(z) D(z)^\sig
\end{equation}
then $Q$ must satisfy the following constant jump RHP:
\begin{rhp}\label{rhp: normalized two band} 
Find a $2 \times 2$ matrix valued function $Q(z)$ such that
\begin{enumerate}[1.]
	\item $Q(z)$ is analytic in $\C \backslash (-1, \lam_s) \cup (\lambda_-, \lambda_+) \bigcup_{i=1}^2 (\Gamma_i \cup \Gamma_i^*)$.
	\item $Q(z) D^\infty(z) ^{-\sig} = I + \bigo{z^{-1}}$ as $z \to \infty$.
	\item $Q(z)$ takes continuous boundary values on $\Gamma_Q$away from the 
	points of self intersection and endpoints, which satisfy the jump relation 
	$Q_+(z) = Q_-(z) V_Q(z)$ where
	\begin{equation}\label{Qmodulation}
		V_Q(z) = 
		\begin{cases}
			\sigma_1 & z \in (-1, \xi_-(\tau)) \cup (\lambda_-, \lambda_+) \smallskip \\
			\begin{pmatrix} 0 & 1 \\ 
			1 & -i e^{i\gamma/\eps} \frac{D_+(z)}{D_-(z)} e^{-2 i \varphi_+(z)/\eps} 
			\end{pmatrix}
			& z \in (\xi_-(\tau), \lam_s(\tau)) \smallskip \\
			\tril{r(z) D^{-2}(z) e^{2i(\varphi(z) + \theta(1))/\eps}} & z \in \Gamma_1 \smallskip \\
			\triu{ \frac{ -r^*(z) D^2(z) }{1 - r(z) r^*(z) } e^{-2i(\varphi(z)+\theta(1))/\eps} } & z \in \Gamma_2.
		\end{cases}
	\end{equation}
	\item $Q(z)$ admits $1/4$-root singularities at $-1, \lambda_-, \lambda_+$.
\end{enumerate}
\end{rhp}

The jump matrix for $Q(z)$ converge pointwise to identity away from the real axis, and to constants on the real axis. The convergence is uniform away from the soft edge $\lam_s$, where the lens contours return to the real axis. We take $\U_{\lam_s}$ a local neighborhood of $\lam_s$ and build local and outer parametrices $P_{\lam_s}$ and $P_\infty$ respectively so that the relation
\begin{equation}\label{Q4}
	Q(z) = \begin{cases}
		E(z) P_{\lam_s}(z) & z \in \U_{\lam_s} \\
		E(z) P_\infty(z) & \text{elsewhere}
	\end{cases}
\end{equation}
results in a residual problem for $E(z)$ which can be proven to exist and asymptotically expanded using the small-norm theory for RHPs. 

To built the outer solution, we replace the jump condition \eqref{Qmodulation} in RHP~\ref{rhp: normalized two band} with $  P_{\infty}(z)_+ =  P_{\infty} (z)_- \sigma_1$ for $z \in (-1, \lam_s) \cap (\lambda_-, \lambda_+)$ and admit $1/4$-root singularities at each endpoint.
The solution of such a multi-cut problem is constructed from theta functions on the hyperelliptic Riemann surface associated with $\RR(z; \vect \lam)$. The construction is standard, so we will provide only the necessary formula to define the solution. 

Let $\vect \lam = (\lambda_+, \lambda_-, \lam_s,-1)$ denote the moduli of the genus-one Riemann surface
$$
	\Scal_1 := \left\{ P = (z, \RR), \ \RR^2 = \prod_{i=1}^4 (z-\lam_i) \right\}
$$
and fix the homology basis as in Figure~\ref{fig:homology}. Define the holomorphic differential
\begin{equation}
	\nu(z,\RR) = c_\nu \frac{dz}{\RR}  ,\qquad c_\nu = 
	\frac{ i \sqrt{(\lam_1- \lam_3)(\lam_2 - \lam_4)}}{4 K(m)},
\end{equation}
normalized so that 
\begin{equation}
	\oint_a \nu = 2c_\nu \int_{\lam_2}^{\lam_1} \frac{dz}{\RR_+(z,\vect \lam) }= 1.
\end{equation}
Then we also have
\begin{equation}
	\tau := \oint_b \nu = 2c_\nu \int_{\lam_2}^{\lam_3} \frac{dz}{\RR(z,\vect \lam) } = i \frac{K(1-m)}{K(m)},
	\quad
	m = \frac{(\lam_1 - \lam_2)(\lam_3-\lam_4)}{(\lam_1-\lam_3)(\lam_2-\lam_4)}	
\end{equation}
where $K(m)$ denotes the complete elliptic integral of the first kind with parameter $m$.
Using these quantities, define the Siegel theta function
 \begin{equation}
 	\Theta(z) = \Theta(z;\tau) = \sum_{n \in \Z} e^{2\pi i(n z + \frac{1}{2}n^2 \tau)}  
	= \theta_3 \lp \pi z , e^{i \pi \tau} \rp 
 \end{equation}
Here, $\theta_3(z,q)$ is the standard Jacobi theta function with nome $q$. Note that $\Theta(z)$ is a quasi-doubly periodic function satisfying:
\begin{equation}\label{Theta periods}
	\Theta(z + 1) = \Theta(z) , \qquad \Theta(z+\tau) = \Theta(z) e^{-2i \pi  z} e^{-i\pi \tau/2},
\end{equation}
and vanishes at the lattice of half periods:
\begin{equation}
	\Theta(z) = 0, \qquad z = \frac{1}{2} + \frac{\tau}{2} + \Z + \tau \Z
\end{equation}

Let $\Abel(z)$ denote the restriction of the standard Abel map to the complex plane:
 \begin{equation}
 	\Abel(z) = \int_{\lam_1}^z \nu = \int_{\lam_1}^z \frac{c_\nu}{\RR(z, \vect \lam)} dz
 \end{equation}
 where the path of integration lies in $\C \backslash \lp [\lam_4, \lam_3] \cup [\lam_2, \lam_1] \rp$. 

We also need the following normalized differential of the second kind
\begin{equation}
	\begin{gathered}
		\upsilon = \upsilon_0 + \eps^{-1} \upsilon_1 \\
		\upsilon_k = \dot \phi_k \omega^{(0)}, \qquad k=0,1
	\end{gathered}
\end{equation}
where $\dot \phi_k$ is the coefficient of the linear term of $\phi_k(z)$ given by \eqref{3.37.0} and $\omega^{(0)}$ is the normalized differential of the second kind defined by \eqref{2.8.0}.	
Let $\Upsilon$ be the $b$-period of this differential
\begin{equation}
	\Upsilon = \Upsilon_0 + \eps^{-1} \Upsilon_1, \qquad 
	\Upsilon_k =\oint_b \upsilon_k = 4\pi i c_\nu \dot \phi_k. 
 \end{equation}	
The purpose of this differential is to cancel the behavior of $D(z)$ at infinity. 
Define $\chi$ by the relation
\begin{equation}
	\chi = \chi_0 + \eps^{-1} \chi_1, \qquad
	\chi_k = -i \log \lp \lim_{z\to \infty} D_k(z) e^{-i \int_{\lam_1}^z \upsilon_k} \rp.
\end{equation}
Clearly, both $\Omega$ and $\chi$ are real quantities. 

The outer model $P_\infty(z)$ approximating the solution of RHP~\ref{rhp: normalized two band} away from $\lam_s$ is given by 	
\begin{equation}
	P_\infty(z) = \frac{ \Theta(0)}{\Theta(\frac{\Upsilon}{2\pi})} e^{-i \chi \sig} 
	\begin{pmatrix}
	\frac{ 1 + \alpha^{-2}}{2} 
	\frac{\Theta( \Abel(z) - \Abel(\infty) - \frac{\Upsilon}{2\pi})}{\Theta( \Abel(z) - \Abel(\infty))} &
	\frac{ 1 - \alpha^{-2}}{2} 
	\frac{\Theta( \Abel(z) + \Abel(\infty) + \frac{\Upsilon}{2\pi})}{\Theta( \Abel(z) + \Abel(\infty))} \\
	\frac{ 1 - \alpha^{-2}}{2} 
	\frac{\Theta( \Abel(z) + \Abel(\infty) - \frac{\Upsilon}{2\pi})}{\Theta( \Abel(z) + \Abel(\infty))} &
	\frac{ 1 + \alpha^{-2}}{2} 
	\frac{\Theta( \Abel(z) - \Abel(\infty) + \frac{\Upsilon}{2\pi})}{\Theta( \Abel(z) - \Abel(\infty))} 
	\end{pmatrix}
	e^{-i\lp \int_{\lambda_+}^z \upsilon \rp\sig}
\end{equation}

At first glance, it seems the outer model depends in a complicated way on the asymptotic parameter. However, it is a simple calculation to show that 
$$
	d( \log D_1) = d \lp \frac{\RR(z, \vect \lam)}{2\pi i} \int_{-1}^{\lam_s} 
	\frac{ i\gamma}{\RR_+(w; \vect \lam)} \frac{dw}{w-z} \rp = \upsilon_1
$$
and as such it follows that 
\begin{equation}
	\chi_1 = \theta(\lambda_+) \qquad
	\Upsilon_1 = \gamma = 4\pi i c_\nu (x + \frac{1}{2} e_1(\vect \lam) t)
\end{equation}
where the last equality comes from explicit computation, by identifying $\gamma = \oint_b d\varphi$ and making use of the Riemann bilinear relations.

Putting all the parts together the matrix $Q(z)$ for large $z$ is given by 
\begin{equation}
	Q(z) = E(z) \frac{ \Theta(0)}{\Theta(\frac{\Upsilon_0+\eps^{-1} \gamma}{2\pi})} e^{-i (\chi_0 + \eps^{-1} \theta(\lambda_+)) \sig} 
	\mathcal{T}(z)
	\lp D_0(z)e^{-i\int_{\lambda_+}^z \upsilon_0}  e^{\frac{i \theta(\lambda_+)}{\eps}} \rp^\sig
\end{equation}
where
$$
	\mathcal{T}(z)= 
	\begin{pmatrix}
	\frac{ \alpha(z) + \alpha^{-1}}{2} 
	\frac{\Theta( \Abel(z) - \Abel(\infty) - \frac{\Upsilon_0+\eps^{-1}\gamma}{2\pi})}{\Theta( \Abel(z) - \Abel(\infty))} &
	\frac{ \alpha(z) - \alpha^{-1}}{2} 
	\frac{\Theta( \Abel(z) + \Abel(\infty) + \frac{\Upsilon_0+\eps^{-1}\gamma}{2\pi})}{\Theta( \Abel(z) + \Abel(\infty))} \\
	\frac{ \alpha(z) - \alpha^{-1}}{2} 
	\frac{\Theta( \Abel(z) + \Abel(\infty) - \frac{\Upsilon_0+\eps^{-1}\gamma}{2\pi})}{\Theta( \Abel(z) + \Abel(\infty))} &
	\frac{ \alpha(z) + \alpha^{-1}}{2} 
	\frac{\Theta( \Abel(z) - \Abel(\infty) + \frac{\Upsilon_0+\eps^{-1}\gamma}{2\pi})}{\Theta( \Abel(z) - \Abel(\infty))} 
	\end{pmatrix}
$$
and $E(z)$ is the solution of the of the residual error RHP. 

The outer model is uniformly accurate except in any fixed neighborhood $\U_{\lam_s}$ of $\lam_s$. A local model must be inserted inside $\U_{\lam_s}$. At $\lam_s$ we have the usual critical behavior $\varphi(z) - \varphi(\lam_s) = \bigo{(z-\lam_s)^{3/2}}$, and the appropriate local model is the well-known Airy model. The details are standard and are omitted here. The important point is that the error introduced by the matching of local and outer models introduces an error bounded by $\bigo{\lp \frac{\eps}{t}\rp^{2/3}}$. Appealing to small norm theory the residual error $E(z)$ can be shown to exist and moreover $E(z) = I + \bigo{\lp \frac{\eps}{t}\rp^{2/3}}$ uniformly for all sufficiently large $t$ and small $\eps$.

Completing the expansion of $E(z)$, it follows that the solution $\psi(x,t)$ of \eqref{0.1}-\eqref{0.2} has the resulting expansion valid for each $x,t$ in the modulation zone:
\begin{multline} \label{modulation zone u}
	\psi(x,t) = \frac{\lam_1 - \lam_2 + \lam_3 - \lam_4}{2} \frac{ \Theta(0)}{\Theta( \frac{\Upsilon_0 + \gamma/\eps}{2\pi})} \frac{\Theta( 2\Abel(\infty)+\frac{\Upsilon_0 + \gamma/\eps}{2\pi})}{\Theta(2\Abel(\infty))}
e^{-2i (\chi_0 + \frac{\pi}{4} + (\theta(\lambda_+)+g(\infty))/\eps)}  \\ +\bigo{\lp \frac{\eps}{t}\rp^{2/3}} 
\end{multline} 

\subsubsection{Computing the leading order square modulus}

Recognizing that $\Upsilon_0$ and $\gamma$ are real, while $\Abel(\infty)$ is pure imaginary, write
$$
	w = \frac{ \Upsilon}{2} = \frac{\gamma}{2\eps} + \frac{\Upsilon_0}{2},
	\qquad
	i v = 2\pi  \Abel(\infty).
$$
Then in terms of these real variables we have
\begin{equation}\label{square mod 1}
	\begin{aligned}
	\rho(x,t) := | \psi(x,t) |^2 = 
		\lp \frac{\lam_1 - \lam_2 + \lam_3 - \lam_4}{2} \rp^2
		\frac{ \theta_3(0)^2 \theta_3(w+iv) \theta_3(w-iv)}{\theta_3^2(w) \theta_3^2(i v)} \\
		=\lp \frac{\lam_1 - \lam_2 + \lam_3 - \lam_4}{2} \rp^2
		\frac{ \theta_3(w)^2 \theta_3^2(iv) + \theta_1(w)^2 \theta_1(iv)^2}{\theta_3^2(w) \theta_3^2(i v)} \\
		=\lp \frac{\lam_1 - \lam_2 + \lam_3 - \lam_4}{2} \rp^2
		\lp 1 + \frac{ \theta_2^2(0) \theta_4^2(0)}{\theta_3^4(0)} 
		 \frac{ \theta_1^2(iv)}{\theta_3^2(iv)} \sd(w \theta_3^2, m)^2 \rp \\
		=\lp \frac{\lam_1 - \lam_2 + \lam_3 - \lam_4}{2} \rp^2
		\lb 1 - \sqrt{\frac{m}{1-m}} \frac{ \theta_1^2(iv)}{\theta_3^2(iv)}
		\cn(w \theta_3^2+K(m), m)^2 \rb \\
\end{aligned}
\end{equation}
To simplify the formula further we can evaluate the ratio of theta functions as follows. Write 
$$
	i v := 2\pi \int_{\lam_1}^{\infty_+} \nu = 2 x
$$
and make use of duplication formulae \cite{Law89} to write 
\begin{equation}\label{Abel constant}
	\frac{\theta_1^2(iv)}{\theta_3^2(iv)} 
	= \frac{ 4\theta_4^2(0)}{\theta_2^2(0)} \lp \frac{ \theta_1(x) \theta_2(x) \theta_3(x) \theta_4(x)}{ \theta_3^2(x) \theta_4^2(x) - \theta_1^2(x) \theta_2^2(x)} \rp^2  
	= 4\sqrt{ \frac{1-m}{m} } \frac{ T^2 }{ \lp T^2 - 1 \rp^2}, 
\end{equation}
where
$$
	 T= \frac{\theta_3(x) \theta_4(x)}{\theta_1(x) \theta_2(x)}.
$$
To compute $T$, consider the function $F(P)$ defined on the Riemann surface $\SSS_1$ by
\begin{equation}\label{FP1}
	F(P) := i e^{-i\pi \tau/2} \frac{ 
	\Theta \lp \int_{\lam_3}^P \nu + \frac{1}{2} + \frac{\tau}{2} \rp
	\Theta \lp \int_{\lam_4}^P \nu + \frac{1}{2} + \frac{\tau}{2} \rp
	}{
	\Theta \lp \int_{\lam_1}^P \nu + \frac{1}{2} + \frac{\tau}{2} \rp
	\Theta \lp \int_{\lam_2}^P \nu + \frac{1}{2} + \frac{\tau}{2} \rp
	}
	e^{-2\pi i \int_{\lam_1}^P \nu}.
\end{equation}
It follows from \eqref{Theta periods} that $F$ is single-valued on $\SSS_1$ and by construction $F$ has simple zeros at $\lam_3$ and $\lam_4$ and simple poles at $\lam_1$ and $\lam_2$. That is, $F$ is meromorphic on $\SSS_1$ and we have
\begin{equation}\label{FP2}
	\begin{aligned}
	F(P) &= \frac{\theta_3(0) \theta_4(0)}
		{\theta_2(0) \theta_1'(0) \frac{d\nu}{dP}(\lam_1)} 
		\frac{ (\lam_1 - \lam_2)^{1/2}}{(\lam_1 - \lam_3)^{1/2} (\lam_1 - \lam_4)^{1/2}} 
		\frac{ (z - \lam_3)^{1/2}(z-\lam_4)^{1/2}}{(z- \lam_2)^{1/2} (z - \lam_1)^{1/2}}  \\
		&= - \frac{ (\lam_1 - \lam_2)^{1/2} }{ (\lam_3 - \lam_4)^{1/2} }  \cdot
		\frac{ (z - \lam_3)^{1/2}(z-\lam_4)^{1/2}}{(z- \lam_2)^{1/2} (z - \lam_1)^{1/2}}
	\end{aligned}
\end{equation}
where the normalization comes from matching the residues at $\lam_1$, and $d\nu/dP(\lam_1)$ is computed in the local coordinate on $\SSS_1$ near $\lambda_1$. Computing $F(\infty_+)$ using both representations of $F$ gives:
\begin{equation*}
	\begin{aligned}
	F(\infty_+) &=
	 i e^{-i\pi \tau/2}  \frac{ 
	\Theta \lp \int_{\lam_1}^{\infty_+} \nu  \rp
	\Theta \lp \int_{\lam_1}^{\infty_+} \nu + \frac{1}{2} \rp
	}{
	\Theta \lp \int_{\lam_1}^{\infty_+} \nu + \frac{1}{2} + \frac{\tau}{2} \rp
	\Theta \lp \int_{\lam_1}^{\infty_+} \nu + \frac{\tau}{2} \rp
	}
	e^{-2\pi i \int_{\lam_1}^{\infty_+} \nu}
	= 	\frac{ \theta_3(x) \theta_4(x) }{ \theta_1(x) \theta_2(x)}  \\
	&= -\frac{ (\lam_1 - \lam_2)^{1/2} }{ (\lam_3 - \lam_4)^{1/2} }  
	\end{aligned}
\end{equation*}
Comparing the two values we see that 
$$
	T= \frac{\theta_3(x) \theta_4(x)}{\theta_1(x) \theta_2(x)} = - \frac{ (\lam_1 - \lam_2)^{1/2} }{ (\lam_3 - \lam_4)^{1/2} }. 
$$
Inserting this into \eqref{Abel constant} and simplifying \eqref{square mod 1} gives the formula
\begin{equation}\label{mod amp}
	\begin{aligned}
	\rho(x,t) &= a_2^2 +(a_3^2-a_2^2) 
	\cn^2 \lp \sqrt{a_1^2-a_3^2} \lp \frac{x- V t}{\eps} +  \phi \rp - K(m), m \rp \\
	 &= a_1^2 - (a_1^2-a_3^2) 
	\dn^2 \lp \sqrt{a_1^2-a_3^2} \lp \frac{x- V t}{\eps} +  \phi \rp - K(m), m \rp 
	\end{aligned}
\end{equation}
where
\begin{equation}\label{amp param}
\begin{aligned}
	a_1 &= -\frac{ \lam_1 + \lam_2 - \lam_3 - \lam_4}{2} \qquad
	a_2 = -\frac{ \lam_1 - \lam_2 + \lam_3 - \lam_4}{2} \\
	a_3 &= -\frac{ \lam_1 - \lam_2 - \lam_3 + \lam_4}{2} \qquad
	V = - \frac{ \lam_1 + \lam_2 + \lam_3 + \lam_4 }{2} =  -\frac{1}{2} e_1(\vect \lam) \\
	\phi &= \frac{1}{2\pi} \int_{-\infty}^{\lam_s} \frac{ \lp z+V \rp 
	\lp \log( 1 - |r(z)|^2) \one_{(-\infty, -1)} + \log r_+(z) \one_{(-1, \lam_s)}\rp}
	{\RR_+(z; \vect \lam)} dz
\end{aligned}
\end{equation}

\subsubsection{computing the leading order phase}
Using \eqref{modulation zone u}, the leading order phase contribution is given by
$$
	\arg \psi(x,t) = -2(\chi_0 + \pi/4) - \frac{2}{\eps} \lp \theta(\lambda_+) + g(\infty) \rp + \arg \theta_3(w+ i v).
$$
These terms can be evaluated explicitly in terms of elliptic integrals \cite{BF71}
\begin{align*}
	-2(\chi_0 + \pi/4) &= -2 \phi \lp \lam_1 - \int_{\lam_1}^{\infty_+} (\omega^{(0)} - dz) \rp 
	= 2\phi \left[ V + \eta \right],  \\
	-2(\theta(\lambda_+) + g(\infty)) &=
		 \theta(\lam_1) - 2t\int_{\lam_1}^{\infty_+} (\omega^{(1)} - zdz) -x \int_{\lam_1}^{\infty_+} (\omega^{(0)} - dz) \\
		&= 2t \left[ \Gamma_2 - V^2 - V \eta \right] + x \left[ V+ \eta \right].
\end{align*}
Putting it all together we have
\begin{multline}\label{mod phase}
	\eps \arg \psi(x,t) = 2t  \lp\Gamma_2 -V^2-V\eta\rp  + 2x \lp V+\eta \rp  \\ 
		+ \eps \arg \left\{ \theta_3 \lb 
		\frac{\pi}{2K(m)}\sqrt{\frac{a}{m}} \lp \frac{x-V t}{\eps} + \phi \rp - i \pi \frac{F(\varphi,1-m)}{K(m)} \rb 			\right\}  + 2\eps\phi(V+\eta),
\end{multline}
where
\begin{gather*}
	\Gamma_2 = \sum_{ j >k \geq 1}^4 \lam_j \lam_k, \\
	\eta = \lam_1 - (\lam_1-\lam_4) \jacobiZ(n,m), \\ 
	n = - \lp \frac{\lam_3-\lam_4}{\lam_1-\lam_4}\rp,  \qquad 
\varphi = \arcsin \sqrt{ \frac{ \lam_2 - \lam_4}{\lam_1-\lam_4} }; 
\end{gather*}
$a, m, V, \phi$ are as in \eqref{amp param}; and $F(\varphi,m)$ and $\jacobiZ(n,m)$ are the elliptic integral of the first kind and the Jacobi zeta function respectively.


\subsubsection{computing the fluid velocity}
Using formula \eqref{modulation zone u} for the leading order behavior of $\psi$ in the modulation zone, 
the velocity $u$ defined by the hydrodynamic change of variables for NLS \eqref{fluid vars} becomes 
\begin{equation}\label{5.85_1}
	u(x,t) = \eps \imag \lb \partial_x \log (\psi(x,t)) \rb = 
	\frac{\gamma_x}{2} \imag \lb \frac{ \theta_3'(w+iv)}{\theta_3(w+iv)} \rb  + 
	2 \int_{\lam_+}^\infty (\omega^{(0)} -d\lambda ) - 2\lambda_+ + \bigo{\eps}
\end{equation}
The first term can be simplified at follows
\begin{align*}
	\frac{\gamma_x}{2} \imag \lb \frac{ \theta_3'(w+iv)}{\theta_3(w+iv)} \rb &= 
	\frac{\gamma_x}{4i} \lb \frac{ \theta_3'(w+iv)}{\theta_3(w+iv)} -  \frac{ \theta_3'(w-iv)}{\theta_3(w-iv)}	\rb \\
\intertext{taking the log derivative of $1.4.25$ in \cite{Law89} and using \eqref{square mod 1} this becomes}
	\frac{\gamma_x}{2} \imag \lb \frac{ \theta_3'(w+iv)}{\theta_3(w+iv)} \rb
	&= \frac{\gamma_x}{2i} \lb  \frac{\theta_1'(iv)}{\theta_1(iv)} - 
	\frac{a_2^2}{\rho} \od{}{\zeta} \log \lp \frac{\theta_3(\zeta)}{\theta_1(\zeta)} \rp \Bigg|_{\zeta=iv}   \rb
	= 2\pi c_\nu \frac{\theta_1'(iv)}{\theta_1(iv)} + \frac{a_1 a_2 a_3}{\rho}
\end{align*}
where in the last step the logarithmic derivative is evaluated using \eqref{Abel constant}-\eqref{FP2} and we use that fact that $\gamma_x = \oint_b \omega^{(0)} = 4\pi i c_\nu$. Inserting this into \eqref{5.85_1} we have
\begin{equation}\label{u1}
	u(x,t) = \frac{a_1 a_2 a_3}{\rho} + 
	2 \pi c_\nu \frac{\theta_1'(iv)}{\theta_1(iv)} 
	+ 2 \int_{\lam_+}^\infty (\omega^{(0)} -d\lambda ) - 2\lambda_+
	 + \bigo{\eps}
\end{equation}	

\begin{prop} 
$$
 	2 \pi c_\nu \frac{\theta_1'(iv)}{\theta_1(iv)} 
	+ 2 \int_{\lam_+}^\infty (\omega^{(0)} -d\lambda ) - 2\lambda_+ =  -\frac{1}{2} \sum_{k=1}^4 \lam_k := V
$$
\end{prop}

\begin{proof}
The function
\begin{equation*}
	G(P) =  2\pi c_\nu \frac{\theta_1' \lp2\pi \int_{\lam_+}^P \nu \rp}{\theta_1 \lp2\pi \int_{\lam_+}^P \nu \rp} + 2\int_{\lam_+}^P (\omega^{(0)} - d\lambda) 
	= \sum_{k=1}^4 \frac{\theta_k' \lp \pi \int_{\lam_+}^P \nu \rp}{\theta_k \lp \pi \int_{\lam_+}^P \nu \rp} + 2\int_{\lam_+}^P (\omega^{(0)} - d\lambda)
\end{equation*}
is a single valued on the Riemann surface $\SSS_1$ and by definition $\lim_{P \to \infty_+} G(P) = K$. The single-valuedness follows from the relations $d\log \theta_k(x+n\pi + m \pi \tau) = d\log \theta_k(x) - 2im$, $\oint_a (\omega^{(0)} - d\lambda) = 0$ and $\oint_b (\omega^{(0)} - d\lambda) = -2\pi i c_\nu$. From the second representation for $G(P)$ above it is clear that $G$ is meromorphic over $\SSS_1$ with five simple poles at $\lam_k,\ k=1,\dots, 4,$ and $\infty_-$, with residues
$$
	\res_{P=\lam_k} G(P) = \frac{1}{2} \prod_{j \neq k} (\lam_k - \lam_j)^{1/2}, \qquad
	\res_{P=\infty_2} G(P) = -4.
$$
The function 
$$
	\widetilde G(P ) = \od{R}{\lam} (\lam(P)) - 2(\lam(P) - \lam_+)
$$ 
is also meromorphic on $\SSS_1$ with the same poles and residues, so that the difference $G(P) - \widetilde G(P)$ is constant. However, expanding the difference at $\lam_+$ we find that $G(P) - \widetilde G(P) = \bigo{ (\lam(P) - \lam_+)^{1/2} }$ so $G(P) = \widetilde G(P)$. The result follows from observing that $\lim_{P \to \infty_+} \widetilde G(P) = V$.
\end{proof}

It immediately follows from the proposition that 
\begin{equation}\label{u2}
	u(x,t) = \frac{a_1 a_2 a_3}{\rho} + V,
\end{equation}
which is in perfect agreement with the Whitham theory for the genus one self-similar solutions of NLS \eqref{one phase solution}.

\subsection{The far right field: $\tau > -\frac{1}{2} \lp \lambda_+  + \lambda_- - 2 \rp 
	+ \frac{2(1+\lambda_-)(1+\lambda_+)} {\lambda_+ + \lambda_- + 2} $}
	
When the moving branch point of the two cut $g$-function collides with $-1$, the left cut closes and what remains is a one-cut $g$-function on the interval $(\lambda_-, \lambda_+)$, as one would expect for the far right field. 
The right field $g$-function is given by
\begin{equation}\label{3.20}
	g(z) = \int_{\lambda_+}^z d\theta - 2t \frac{ (\lambda-\xi_-)(\lambda-\xi_+)}{\RR(\lambda; \lambda_-,\lambda_+)} d\lambda 
	= \theta \big|^z_{\lambda_+} - t \RR(\lambda; \lambda_-,\lambda_+)  (\lambda - \xi_0)
\end{equation}
where
\begin{equation}
	\begin{gathered}
	\xi_0 = -\frac{1}{2} \lp \lambda_+ + \lambda_- + 2\tau \rp \\
	\xi_\pm = \frac{1}{4} \lp \lambda_+ + \lambda_- - \tau \pm 
	\sqrt{ (\lambda_+ + \lambda_- + \tau)^2 + 2(\lambda_+ - \lambda_-)^2} \rp
	\end{gathered}
\end{equation}

In order for the two-cut $g$-function to degenerate continuously into this equation we need $\xi_-(\tau) = -1$. This is exactly the condition which bounds the far-right field:
$$
	\xi_-(\tau) < -1 
	\quad \Longleftrightarrow \quad
	\tau > -\frac{1}{2} \lp \lambda_+  + \lambda_- - 2 \rp 
	+ \frac{2(1+\lambda_-)(1+\lambda_+)} {2+\lambda_- + \lambda_+}.
$$
As such the modified phase function
\begin{equation}
	\begin{aligned}
	\varphi(z) 
	&= 2t \int_{\xi_-}^z  \frac{ (\lambda-\xi_-)(\lambda-\xi_+)}{\RR(\lambda; \lambda_-,\lambda_+)} d\lambda \\
	&= t \, \RR(z;\lambda_-, \lambda_+)(z + \frac{\lambda_- + \lambda_+}{2} + \tau)
	\end{aligned}
\end{equation}
has an imaginary sign table resembling Figure~\ref{fig:1}a. We open contours $\Gamma_i, \, i=1,2$ from $\xi_-(\tau)$ which divide $\C^+$ into three sectors $\Omega_i, \, i=1,2,3$ as shown in Figure~\ref{fig:case1fig5}.

\begin{figure}[t]
	\begin{center}
		\includegraphics[width=.75\textwidth]{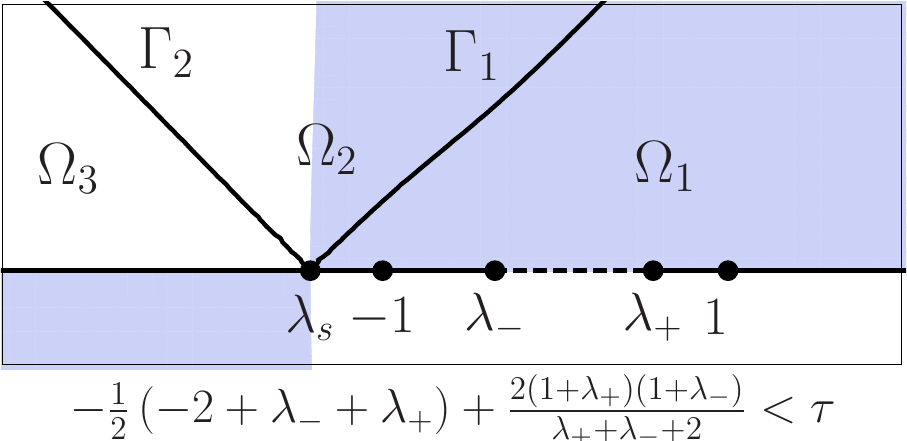}
		\caption{
		The contours $\Gamma_i$ and regions $\Omega_i$ used to define the 
		map $M \mapsto N$ (c.f. \eqref{Lenses}) for $\tau =x/t$ in the right planar zone 
		(defined above). The stationary point $\xi_-(\tau)$ is a decreasing function of 
		$\tau$; the lower boundary of the zone is characterized by the collision $\xi_-(\tau) = -1$.
		Blue regions correspond to $\imag \varphi >0$ and white regions $\imag \varphi < 0$.
		\label{fig:case1fig5}
		}
	\end{center}	
\end{figure}	

The result of \eqref{gtrans}-\eqref{Lenses} using \eqref{3.20} is the following RHP for $N(z)$:

\begin{rhp}\label{rhp:N right} 
Find a $2\times2$ matrix-valued function $N$ with the following properties
\begin{enumerate}[1.]
	\item $N(z)$ is analytic in $\C \backslash \Gamma_N$, 
	$\Gamma_N = (-\infty, \xi_-(\tau)) \cup (\lambda_-, \lambda_+) \bigcup_{i=1}^2 (\Gamma_i \cup \Gamma_i^*)$.
	\item $N(z) = I + \bigo{z^{-1} }$ as $z\to \infty$.
	\item $N(z)$ takes continuous boundary values on $\Gamma_N$ away from points of self intersection and branch points which satisfy the jump relation $N_+(z) = N_-(z) V_N(z)$ where
	\begin{align}\label{3.15}
		V_N(z) = 
		\begin{cases}
			(1 - r(z) r^*(z) )^{\sig} & z \in (-\infty, \xi_-(\tau))  \\
			\offdiag{ -e^{-2i \theta(\lambda_+)/\eps} }{ e^{2i \theta(\lambda_+)/\eps} }
				& z \in (\lambda_-, \lambda_+) \\
			\tril{r(z) e^{2i (\varphi(z) + \theta(\lambda_+))/\eps}} & z \in \Gamma_1 \smallskip \\
			\triu{ \frac{ -r^*(z)}{1 - r(z) r^*(z) } e^{-2i( \varphi(z) + \theta(\lambda_+))/\eps} } 
			& z \in \Gamma_2
		\end{cases}
	\end{align}
	\item $N(z)$ is bounded except at the points $\lambda_+$ and $\lambda_- $ where
	\begin{equation}
		N(z) = \bigo{ \begin{matrix} 
				(z-p)^{-1/4} & (z-p)^{-1/4} \\ 
				(z-p)^{-1/4} & (z-p)^{-1/4} 
				\end{matrix} },	\quad  p \in  \{\lambda_+, \lambda_- \}.
	\end{equation}	
\end{enumerate}
\end{rhp}

\subsubsection{Constructing a parametrix for the far right field}

Clearly, the jump matrices along $\Gamma_i$, $i=1,2$ are near identity at any fixed distance from $\xi_-$. The remaining jumps on the real axis can be dealt with as before. In fact, comparing RHP~\ref{rhp:N right} to RHP~\ref{rhp:N left} we see that the problems in the far right field is a simpler version of that for the left field, the twist jumps along $(-1,1)$ are exchanged for a simpler twist along $(\lam_-, \lam_+)$, and the diagonal jump $(1-|r(z)|^2)^\sig$ lies only on $(-\infty, \xi_-(\tau))$ which is separated from the twist. As such the parametrix is constructed in the same way, but with less effort needed to construct the scalar function $D(z)$. 

Define
\begin{equation}
	D(z) = \exp \lb  \frac{ i \theta(\lambda_+)}{\eps} + 
	\frac{\RR(z; \lambda_+, \lambda_-)}{2\pi i} \int_{-\infty}^{\xi_-(\tau)} 
	\frac{\log \lp 1 - |r(\lambda)|^2 \rp}{\RR(\lambda; \lambda_+, \lambda_-)} 
	\frac{d \lambda}{\lambda -z} \rb.
\end{equation}
so that 
\begin{equation}
	D(\infty) = \exp \lb  \frac{ i \theta(\lambda_+)}{\eps} - 
	\frac{1}{2\pi i} \int_{-\infty}^{\xi_-(\tau)} 
	\frac{\log \lp 1 - |r(\lambda)|^2 \rp}{\RR(\lambda; \lambda_+, \lambda_-)} 
	d\lambda \rb
\end{equation}

Then introducing a fixed neighborhood $\U_{\xi_-}$ of the stationary point $\xi_-(\tau)$ which remains bounded away from $-1$, we can write the solution $N(z)$ to RHP~\ref{rhp:N right} in the form
\begin{equation}
	N(z) = \begin{cases} 
		D(\infty)^{-\sig} E(z) \Ecal(z;\lam_+, \lam_-) D(z)^\sig 
		& z \in \C \backslash \U_{\xi_-(\tau)} \\
		D(\infty)^{-\sig} E(z) \Psi(z) D(z)^\sig 
		& z \in \U_{\xi_-} 
		\end{cases}
\end{equation}	
where $\Ecal(z;\lam_+,\lam_-)$, defined by \eqref{0.4b}, is related to the plane wave solution of the Lax-Pair \eqref{0.3} for a constant plane wave with Riemann invariants $\lam_\pm$. The need for a different model $\Psi$ in the neighborhood $\U_{\xi_-}$ of $\xi_-$ is that inside this neighborhood the jump matrices of $N(z)$ along each $\Gamma_i$ cannot be uniformly approximated by identity. Nevertheless, a local model can be constructed which exactly matches the jump matrices along each $\Gamma_i$, at the cost of introducing a matching error on the boundary $\partial \U_{\xi_-}$ of the neighborhood. The construction of this model from parabolic cylinder functions is standard in the Riemann-Hilbert literature, see \cite{JM11} for details of its construction.
The important result is jump along $\partial \U_{\xi_-}$ is $\bigo{ \sqrt{\frac{\eps}{t}} \log \frac{\eps}{t}}$. As such the Riemann-Hilbert problem for the resulting error matrix $E(z)$ is in the small norm class and using standard estimates one can show that $E(z) = I + \bigo{ \sqrt{\frac{\eps}{t}} \log \frac{\eps}{t}}$, moreover $E(z)$ admits an asymptotic expansion whose terms, given sufficient effort, can be explicitly computed.

The resulting behavior of the solution of \eqref{0.1}-\eqref{0.2} for $\tau > -\frac{1}{2} \lp \lambda_+  + \lambda_- - 2 \rp + \frac{2(1+\lambda_-)(1+\lambda_+)} {\lambda_+ + \lambda_- + 2} $ is given by 
\begin{equation}
	\begin{gathered}
	\psi(x,t) = \sqrt{\rho}  e^{i (k x - \omega t)/\eps}  e^{ - i \chi(x/t) } + \bigo{ \sqrt{\frac{\eps}{t}} \log \frac{\eps}{t}  } \\
	 \rho = \lp \frac{\lambda_+ - \lambda_-}{2} \rp^2, 
	 \qquad k = -(\lambda_+ + \lambda_-), 
	 \qquad  
	 \omega = \frac{1}{2} k^2 + \rho \\
	\chi(\tau) = 
	\frac{1}{\pi}  \int_{-\infty}^{\xi_-(\tau)}  \frac{\log(1 -|r(z)|^2 )}{\sqrt{(\lambda_+-z)(\lambda_- - z)}} dz 
	\end{gathered}
\end{equation}

\subsection*{Acknowledgments} The author was partially supported by the European Research Council Advanced Grant FroM-PDE, by PRIN 2010-11 Grant ``Geometric and analytic theory of Hamiltonian systems in finite and infinite dimensions" of Italian Ministry of Universities and Researches and by the FP7 IRSES grant RIMMP ``Random and Integrable Models in Mathematical Physics". The author would also like to thank Prof. Tamara Grava for comments which undoubtably improved the manuscript.

\bibliographystyle{abbrv}
\bibliography{bibliography}

\def\cprime{$'$} \def\cprime{$'$} \def\cprime{$'$} \def\cprime{$'$}
  \def\cprime{$'$}
\begin{thebibliography}{10}

\bibitem{Ania06}
J.~D. Ania-Casta\~n\'on, T.~J. Ellingham, R.~Ibbotson, X.~Chen, L.~Zhang, and
  S.~K. Turitsyn.
\newblock Ultralong raman fiber lasers as virtually lossless optical media.
\newblock {\em Phys. Rev. Lett.}, 96:023902, Jan 2006.

\bibitem{BT10}
M.~Bertola and A.~Tovbis.
\newblock Universality in the profile of the semiclassical limit solutions to
  the focusing nonlinear {S}chr\"odinger equation at the first breaking curve.
\newblock {\em Int. Math. Res. Not. IMRN}, (11):2119--2167, 2010.

\bibitem{BK06}
G.~Biondini and Y.~Kodama.
\newblock On the {W}hitham equations for the defocusing nonlinear
  {S}chr\"odinger equation with step initial data.
\newblock {\em J. Nonlinear Sci.}, 16(5):435--481, 2006.

\bibitem{BP82}
M.~Boiti and F.~Pempinelli.
\newblock The spectral transform for the {NLS} equation with left-right
  asymmetric boundary conditions.
\newblock {\em Nuovo Cimento B (11)}, 69(2):213--227, 1982.

\bibitem{BE00}
A.~Boutet~de Monvel and I.~Egorova.
\newblock The {T}oda lattice with step-like initial data. {S}oliton
  asymptotics.
\newblock {\em Inverse Problems}, 16(4):955--977, 2000.

\bibitem{BKS11}
A.~Boutet~de Monvel, V.~P. Kotlyarov, and D.~Shepelsky.
\newblock Focusing {NLS} equation: long-time dynamics of step-like initial
  data.
\newblock {\em Int. Math. Res. Not. IMRN}, (7):1613--1653, 2011.

\bibitem{BV07}
R.~Buckingham and S.~Venakides.
\newblock Long-time asymptotics of the nonlinear {S}chr\"odinger equation shock
  problem.
\newblock {\em Comm. Pure Appl. Math.}, 60(9):1349--1414, 2007.

\bibitem{BF71}
P.~F. Byrd and M.~D. Friedman.
\newblock {\em Handbook of elliptic integrals for engineers and scientists}.
\newblock Die Grundlehren der mathematischen Wissenschaften, Band 67.
  Springer-Verlag, New York, 1971.
\newblock Second edition, revised.

\bibitem{CG10}
T.~Claeys and T.~Grava.
\newblock Painlev\'e {II} asymptotics near the leading edge of the oscillatory
  zone for the {K}orteweg-de {V}ries equation in the small-dispersion limit.
\newblock {\em Comm. Pure Appl. Math.}, 63(2):203--232, 2010.

\bibitem{CK85}
A.~Cohen and T.~Kappeler.
\newblock Scattering and inverse scattering for steplike potentials in the
  {S}chr\"odinger equation.
\newblock {\em Indiana Univ. Math. J.}, 34(1):127--180, 1985.

\bibitem{DKMVZ}
P.~Deift, T.~Kriecherbauer, K.~T.-R. McLaughlin, S.~Venakides, and X.~Zhou.
\newblock Uniform asymptotics for polynomials orthogonal with respect to
  varying exponential weights and applications to universality questions in
  random matrix theory.
\newblock {\em Comm. Pure Appl. Math.}, 52(11):1335--1425, 1999.

\bibitem{DVZ98}
P.~Deift, S.~Venakides, and X.~Zhou.
\newblock An extension of the steepest descent method for {R}iemann-{H}ilbert
  problems: the small dispersion limit of the {K}orteweg-de {V}ries ({K}d{V})
  equation.
\newblock {\em Proc. Natl. Acad. Sci. USA}, 95(2):450--454 (electronic), 1998.

\bibitem{DZ93}
P.~Deift and X.~Zhou.
\newblock A steepest descent method for oscillatory {R}iemann-{H}ilbert
  problems. {A}symptotics for the {MK}d{V} equation.
\newblock {\em Ann. of Math. (2)}, 137(2):295--368, 1993.

\bibitem{DZ94}
P.~Deift and X.~Zhou.
\newblock Long-time behavior of the non-focusing nonlinear schr\"odinger
  equation: A case study.
\newblock {\em New Series: Lectures in Mathematical Sciences}, (5), 1994.

\bibitem{DP13}
F.~Demontis, B.~Prinari, C.~van~der Mee, and F.~Vitale.
\newblock The inverse scattering transform for the defocusing nonlinear
  {S}chr\"odinger equations with nonzero boundary conditions.
\newblock {\em Stud. Appl. Math.}, 131(1):1--40, 2013.

\bibitem{DiFM05}
J.~C. DiFranco and K.~T.-R. McLaughlin.
\newblock A nonlinear {G}ibbs-type phenomenon for the defocusing nonlinear
  {S}chr\"odinger equation.
\newblock {\em IMRP Int. Math. Res. Pap.}, (8):403--459, 2005.

\bibitem{EGK13}
I.~Egorova, Z.~Gladka, V.~Kotlyarov, and G.~Teschl.
\newblock Long-time asymptotics for the {K}orteweg--de {V}ries equation with
  step-like initial data.
\newblock {\em Nonlinearity}, 26(7):1839--1864, 2013.

\bibitem{El}
G.~A. {\`E}l{\cprime}, V.~V. Geogjaev, A.~V. Gurevich, and A.~L. Krylov.
\newblock Decay of an initial discontinuity in the defocusing {NLS}
  hydrodynamics.
\newblock {\em Phys. D}, 87(1-4):186--192, 1995.
\newblock The nonlinear Schr{\"o}dinger equation (Chernogolovka, 1994).

\bibitem{FT07}
L.~D. Faddeev and L.~A. Takhtajan.
\newblock {\em Hamiltonian methods in the theory of solitons}.
\newblock Classics in Mathematics. Springer, Berlin, english edition, 2007.
\newblock Translated from the 1986 Russian original by Alexey G. Reyman.

\bibitem{FFM}
H.~Flaschka, M.~G. Forest, and D.~W. McLaughlin.
\newblock Multiphase averaging and the inverse spectral solution of the
  {K}orteweg-de {V}ries equation.
\newblock {\em Comm. Pure Appl. Math.}, 33(6):739--784, 1980.

\bibitem{FL}
M.~G. Forest and J.~E. Lee.
\newblock Geometry and modulation theory for the periodic nonlinear
  {S}chr\"odinger equation.
\newblock In {\em Oscillation theory, computation, and methods of compensated
  compactness ({M}inneapolis, {M}inn., 1985)}, volume~2 of {\em IMA Vol. Math.
  Appl.}, pages 35--69. Springer, New York, 1986.

\bibitem{GT02}
T.~Grava and F.-R. Tian.
\newblock The generation, propagation, and extinction of multiphases in the
  {K}d{V} zero-dispersion limit.
\newblock {\em Comm. Pure Appl. Math.}, 55(12):1569--1639, 2002.

\bibitem{GKE92}
A.~V. Gurevich, A.~L. Krylov, and G.~A. {\`E}l{\cprime}.
\newblock Evolution of a {R}iemann wave in dispersive hydrodynamics.
\newblock {\em Zh. \`Eksper. Teoret. Fiz.}, 101(6):1797--1807, 1992.

\bibitem{GP87}
A.~V. Gurevich and L.~P. Pitaevski{\u\i}.
\newblock Averaged description of waves in the {K}orteweg-de {V}ries-{B}urgers
  equation.
\newblock {\em Zh. \`Eksper. Teoret. Fiz.}, 93(3):871--880, 1987.

\bibitem{AH07}
M.~A. Hoefer and M.~J. Ablowitz.
\newblock Interactions of dispersive shock waves.
\newblock {\em Phys. D}, 236(1):44--64, 2007.

\bibitem{Hoeffer06}
M.~A. Hoefer, M.~J. Ablowitz, I.~Coddington, E.~A. Cornell, P.~Engels, and
  V.~Schweikhard.
\newblock Dispersive and classical shock waves in bose-einstein condensates and
  gas dynamics.
\newblock {\em Phys. Rev. A}, 74:023623, Aug 2006.

\bibitem{Its2}
A.~R. Its and A.~F. Ustinov.
\newblock Time asymptotics of the solution of the {C}auchy problem for the
  nonlinear {S}chr\"odinger equation with boundary conditions of finite density
  type.
\newblock {\em Dokl. Akad. Nauk SSSR}, 291(1):91--95, 1986.

\bibitem{Its1}
A.~R. Its and A.~F. Ustinov.
\newblock Formulation of the scattering theory for the {NLS} equation with
  boundary conditions of finite density type in a soliton-free sector.
\newblock {\em Zap. Nauchn. Sem. Leningrad. Otdel. Mat. Inst. Steklov. (LOMI)},
  169(Voprosy Kvant. Teor. Polya i Statist. Fiz. 8):60--67, 186--187, 1988.

\bibitem{JM11}
R.~{Jenkins} and K.~D.~T.~. {McLaughlin}.
\newblock {The semiclassical limit of focusing NLS for a family of non-analytic
  initial data}.
\newblock {\em CPAM}, to appear.

\bibitem{JLM}
S.~Jin, C.~D. Levermore, and D.~W. McLaughlin.
\newblock The behavior of solutions of the {NLS} equation in the semiclassical
  limit.
\newblock In {\em Singular limits of dispersive waves ({L}yon, 1991)}, volume
  320 of {\em NATO Adv. Sci. Inst. Ser. B Phys.}, pages 235--255. Plenum, New
  York, 1994.

\bibitem{KMM03}
S.~Kamvissis, K.~D. T.-R. McLaughlin, and P.~D. Miller.
\newblock {\em Semiclassical soliton ensembles for the focusing nonlinear
  {S}chr\"odinger equation}, volume 154 of {\em Annals of Mathematics Studies}.
\newblock Princeton University Press, Princeton, NJ, 2003.

\bibitem{Kodama99}
Y.~Kodama.
\newblock The {W}hitham equations for optical communications: mathematical
  theory of {NRZ}.
\newblock {\em SIAM J. Appl. Math.}, 59(6):2162--2192, 1999.

\bibitem{KA12}
V.~Kotlyarov and A.~Minakov.
\newblock Riemann-{H}ilbert problems and the m{K}d{V} equation with step
  initial data: short-time behavior of solutions and the nonlinear {G}ibbs-type
  phenomenon.
\newblock {\em J. Phys. A}, 45(32):325201, 17, 2012.

\bibitem{Lake77}
B.~M. Lake, H.~C. Yuen, H.~Rungaldier, and W.~E. Ferguson.
\newblock Nonlinear deep-water waves: theory and experiment. part 2. evolution
  of a continuous wave train.
\newblock {\em Journal of Fluid Mechanics}, 83:49--74, 11 1977.

\bibitem{Law89}
D.~F. Lawden.
\newblock {\em Elliptic functions and applications}, volume~80 of {\em Applied
  Mathematical Sciences}.
\newblock Springer-Verlag, New York, 1989.

\bibitem{LL83}
P.~D. Lax and C.~D. Levermore.
\newblock The small dispersion limit of the {K}orteweg-de {V}ries equation. {I,
  II, II}.
\newblock {\em Comm. Pure Appl. Math.}, 36((3,5,6)):(253--290, 571--593,
  809--829), 1983.

\bibitem{Madelung}
E.~Madelung.
\newblock Quantentheorie in hydrodynamischer form.
\newblock {\em Zeitschrift für Physik}, 40(3-4):322--326, 1927.

\bibitem{Musk}
N.~I. Muskhelishvili.
\newblock {\em Singular integral equations}.
\newblock Dover Publications Inc., New York, 1992.
\newblock Boundary problems of function theory and their application to
  mathematical physics, Translated from the second (1946) Russian edition and
  with a preface by J. R. M. Radok, Corrected reprint of the 1953 English
  translation.

\bibitem{Taylor70}
R.~Taylor, D.~Baker, and H.~Ikezi.
\newblock {Observation of colisionless electrostatic shocks}.
\newblock {\em {Physical Review Letters}}, {24}({5}):{206--\&}, {1970}.

\bibitem{TVZ04}
A.~Tovbis, S.~Venakides, and X.~Zhou.
\newblock On semiclassical (zero dispersion limit) solutions of the focusing
  nonlinear {S}chr\"odinger equation.
\newblock {\em Comm. Pure Appl. Math.}, 57(7):877--985, 2004.

\bibitem{Vartanian2}
A.~H. Vartanian.
\newblock Long-time asymptotics of solutions to the {C}auchy problem for the
  defocusing nonlinear {S}chr\"odinger equation with finite-density initial
  data. {II}. {D}ark solitons on continua.
\newblock {\em Math. Phys. Anal. Geom.}, 5(4):319--413, 2002.

\bibitem{Vartanian1}
A.~H. Vartanian.
\newblock Long-time asymptotics of solutions to the {C}auchy problem for the
  defocusing non-linear {S}chr\"odinger equation with finite-density initial
  data. {I}.\ {S}olitonless sector.
\newblock In {\em Recent developments in integrable systems and
  {R}iemann-{H}ilbert problems ({B}irmingham, {AL}, 2000)}, volume 326 of {\em
  Contemp. Math.}, pages 91--185. Amer. Math. Soc., Providence, RI, 2003.

\bibitem{Venakides90}
S.~Venakides.
\newblock The {K}orteweg-de {V}ries equation with small dispersion: higher
  order {L}ax-{L}evermore theory.
\newblock {\em Comm. Pure Appl. Math.}, 43(3):335--361, 1990.

\bibitem{Whit}
G.~B. Whitham.
\newblock Non-linear dispersive waves.
\newblock {\em Proc. Roy. Soc. Ser. A}, 283:238--261, 1965.

\bibitem{Whitham}
G.~B. Whitham.
\newblock {\em Linear and nonlinear waves}.
\newblock Pure and Applied Mathematics (New York). John Wiley \& Sons Inc., New
  York, 1999.
\newblock Reprint of the 1974 original, A Wiley-Interscience Publication.

\bibitem{CX13}
J.~Xu, E.~Fan, and Y.~Chen.
\newblock Long-time asymptotic for the derivative nonlinear schrödinger
  equation with step-like initial value.
\newblock {\em Mathematical Physics, Analysis and Geometry}, 16(3):253--288,
  2013.

\bibitem{Zak}
V.~E. {Zakharov}.
\newblock {Collapse of Langmuir Waves}.
\newblock {\em Soviet Journal of Experimental and Theoretical Physics},
  35:908--+, 1972.

\bibitem{ZS71}
V.~E. Zakharov and A.~B. Shabat.
\newblock Exact theory of two-dimensional self-focusing and one-dimensional
  self-modulation of waves in nonlinear media.
\newblock {\em \v Z. \`Eksper. Teoret. Fiz.}, 61(1):118--134, 1971.

\end{thebibliography}

\end{document}